\documentclass[11pt]{amsart}
\usepackage[margin=1in]{geometry}
\usepackage{latexsym}

\usepackage{libertine}
\usepackage[utf8]{inputenc}
\usepackage{amsfonts}
\usepackage{amsmath}
\usepackage{amssymb}
\usepackage{amsthm}
\usepackage{commath}
\usepackage{xcolor}
\usepackage{subfig}
\usepackage{float}
\usepackage{mathrsfs}
\usepackage{appendix}
\usepackage{needspace}
\usepackage{bbm}
\usepackage[colorlinks, breaklinks=true]{hyperref}
\usepackage[noabbrev,capitalise,nameinlink]{cleveref}
\hypersetup{
	linkcolor=blue,
	citecolor=Green,
	filecolor=Red,
	urlcolor=NavyBlue,
	menucolor=Red,
	runcolor=Red}
\numberwithin{equation}{section}


\usepackage{thmtools}
\usepackage{thm-restate}
\numberwithin{equation}{section}
\allowdisplaybreaks

\crefname{appendix}{Appendix}{Appendices}
\Crefname{appendix}{Appendix}{Appendices}
\crefname{subappendix}{Appendix}{Appendices}
\Crefname{subappendix}{Appendix}{Appendices}

\crefformat{appendix}{#2Appendix~#1#3}
\Crefformat{appendix}{#2Appendix~#1#3}

\crefformat{subappendix}{#2Appendix~#1#3}
\Crefformat{subappendix}{#2Appendix~#1#3}

\AddToHook{cmd/appendix/before}{%
  \crefalias{section}{appendix}%
  \crefalias{subsection}{subappendix}%
}
\makeatletter
\def\@fnsymbol#1{%
  \ensuremath{\ifcase#1\or *\or \circ\or \dagger\or \ddagger\or
  \mathsection\or \mathparagraph\else\@ctrerr\fi}}
\DeclareRobustCommand{\authmark}[1]{\textsuperscript{\@fnsymbol{#1}}}
\makeatother

\usepackage[dvipsnames]{xcolor}

\usepackage[style=alphabetic, maxbibnames=9, giveninits]{biblatex}
\renewbibmacro{in:}{}

\usepackage{todonotes}
\usepackage{amsmath}
\DeclareMathOperator*{\argmax}{arg\,max}

\newcommand{\brac}[1]{\left \langle #1 \right \rangle}

\newcommand{\1}{\mathbbm{1}}
\newcommand{\calS}{\mathcal{S}}
\newcommand{\R}{\mathbb{R}}
\newcommand{\E}{\mathbb{E}}
\newcommand{\calD}{\mathcal{D}}

\newcommand{\Var}{\textup{Var}}

\newcommand{\Geom}{\textup{Geom}}
\newcommand{\supp}{\textup{supp}}

\newcommand{\KL}{D_{\mathrm{KL}}}
\newcommand{\Prob}{\mathbb{P}}
\newcommand{\Ber}{\textup{Ber}}
\newcommand{\Bin}{\textup{Bin}}
\newcommand{\unif}{\textup{unif}}
\newcommand{\Unif}{\textup{Unif}}

\newcommand{\doubleline}{\;\middle\|\;}

\newcommand{\maj}{\textup{maj}}

\newcommand{\Z}{\mathbb{Z}}
\newcommand{\N}{\mathbb{N}}

\newcommand{\sign}{\textup{sign}}
\newcommand{\BDC}{\textup{\textsf{BDC}}}

\newcommand{\ind}{\textup{\textsf{ind}}}
\newcommand{\std}{\textup{\textsf{std}}}
\newcommand{\loc}{\textup{\textsf{loc}}}
\newcommand{\exc}{\textup{exc}}

\newcommand{\fpl}{f_{\mathrm{pl}}}
\newcommand{\fnull}{f_{\mathrm{null}}}
\newcommand{\fplann}{f_{\mathrm{pl}}^{\mathrm{ann}}}
\newcommand{\fnullann}{f_{\mathrm{null}}^{\mathrm{ann}}}
\newcommand{\LCS}{\mathsf{LCS}}

\DeclareFieldFormat[article]{title}{#1}

\newtheorem{theorem}{Theorem}[section]
\newtheorem{lemma}[theorem]{Lemma}
\newtheorem{proposition}[theorem]{Proposition}
\newtheorem{corollary}[theorem]{Corollary}

\newtheorem{conjecture}[theorem]{Conjecture}

\newtheorem{problem}[theorem]{Problem}

\theoremstyle{definition}
\newtheorem{definition}[theorem]{Definition}
\newtheorem{remarkx}[theorem]{Remark}
\newenvironment{remark}
  {\pushQED{\qed}\remarkx}
  {\popQED\endremarkx}

\crefname{appendix}{appendix}{appendices}
\Crefname{appendix}{Appendix}{Appendices}

\title{The Random Subsequence Model and Uniform Codes for the Deletion Channel}

\author[Ryan Jeong]{Ryan Jeong\authmark{1}}
\author[Francisco Pernice]{Francisco Pernice\authmark{2}}

\begin{document}

\begingroup
\renewcommand\thefootnote{\fnsymbol{footnote}}
\footnotetext[1]{Department of Statistics, Stanford University. Email: \nolinkurl{rsjeong@stanford.edu}.}
\footnotetext[2]{Department of Electrical Engineering and Computer Science, Massachusetts Institute of Technology. Email: \nolinkurl{fpernice@mit.edu}.}
\endgroup

\begin{abstract}

    We introduce the \textit{Random Subsequence Model}, a spin glass model on pairs of random strings $(X,Y) \in \{0,1\}^N \times \{0,1\}^M$ whose partition function counts subsequence embeddings of $Y$ into $X$. We study two variants: the null model, where $X$ and $Y$ are independent and uniform, and the planted model, where $X$ is uniform and $Y$ is a uniformly-random length-$M$ subsequence of $X$. We connect the Random Subsequence Model to longstanding problems in various fields, including the best rate achievable by uniformly-random codes in the deletion channel, the longest common subsequence problem between two random strings, and models of directed polymers in statistical physics.

    In the regime where $N,M\to\infty$ at a fixed ratio $\alpha = M/N \in (0,1)$, we exhibit strict asymptotic separations between the null annealed free energy and the quenched free energies of the null and planted models at all values of the density parameter $\alpha$. This suggests that these models are in a spin glass phase at zero temperature throughout the entire dense regime. As a consequence, we show that uniformly-random codes achieve a positive rate in the deletion channel for all deletion probabilities $p\in [0,1),$ settling multiple conjectures of \cite{pernice2024mutual} and proving the first such positive rate result for the regime $p \geq 1/2$.
    
    We also give an exact analytic formula for the annealed free energy of the planted model for all values of the density parameter. This implies a corresponding analytic upper bound on the best rate achievable by uniformly-random codes in the deletion channel, complementing the lower bound from our first result. Our upper and lower bounds for the capacity of the deletion channel under uniform codes are far closer to each other than the best known upper and lower bounds for the capacity of the deletion channel.
\end{abstract}

\maketitle


\section{Introduction}

This paper introduces and studies the \textit{Random Subsequence Model}, a new spin glass model whose zero-temperature partition function counts order-preserving subsequence embeddings for pairs of random binary strings. Given natural numbers $M \leq N$, the \textit{configuration space} of the model is
\begin{align*}
    \Sigma = \Sigma_{N,M} := \left\{\sigma : [M] \to [N] \,:\, \sigma \text{ is strictly increasing} \right\}.
\end{align*}
Given a pair of random strings $(X,Y) \in \{0,1\}^N\times \{0,1\}^M$, which we call the \textit{disorder}, the \textit{partition function} is defined as
\begin{align} \label{eq:Z-defn}
    Z_{X,Y} := |S_{X,Y}| := |\{\sigma\in\Sigma : X_\sigma=Y\}|,
\end{align}
where, for a configuration $\sigma \in \Sigma$, we write 
\begin{align*}
    X_\sigma = \big( X_{\sigma(1)},\dots,X_{\sigma(M)} \big) \in \{0,1\}^M.
\end{align*} 
In words, $Z_{X,Y}$ counts the number of order-preserving embeddings of \(Y\) into \(X\) as a subsequence. Fixing a density parameter $\alpha \in (0,1)$, we are interested in the asymptotic regime in which \(N,M\to\infty\) with \(M/N\to\alpha\in[0,1]\) under two natural laws on the disorder. 

\medskip

\noindent {\bf Null.} $X$ and $Y$ are drawn independently and uniformly from $\{0,1\}^N$ and $\{0,1\}^M$, respectively.

\medskip

\noindent {\bf Planted.} $X$ is drawn uniformly from $\{0,1\}^N$, $\sigma^*$ is drawn independently of $X$ and uniformly from $\Sigma,$ and we set $Y = X_{\sigma^*}$.

\medskip

We often write $X'$ for an independent copy of $X$ to distinguish the null ambient string from the planted one, so that the null partition function is $Z_{X',Y}$ and the planted one is $Z_{X,Y}$.

\subsection{Connections} \label{sec:connections}
Before presenting our results, we connect the Random Subsequence Model to some longstanding problems in information theory, discrete probability, theoretical computer science, and statistical physics. The connections to problems in information theory are repeated here for the reader's convenience from \cite{pernice2024mutual}, where the planted variant of the Random Subsequence Model was implicitly studied. For further such background, including connections to Slepian--Wolf theory and distributed storage, we refer the reader to the introduction of that paper.

\subsubsection{The deletion channel.}

The \textit{binary deletion channel with deletion probability $p \in [0,1]$} is the communication channel which takes input $x \in \{0,1\}^N$ and deletes each bit of $x$ independently with probability $p$ to produce a random output $y \in \{0,1\}^*$. The deletion channel is often regarded as the canonical example of a channel with synchronization errors, that is, in which synchronization between input and output bit positions is lost. Dobrushin's classical coding theorem for synchronization channels \cite{dobrushin1967shannon} showed that the \emph{capacity} of the deletion channel, which can be thought of informally as the maximum rate at which one can reliably transmit information through the channel, admits the variational representation
\begin{align*}
    \lim_{N\to\infty} \max_{X} \frac{1}{N} I\left( X; Y \right),
\end{align*}
where the maximum is over all distributions of $X$ supported on $\{0,1\}^N$, $Y$ is the output of the deletion channel on input $X$, and $I(\cdot\,; \cdot)$ is the mutual information. Finding an analytic formula expressing the capacity above as a function of the deletion probability $p$ has been one of the outstanding problems of information theory over the past several decades. 

A natural relaxation of the capacity problem is to consider $X\sim \Unif\{0,1\}^N$ rather than optimizing over all possible laws of $X.$ We refer to the corresponding limit
\begin{align}\label{eq:unif-mut-inf-limit}
    \lim_{N\to\infty}\frac{1}{N} I\left( X; Y \right), \qquad X\sim \Unif\{0,1\}^N
\end{align}
as the \emph{uniform capacity} of the deletion channel, which gives a lower bound on the full channel capacity, and can be interpreted as the maximum rate that can be achieved with \emph{uniformly random error correcting codes}. Uniformly random codes are of substantial importance as they are known to be asymptotically optimal (i.e., they achieve the channel capacity) in many well-understood memoryless channels like the binary symmetric channel and the binary erasure channel (e.g., see \cite{shannon1948mathematical}), the latter of which is often thought of as the memoryless analogue of the binary deletion channel. As such, uniform codes are widely studied throughout information theory and coding theory. While it is known (e.g., see \cite{drmota2012mutual} and \cite{drinea2006simple}) that uniform codes cannot achieve the capacity of the deletion channel for $p$ close to $1$, we believe that studying uniformly-random codes in this setting is both interesting in its own right and potentially an important stepping stone towards the full capacity problem, since deriving an analytic formula or even a convincing conjectural expression for \eqref{eq:unif-mut-inf-limit} is wide open.

The uniform capacity of the deletion channel is closely connected to the planted variant of the Random Subsequence Model, and this connection is the primary motivation for the present work. Specifically, it is easy to show (see \cref{app:del-chann-cap-proof}) that, letting $\alpha = 1-p,$ one can write
\begin{align} \label{eq:del-chann-cap-connection}
    \lim_{N\to \infty} \frac{1}{N}I(X;Y) &= \alpha \log 2 - h(\alpha) + \lim_{\substack{N,M\to\infty \\ M/N = \alpha}} \frac{1}{N}\E\left[ \log Z_{X,Y} \right],
\end{align}
where all logs are taken base $e$ and 
\begin{align*}
    h(\alpha) = -\alpha \log \alpha - (1-\alpha)\log (1-\alpha)
\end{align*}
denotes the usual binary entropy function. Hence, understanding the \emph{limiting free energy density} 
\begin{align*}
    \fpl(\alpha) := \lim_{\substack{N,M\to\infty \\ M/N = \alpha}} \frac{1}{N}\E\left[ \log Z_{X,Y} \right]
\end{align*}
of this model is equivalent to computing the maximum rate achieved by uniformly random codes in the deletion channel.

Finally, we note that the null variant of the Random Subsequence Model is relevant to uniformly random codes for the deletion channel as well, and refer the reader to \cite{mitzenmacher2009survey} for additional discussion in this direction. Indeed, let $\mathcal{C}\subseteq \{0,1\}^N$ be a uniformly-random code, with all elements of $\mathcal{C}$ drawn i.i.d. uniformly from $\{0,1\}^N.$ Let $X \sim \mathcal{C}$ be a uniformly-random codeword and $Y$ the output of the deletion channel on input $X$. The maximum-likelihood decoder outputs
\begin{align*}
    \hat{x}(Y)=\argmax_{x\in \mathcal{C}} \Prob(Y \mid X=x) = \argmax_{x\in \mathcal{C}} Z_{x, Y},
\end{align*}
since as a function of $x \in \{0,1\}^N$, it holds that $\Prob(Y \mid X=x) \propto Z_{x, Y}$. Thus, we have that
\begin{align*}
    \Prob\left(\hat{x}(Y)\neq X \right) & \leq \Prob\left( Z_{x',Y} \geq Z_{X,Y} \text{ for some } x' \in \mathcal{C} \setminus \{X\} \right) \leq (|\mathcal{C}| - 1) \Prob\left( Z_{X',Y} \geq Z_{X,Y} \right).
\end{align*}
But note that, since $X'$ is independent of $X,$ $Z_{X',Y}$ and $Z_{X,Y}$ are exactly distributed as the partition functions of the null and planted variants of the Random Subsequence Model, respectively. Hence, if one shows that the planted partition function is greater than the null partition function with high probability, that gives a bound on the maximum-likelihood decoder's probability of failure. This maximum likelihood decoding strategy, combined with the sub-optimal bound
\begin{align*}
    \Prob\left(Z_{X',Y} \geq Z_{X,Y}\right) \leq \Prob\left( Z_{X',Y} \geq 1 \right),
\end{align*}
is the heart of the argument of \cite{gallager1961sequential, zigangirov1969sequential, diggavi2001transmission}, which gave one of the early lower bounds on the capacity of the deletion channel.

\subsubsection{Longest common subsequence of two random strings.} 

A second source of motivation for the Random Subsequence Model comes from the longest common subsequence ($\LCS$) problem. If $X_1$ and $X_2$ are independent uniformly random binary strings of respective lengths $N_1$ and $N_2$, we let $\LCS(X_1,X_2)$ denote the length of their longest common subsequence. A classical open problem of more than fifty years is to determine the limit
\begin{align*}
    f(\gamma)& := \lim_{\substack{N_1,N_2\to\infty \\ N_2/N_1 = \gamma}} \frac{1}{N_1} \E\left[ \LCS(X_1,X_2) \right].
\end{align*}
The case $\gamma=1$, proposed by \cite{chvatal1975longest}, has been the subject of intensive study \cite{deken1979limit, steele1982long, alexander1994rate, dancik1995upper, heineman2024improved}. Following \cite{boutetdemonvel1999extensive}, one may instead study the partition function
\begin{align*}
    Z_{N_1,N_2,M} & := |\{\sigma^1\in \Sigma_{N_1,M},\, \sigma^2\in \Sigma_{N_2,M} : (X_1)_{\sigma^1} = (X_2)_{\sigma^2}\}|,
\end{align*}
which counts common length-$M$ subsequences of $X_1$ and $X_2$. So in particular, $Z_{N_1,N_2,M} > 0$ if and only if there exists a common subsequence between $X_1$ and $X_2$ of length $M$. Then one can compute the associated free energy
\begin{align*}
    g(\gamma, \alpha) & := \lim_{\substack{N_1,N_2,M\to\infty \\ N_2/N_1 = \gamma,\, M/N_1 = \alpha}} \frac{1}{N_1} \E\left[ \log (1+Z_{N_1,N_2,M}) \right].
\end{align*}
This free energy encodes the $\LCS$ constant, as it is easy to show that
\begin{align*}
    f(\gamma) = \sup\left\{\alpha : g(\gamma,\alpha) > 0 \right\}.
\end{align*}
The null variant of our Random Subsequence Model corresponds to the special case where $\gamma = \alpha.$ Indeed, in that case one has $M=N_2$, and then the second ambient string is itself the candidate subsequence, so
\begin{align*}
    Z_{N_1,N_2,M} = |\{\sigma\in \Sigma_{N_1,M}: (X_1)_{\sigma} = X_2\}| = Z_{X_1, X_2},
\end{align*}
which is exactly \eqref{eq:Z-defn}. Thus, the Random Subsequence Model may be viewed as a natural slice of the partition-function framework for the longest common subsequence problem. We therefore believe that a better understanding of the Random Subsequence Model is likely to be of interest for the longest common subsequence problem as well.

Lastly, we also note that longest common subsequences play a central role in deletion coding more broadly \cite{guruswami2022zero}. For adversarial deletions, extremal bounds on pairwise $\LCS$ govern codebook feasibility, and even the analysis of random deletion codes is limited by the current lack of sharp estimates for the expected $\LCS$ of two random binary strings.

\subsubsection{Mean-field variants and directed polymers.} \label{subsubsec:mean-field}

Our third and final connection is to \emph{directed polymers} in statistical physics. This connection is nontrivial, and will help us frame the discussion by using the experience of the directed polymers literature to identify which models can be expected to admit exact analytic solutions (see \Cref{subsec:main_result}). 

Given $(X,Y)\in \{0,1\}^N\times \{0,1\}^M,$ $n\leq N$ and $m\leq M,$ we use $X_{1:n}, Y_{1:m}$ to denote the prefixes of $X,Y$ from index 1 to $n,m,$ respectively. We also denote
\[
Z_{n,m} := Z_{X_{1:n}, Y_{1:m}},
\]
noting that $Z_{N,M} = Z_{X,Y}.$ By partitioning
\begin{align*}
    S_{X_{1:n},Y_{1:m}} = (S_{X_{1:n},Y_{1:m}}\cap \{\sigma_m = n\})\sqcup (S_{X_{1:n},Y_{1:m}}\cap \{\sigma_m \neq n\}),
\end{align*}
we observe the recurrence relation
\begin{align} \label{eq:Z-recurrence}
    Z_{n,m} &= Z_{n-1, m} + \1\{X_n = Y_m\}\cdot Z_{n-1,m-1}.
\end{align}
Note that the above gives an efficient dynamic programming algorithm to compute $Z_{N,M}$ exactly. It also leads to a natural generalization of our model. Rather than specifying two strings $X$ and $Y,$ we can instead specify a matrix $B \in \R_+^{N\times M}$ and \emph{define} the partition function $Z_{N,M}$ via the recurrence relation
\begin{align}\label{eq:Z-B-recurrence}
    Z_{n,m} = Z_{n-1, m} +B_{n,m}\cdot Z_{n-1,m-1}, \qquad 1\leq n\leq N,\ 1\leq m\leq M,
\end{align}
with the initial condition $Z_{0,m} = \delta_{0,m}.$ This can be interpreted as a directed polymer model, and in the case that $B$ has entries in $\{0,1\}$, a directed percolation model, in a random environment determined by $B$. Indeed, consider a box in $\Z^2$ with bottom-left endpoint at $(0,0)$ and top-right endpoint at $(N,M).$ Put an edge between any two horizontal neighbors $(n-1,m),(n,m)$ with weight 1 and an edge between diagonal neighbors $(n-1,m-1),(n,m)$ with weight $B_{n,m}.$ Then $Z_{N,M}$ counts weighted paths from $(0,0)$ to $(N,M)$ that are monotonically increasing in the first coordinate.

This kind of model has received considerable attention in the directed polymers literature. An insight of those investigations is that only distributions of $B$ with special algebraic structure admit exact analytic solutions with existing techniques. The case where $B$ has i.i.d. Gamma-distributed entries falls in that category, and its exact analytic solution was found in the important work \cite{corwin2015strict}. We refer the reader to the next section for a discussion of how this relates to the Random Subsequence Model. As far as the authors are aware, the Random Subsequence Model itself, which corresponds to a natural ``rank-one'' random environment $B,$ has not been explicitly studied in the directed polymers literature. Cases where the entries of $B$ are i.i.d. represent a natural \emph{mean field} version of the Random Subsequence Model, and the setting where $B$ has i.i.d. $\Unif\{0,1\}$ entries has been studied under the name ``Bernoulli Matching Model'' in connection to the longest common subsequence problem (see \cite{boutetdemonvel1999extensive, majumdar2005exact}).

\subsection{Our results} \label{subsec:main_result}

Our first result gives strict bounds on the free energies of the null and planted variants of the Random Subsequence Model. For the rest of the paper, following the statistical physics terminology, we use the term \emph{annealed} free energy to refer to the quantities of the form
\[
\frac{1}{N}\log \E[Z],
\]
i.e., with the expectation \emph{inside} the log. This is in contrast to the default, \emph{quenched} free energy 
\begin{align*}
    \frac{1}{N}\E[\log Z],
\end{align*}
where the expectation is outside the log.

\begin{theorem}[Quenched-annealed gaps] \label{thm:weak_law}
    Fix $\alpha \in (0, 1)$. Let $X, X' \in \{0,1\}^N$ be independent uniform random strings and let $Y \in \{0,1\}^M$ be a uniform random length-$M$ subsequence of $X$ (with $M = \lfloor \alpha N \rfloor$). Defining
    \begin{align*}
        \fnullann(\alpha) := \lim_{N\to\infty} \frac{1}{N}\log \E\left[Z_{X',Y}\right],
    \end{align*}
    there exists a constant $\fpl(\alpha) > 0$ such that
    \begin{align} \label{eq:weak_law_planted_bounds}
        \frac{1}{N} \log Z_{X,Y} \xrightarrow{p} \fpl(\alpha) > \fnullann(\alpha),
    \end{align}
    where $\xrightarrow{p}$ denotes convergence in probability.
    Moreover, if $\alpha \in (0, 1/2)$, then there exists a constant $\fnull(\alpha) > 0$ such that
    \begin{align} \label{eq:weak_law_null}
        \frac{1}{N} \log Z_{X',Y} \xrightarrow{p} \fnull(\alpha)
    \end{align}
    and this constant satisfies
    \begin{align} \label{eq:weak_law_null_bounds}
        h( 2\alpha )/2 \leq \fnull(\alpha) < \fnullann(\alpha).
    \end{align}
\end{theorem}

The threshold $\alpha=1/2$ is the natural boundary for the weak limit \eqref{eq:weak_law_null} to hold under the null model. Indeed, as discussed at the beginning of \cref{sec:weak_law_existence}, it is easy to show that $Z_{X',Y} = 0$ with high probability for $\alpha > 1/2$ and that at the boundary $\alpha = 1/2$, there does not exist a weak limit of the form \eqref{eq:weak_law_null} for the null model. Additionally, the main quantitative result in \cref{sec:main_result_upper_bound} which separates the exponential behavior of $Z_{X',Y}$ from the annealed free energy of the null model, namely \eqref{eq:null_quantitative}, holds for all $\alpha \in (0,1)$ (though with implicit constants depending on the choice of $\alpha$). Together with \cref{thm:weak_law}, these observations give strict exponential-scale bounds for the Random Subsequence Model throughout the entire dense regime $\alpha \in (0,1)$. We record in \cref{app:no-double-phase-transition} a combinatorial consequence of \eqref{eq:weak_law_null_bounds}, showing in particular that there is no further transition in the relation between $\fnull(\alpha)$ and $\fnullann(\alpha)$ within the range $\alpha < 1/2$.

\smallskip


We now connect \cref{thm:weak_law} to channel coding over the binary deletion channel with uniform random codes; see \cite{pernice2024mutual} for a more complete discussion of this viewpoint. For $p \in [0,1]$, we let $C_{\unif}(p) \in [0,1]$ denote the maximum rate achievable through the deletion channel with deletion probability $p$ using uniform random codes, i.e., the quantity that was introduced in \eqref{eq:unif-mut-inf-limit}.

\begin{corollary}[Positive rate for uniform codes; {\cite[Conjecture 3]{pernice2024mutual}}]\label{cor:pos_mutual_info}
    It holds for all deletion probabilities $p \in [0, 1)$ that $C_{\unif}(p) > 0$.
\end{corollary}

\cref{cor:pos_mutual_info} confirms \cite[Conjecture 3]{pernice2024mutual} and resolves \cite[Question 1]{pernice2024mutual} by showing that uniformly random codes achieve a strictly positive rate even when $p \geq 1/2$, i.e., in the regime where each bit is more likely to be deleted than retained. This is the first such positive-rate result for uniform codebooks in this regime, and it strengthens (in a qualitative sense) the earlier lower bounds of \cite{gallager1961sequential, zigangirov1969sequential, diggavi2001transmission, drinea2006simple, rahmati2013bounds, han2016mutual}, all of which were only able to address the $p < 1/2$ setting.

In fact, the proof of \cref{thm:weak_law} yields a quantitative version of \cref{cor:pos_mutual_info} in the likely deletion regime. Specifically, there exists $k \in \mathbb{N}$ (the crude explicit bound of \cref{thm:quantitative_capacity_bound} gives $k=362$, which we expect to be far from sharp) such that as $p \to 1$,
\begin{align*}
    C_{\mathrm{unif}}(p) = \Omega\bigl((1-p)^k\bigr).
\end{align*}
On the other hand, as $p \to 1$, \cite{drmota2012mutual} gives the upper bound
\begin{align*}
    C_{\mathrm{unif}}(p) = O\left((1-p)^{4/3}\log \frac{1}{1-p}\right).
\end{align*}
These bounds together show that the uniform capacity decays polynomially but faster than linearly in $1-p$ as $p \to 1$.

\smallskip

We note that $\fnullann(\alpha)$ trivially admits an exact analytic formula. Indeed, we have
\begin{align*}
    f_{\text{null}}^{\text{ann}}(\alpha) &= \lim_{N\to \infty} \frac{1}{N}\log \E [Z_{X',Y}] \\
    &= \lim_{N\to \infty} \frac{1}{N}\log \left(|\Sigma_{N,M}|2^{-M} \right) = \lim_{N\to \infty} \frac{1}{N}\log \left(\binom{N}{M}2^{-M} \right) = h(\alpha) - \alpha \log 2.
\end{align*}
The situation is significantly more involved for the annealed free energy of the planted model
\[
f_{\text{pl}}^{\text{ann}}(\alpha) = \lim_{N\to \infty} \frac{1}{N}\log \E [Z_{X,Y}].
\]
However, our next result gives an exact analytic formula for this quantity as well.

\begin{theorem}[Annealed free energy of the planted model]\label{thm:main-annealed-intro}
For every $\alpha \in (0,1)$, define
\begin{align*}
    \Delta(\alpha)\;=\;\sqrt{9\alpha^{2}-4\alpha+4};\qquad x_\alpha\;=\;\frac{\Delta(\alpha)-3\alpha}{2}; \qquad y_\alpha\;=\;\frac{(3\alpha+2-\Delta(\alpha))^2}{2(\Delta(\alpha)-3\alpha)(2+\Delta(\alpha)-3\alpha)}.
\end{align*}
Then $(x_\alpha,y_\alpha)\in(0,\infty)^2$ and
\begin{align*}
    \fplann(\alpha)\;=\;-h(\alpha) - \alpha \log 2- \log x_\alpha\;-\;\alpha\log y_\alpha.
\end{align*}
\end{theorem}
Note that, by Jensen's inequality, $\fplann(\alpha)$ is an upper bound on $\fpl(\alpha).$ Combining this with \eqref{eq:del-chann-cap-connection} yields an upper bound on the uniform capacity of the deletion channel for all deletion probabilities. We note that this is one of the very few known analytic bounds for the deletion channel \cite{diggavi2001transmission, cheraghchi2019capacity}, and is numerically far closer than any other upper bounds of which we are aware, albeit bounding the uniform capacity rather than the full capacity. We view this as further evidence that the uniform capacity is a fruitful stepping stone towards better understanding the full capacity problem. Finally, we note that \cite{pernice2024mutual} previously gave an efficient \emph{algorithm} to compute the annealed free energy in the planted model to any desired precision. Our result significantly improves on theirs by being exact and analytic.

\begin{figure}[ht]
    \centering
    \includegraphics[width=0.8\linewidth]{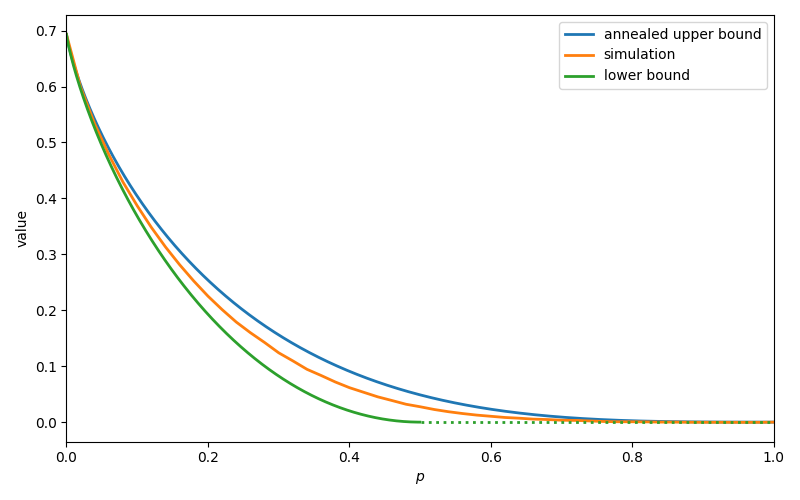}
    \caption{Our lower (green) and upper (blue) bounds on the uniform capacity of the deletion channel, together with a simulation-based computation of the true capacity curve (orange). All curves equal $\log 2$ at $p=0$ since the $y$-axis is in nats. The green curve is obtained by combining the lower bound $(\log 2-h(p))\1\{p\leq 1/2\}$ from \cite{diggavi2001transmission} with \Cref{cor:pos_mutual_info} (the dotted line means that the inequality is strict). The blue curve is obtained from \Cref{thm:main-annealed-intro} and \eqref{eq:del-chann-cap-connection}. The orange curve is obtained by numerically solving \eqref{eq:Z-recurrence} up to $N=10{,}000.$}
    \label{fig:capacity-simulation}
\end{figure}
We conclude this section by returning to the connection with directed polymers. As we mentioned in \Cref{sec:connections}, the Random Subsequence Model can be understood as a directed polymer model in a random ``rank-one'' environment. In the directed polymer literature, only models with special algebraic structure have been found to admit exact solutions. Since this kind of structure does not seem to be present in the Random Subsequence Model, the experience of the directed polymer literature suggests that it may be intractable to obtain exact analytic formulas for $\fpl$ or $\fnull$ with existing mathematical techniques, and results of the kind proved in \Cref{thm:weak_law} and \Cref{thm:main-annealed-intro} may be the best we can hope for. However, we think that a promising direction is to study variants of the Random Subsequence Model which do admit such exact solutions. There is a long tradition of doing this in the literature on spin glasses. For example, the Sherrington-Kirkpatrick model, which led to tremendous progress, was first proposed as a simplified, mean-field version of the earlier Edwards-Anderson model \cite{edwards1975theory, sherrington1975solvable}. In our case, the natural solvable analogue is the Strict-Weak Polymer Model, introduced and solved in \cite{corwin2015strict}. After a trivial mapping, this corresponds to exactly the same recurrence as \eqref{eq:Z-B-recurrence}, but where we take $B$ to have i.i.d. $\text{Gamma}(a, b)$ entries. The exact solution in this case is (see \cite[Theorem 1.3]{corwin2015strict})
\begin{align}\label{eq:strict-weak-exact-soln}
    \lim_{\substack{N,M\to\infty \\ M/N = \alpha}} \frac{1}{N}\E\left[ \log Z_{N,M} \right] &=  \inf_{\lambda>0} \Bigl\{ -(1-\alpha)\,\Psi(\lambda) +\Psi(a+\lambda) +\alpha\log b \Bigr\},
\end{align}
where $\Psi$ is the digamma function 
\[ 
\Psi(x)=\frac{d}{dx}\log \Gamma(x)=\frac{\Gamma'(x)}{\Gamma(x)}. 
\]
We can choose $a$ and $b$ so as to agree with the null variant of the Random Subsequence Model as much as possible by insisting that the entries of $B$ have the same mean and variance as they would in the Random Subsequence Model. This amounts to choosing $a=1,b=1/2,$ which corresponds to the entries of $B$ being i.i.d. $\text{Exponential}(2).$ In \cref{fig:fig2}, we plot the exact solution \eqref{eq:strict-weak-exact-soln} alongside the numerical solution of \eqref{eq:Z-recurrence} for the null Random Subsequence Model. The fact that the two curves are very close suggests that the Strict-Weak Polymer Model may be fruitful to study to gain insight on the Random Subsequence Model. We pose some open problems in this direction in \Cref{sec:open_problems}.
\begin{figure}[ht]
    \centering
    \includegraphics[width=0.8\linewidth]{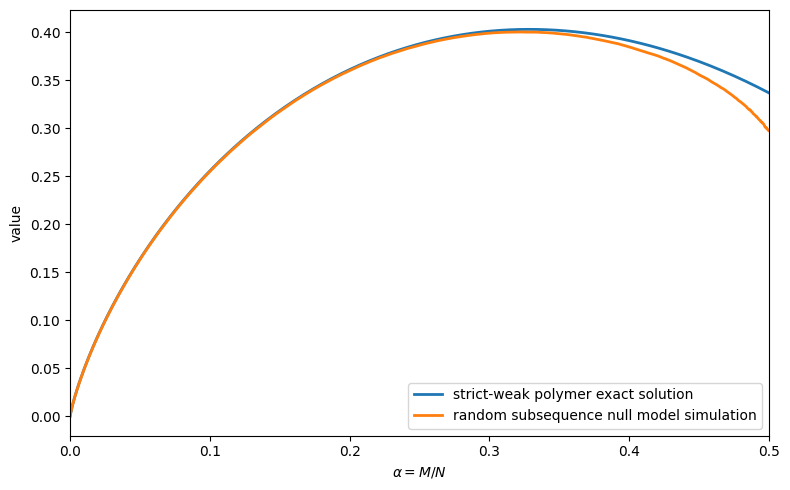}
    \caption{The blue curve is the exact solution \eqref{eq:strict-weak-exact-soln} for $a=1,b=1/2,$ and the orange curve is obtained by numerically solving \eqref{eq:Z-recurrence} for the null Random Subsequence Model, $N=10{,}000.$ The plot is in the interval $\alpha \in [0,1/2]$ because, outside of that range, we have $Z_{X,Y}=0$ with high probability for the null Random Subsequence Model.}
    \label{fig:fig2}
\end{figure}

\subsection{Outline of the proofs}

We now sketch the proofs of the main results, focusing on the key mechanisms that drive the analysis and postponing the technical details to the body of the paper.

\subsubsection{Proof overview of \cref{thm:weak_law} and \cref{cor:pos_mutual_info}.}

The convergence statements in \eqref{eq:weak_law_planted_bounds} and \eqref{eq:weak_law_null} are essentially obtained in \cref{sec:weak_law_existence} by invoking the superadditive ergodic theorem on the log-partition function of the model. In the planted model, $\log Z_{X,\BDC_{1-\alpha}(X)}$ is genuinely superadditive, and a coupling between $\BDC_{1-\alpha}(X)$ and a uniform length-$M$ subsequence transfers the weak limit to the fixed-length planted model. Rare occurrences can break this superadditivity in the null model, so we instead prove the weak law by directly exploiting the self-replicating structure of the subsequence count. We also obtain the first lower bound of \eqref{eq:weak_law_null_bounds} by encoding embeddings of $S_{X',Y}$ via \textit{skip vectors}, which record at each step how many occurrences of the current target bit are passed over a greedy embedding procedure. Realizing a vector with $s$ total skips is equivalent to requiring that a sum of $M+s$ i.i.d. $\Geom(1/2)$ random variables be at most $N$.

The main content of \cref{thm:weak_law} is therefore the strict chain of quenched-annealed gaps given across \eqref{eq:weak_law_planted_bounds} and \eqref{eq:weak_law_null_bounds}, and we focus here on the ideas behind its proof. Our strategy is to recast separation between the null and planted variants of the Random Subsequence Model as a structural hypothesis testing problem. For each ambient string $x \in \{0,1\}^N$, we define an $x$-\textit{good set}
\begin{align*}
    \mathcal G(x) \subseteq \{0,1\}^M
\end{align*}
of strings that are ``well-aligned" to $x$ so that we obtain the following distinguishing property. 
\begin{itemize}
    \item For $(X,Y)$ drawn from the planted model, $Y \in \mathcal G(X)$ with probability $1 - e^{-\Omega(N)}$.

    \item For $(X',Y)$ drawn from the null model, $Y \notin \mathcal G(X')$ with probability $1 - e^{-\Omega(N)}$.
\end{itemize}
The strict inequality of \eqref{eq:weak_law_null_bounds} is then obtained by showing that for typical $x$, only an exponentially small fraction of its length-$M$ subsequences belong to $\mathcal G(x)$. The strict inequality of \eqref{eq:weak_law_planted_bounds} essentially comes from combining this null estimate with the size-bias relation between the null and planted laws.

To define $\mathcal G(x)$, we partition $x$ into $B = B(N) := N/b$ contiguous blocks $x^{(1)},\dots,x^{(B)}$ of length $b$, where $b$ is large but fixed. If $Y$ is a planted subsequence of $x$, then the embedding induces a decomposition
\begin{align*}
    Y = \big(Y^{(1)},\dots,Y^{(B)}\big),
\end{align*}
where $Y^{(i)}$ is the portion of $Y$ contributed by the block $x^{(i)}$. For a typical block of $x$, the sign of the block sum in $Y^{(i)}$ is more likely than not to agree with that of $x^{(i)}$, and the strength of this bias is captured by the random-walk displacement $\Delta(Y^{(i)})$. This leads to the \textit{local alignment} $A_\loc(x^{(i)},y^{(i)})$, which rewards agreement of block majorities while weighting it by the size of the displacement. Taking the supremum of the resulting average local alignment over \textit{induced near-equipartitions} of $y$, thought of as the collection of typical block structures of $Y$ drawn from the planted model, yields the \textit{induced total alignment} $T_\ind(x,y)$. \cref{subsec:alignment_W_random_Y} shows that for every typical $x$, with probability $1 - e^{-\Omega(N)}$, a planted subsequence $Y$ satisfies
\begin{align*}
    T_\ind(x,Y) \geq \frac{1}{2} + \beta^\star(\alpha),
\end{align*}
where $\beta^\star(\alpha) > 0$ is an appropriate constant. This is the \textit{regular alignment property} that defines $\mathcal G(x)$.

The heart of the proof is the corresponding inverse problem for the null model. Here $Y$ is independent of $X'$, and we must rule out the possibility that an adversary can produce an induced near-equipartition of $Y$ so as to mimic the macroscopic block structure of $X'$. This is difficult because it can be shown that
\begin{enumerate}
    \item many natural block statistics can be driven above $1/2$ by a favorable induced near-equipartition, so one needs a score that is robust under adversarial choices;

    \item there are exponentially many induced near-equipartitions of $Y$, and this family is large enough that standard metric entropy arguments over the collection of all such induced near-equipartitions do not provide appropriate guarantees.
\end{enumerate}
The notions of alignment introduced above are so that we may overcome the first obstacle. Our main device for overcoming the second obstacle is the \emph{standardization algorithm}. This algorithm defines a map
\begin{align*}
    \varphi_Y : \mathcal{NE}_\ind(Y) \to \mathcal{NE}_\std(Y),
\end{align*}
sending induced near-equipartitions of $Y$ to a much smaller class of \textit{standardized near-equipartitions} in which almost all block lengths are some fixed length. We construct the standardization algorithm in such a way that the typical ``extremal increase'' in the value of the average local alignment is modest, so we may think of the standardization algorithm as an approximation algorithm. More precisely, \cref{subsec:alignment_impossibility_ind_Y} shows that with probability $1 - e^{-\Omega(N)}$, every induced near-equipartition is mapped to a standardized one whose score is larger by at most $\beta^\star(\alpha)/2$. The proof reduces the worst-case effect of standardization to fluctuation bounds for contiguous substrings of $Y$: large gains can only come from \textit{biased stretches}, and a uniform random string contains too few such stretches for them to matter on the exponential scale. Once this reduction is established, a union bound over the much smaller family $\mathcal{NE}_\std(Y)$, together with concentration for each fixed partition, bounds the \textit{standardized total alignment} via
\begin{align*}
    T_\std(X',Y) < \frac{1+\beta^\star(\alpha)}{2}
\end{align*}
with overwhelming probability. Hence, $Y \notin \mathcal G(X')$.

\cref{subsec:proof_of_thm} converts this structural separation into free-energy separation. Under the null model, the event $Y \notin \mathcal G(X')$ implies that only an exponentially sparse subcollection of the $\binom{N}{M}$ candidate embeddings contributes. Together with the regular alignment property for typical $x$, this yields the quantitative result
\begin{align} \label{eq:null_quantitative_intro}
    \Prob\left( Z_{X', Y} < e^{N\left( \fnullann(\alpha) - \Omega(1) \right)} \right) \geq 1 - e^{-\Omega(N)}.
\end{align}
This implies that $f_{\mathrm{null}}(\alpha) < f_{\mathrm{null}}^{\mathrm{ann}}(\alpha)$. Under the planted model, \eqref{eq:null_quantitative_intro} together with the size-bias relation between planted and null embeddings forces $Z_{X,Y}$ above the null annealed scale by a fixed exponential margin, giving $\fnullann(\alpha) < \fpl(\alpha)$. This proves \cref{thm:weak_law}. \cref{cor:pos_mutual_info} then follows by substituting the resulting gap into the mutual information identity \eqref{eq:del-chann-cap-connection}.

\subsubsection{Proof overview of \cref{thm:main-annealed-intro}}
The starting point of the proof is to rewrite the planted annealed partition function in a way that exposes the geometry of \emph{pairs} of embeddings rather than individual embeddings. Indeed, a simple argument leads to the formula
\[
\E\left[ Z_{X, Y} \right] = \frac{1}{\binom{N}{M}2^M}\sum_{\sigma,\tau \in\Sigma}  2^{\brac{\sigma,\tau}},
\]
where $\brac{\sigma,\tau}$ denotes the \emph{overlap} 
\begin{align*}
    \sum_{j=1}^M \1\{\sigma(j) = \tau(j)\}.
\end{align*}
Hence, the goal of the rest of the proof is to understand the moment generating function of the overlap of a random pair of configurations. The key structural property used to understand this is a certain Markov property for the distribution of $(\sigma,\tau) \sim \Sigma^2:$ if we select indices $i\in [N]$ and $j\in [M]$ and condition on the event that $\sigma(j) = \tau(j) = i$, then the pairs 
\begin{align*}
    & (\sigma|_{[j-1]},\tau|_{[j-1]}),
    & (\sigma|_{[M]\setminus [j]},\tau|_{[M]\setminus [j]})
\end{align*}
become conditionally independent and, after a trivial relabeling, uniform in $\Sigma_{i-1,j-1}^2,\Sigma_{N-i,M-j}^2,$ respectively. This leads to a recursive structure for the moment generating function which underlies the proof of \cref{lem:gap}. This lemma proves the formula
\begin{align}\label{eq:mgf-simplification}
\sum_{\sigma,\tau \in\Sigma}  2^{\brac{\sigma,\tau}} = \sum_{\ell=0}^{M}\ \sum_{1\le j_1<\cdots<j_\ell\le M}\ \sum_{1\le i_1<\cdots<i_\ell\le N} \ \prod_{k=1}^{\ell+1}\binom{i_k-i_{k-1}-1}{\,j_k-j_{k-1}-1\,}^{\!2}
\end{align}
where we define $i_0=j_0=0$ and $i_{\ell+1}=N+1,j_{\ell+1}=M+1$. The benefit of this formula is that the expression inside the summand depends on the tuples $(j_1,\dots,j_\ell)$ and $(i_1,\dots,i_\ell)$ only through the empirical measure
\[
\mu = \frac{1}{\ell + 1} \sum_{k=1}^{\ell+1} \delta_{(i_k-i_{k-1}, j_k-j_{k-1})}
\]
supported on $\N^2_{>0}.$ This is finite-dimensional, so as $N,M\to\infty,$ \eqref{eq:mgf-simplification} turns into a constrained variational problem in the space of probability measures in $\N^2_{>0}.$ Indeed, the limit of the log of \eqref{eq:mgf-simplification} normalized by $N$ is shown to be 
\[
\mathsf R(\alpha)=\sup_{\rho\in(0,1)} \alpha\rho\,\Phi(\rho),
\]
where for each fixed $\rho$, the inner quantity $\Phi(\rho)$ is the supremum of
\[
H(\nu)+\E_\nu[\log w(a,b)],
\]
where $H$ is the Shannon entropy and $w(a,b) = \binom{a-1}{b-1}^2,$
over probability measures $\nu$ supported in $\N^2_{>0}$ satisfying the mean constraints
\[
\E_{(a,b)\sim \nu}[a]=\frac{1}{\alpha\rho},\qquad \E_{(a,b)\sim \nu}[b]=\frac{1}{\rho}.
\]
Conceptually, this is the heart of the proof: all of the combinatorics has now been absorbed into a finite-dimensional order parameter, namely the law solving the variational problem $\Phi(\rho).$

The second half of the argument is to solve this variational problem explicitly. For fixed $\rho$, one introduces Lagrange multipliers for the normalization and moment constraints. This shows that the maximizing measure must lie in an exponential family:
\[
\nu^*(a,b)\propto w(a,b)\,x^a y^b
\]
for suitable parameters $x,y>0$. Thus the optimization is encoded by a two-variable partition function
\[
Z(x,y)=\sum_{a\ge 1}\sum_{1\le b\le a} w(a,b)x^a y^b.
\]
The value of the variational problem can then be expressed in terms of $\log Z(x,y)$ together with the moment constraints. The next step is therefore analytic rather than probabilistic: one computes \(Z(x,y)\) in closed form. This is done by rewriting the sum using Vandermonde’s identity and then evaluating the resulting generating function explicitly.

Once the inner variational problem over $\nu$ has been solved, one returns to the outer optimization in $\rho$. Here the key observation is that at an interior optimizer, the envelope theorem implies that the derivative of the outer objective is proportional to \(\log Z(x,y)\). Hence optimality forces the normalization condition
\[
Z(x,y)=1.
\]
This is what collapses the remaining optimisation to a purely algebraic system. Combining the equation \(Z(x,y)=1\) with the stationarity relation coming from the moment constraints yields two equations in the two unknowns \(x\) and \(y\). Solving this system produces the explicit formulas for \(x_\alpha\) and \(y_\alpha\), and substituting them back gives
\[
\mathsf R(\alpha)=-\log x_\alpha-\alpha\log y_\alpha,
\]
which concludes the proof.

\subsection{Notation and conventions} \label{subsec:notation}

We employ standard asymptotic notation throughout, occasionally subscripted to clarify the relevant limiting variable (e.g., $o_b(1)$ denotes a quantity that vanishes as $b \to \infty$), though $o_p$ retains its usual probabilistic meaning. Unless otherwise stated, all asymptotics are taken as $N, M \to \infty$ with $\alpha = M/N$ fixed, and we omit floor and ceiling symbols when doing so does not affect asymptotic behavior. Because many of our estimates are meaningful only up to subexponential factors, we introduce the shorthand
\begin{align*}
    f(N) \approx g(N)
    \quad \Longleftrightarrow \quad
    \frac{f(N)}{g(N)} = e^{o(N)} \;\; \text{and} \;\; \frac{g(N)}{f(N)} = e^{o(N)}.
\end{align*}
It is immediate that $\approx$ defines an equivalence relation on functions from $\N$ to $\R_{>0}$. Given $p \in [0,1]$, we let $\BDC_p(X)$ denote the (random) string resulting from passing $X$ through the binary deletion channel with deletion probability $p$. All logarithms in this paper are taken base $e$.

\subsection{Organization} The rest of the paper is organized as follows. In \cref{sec:weak_law_existence}, we prove the lower bound of \eqref{eq:weak_law_null_bounds} and establish the weak limits of \eqref{eq:weak_law_planted_bounds} and \eqref{eq:weak_law_null}. Proving the strict inequalities of \eqref{eq:weak_law_planted_bounds} and \eqref{eq:weak_law_null_bounds} constitutes the main technical challenge of the proof of \cref{thm:weak_law}. We carry this out in \cref{sec:main_result_upper_bound} and then establish \cref{thm:weak_law} and \cref{cor:pos_mutual_info}. In \cref{sec:annealed_free_planted}, we derive the exact annealed free energy formula of \cref{thm:main-annealed-intro}. We conclude in \cref{sec:open_problems} with some open problems and suggestions for future research.

\section{Weak Limits for the Null and Planted Models} \label{sec:weak_law_existence}

We initiate the proof of \cref{thm:weak_law} by establishing the existence of weak limits $\fnull(\alpha) > 0$ (for values $\alpha \in (0,1/2)$ of the density parameter) and $\fpl(\alpha) > 0$ for which \eqref{eq:weak_law_planted_bounds} and \eqref{eq:weak_law_null} hold.

\subsection{Null model} \label{subsec:null_model}

Throughout \cref{subsec:null_model}, we fix $\alpha \in (0,1/2)$. We begin by observing that the existence of the weak limit $\fnull(\alpha)$ is enough to prove the first inequality of \eqref{eq:weak_law_null_bounds}. Towards this end, we first observe that there is a natural greedy algorithm \cite{diggavi2001transmission, mitzenmacher2009survey} for attempting to embed $Y$ into $X'$. Writing $\tau(i)$ for the position selected at step $i$ and setting $\tau(0) := 0$, we define $\tau(i)$ recursively as the least index $t > \tau(i-1)$ such that $X'_t = Y_i$ whenever such an index exists. If no such index exists at some step, we say that the greedy embedding algorithm fails. Then
\begin{align*}
    Z_{X',Y}=0 \qquad\Longleftrightarrow\qquad \text{the greedy embedding algorithm fails},
\end{align*}
since any $\sigma\in S_{X',Y}$ must satisfy $\tau(i) \le \sigma(i)$ at every step for which the greedy algorithm is defined. Equivalently, the event that $Z_{X',Y} = 0$ is exactly the event that
\begin{align*}
    \sum_{i=1}^M G_i > N,
\end{align*}
where $G_1,\dots,G_M$ are i.i.d. $\Geom(1/2)$ random variables. Here, $G_i$ corresponds to the number of bits of $X'$ that must be examined after $\tau(i-1)$ in order to locate $\tau(i)$. It easily follows from this latter representation that $Z_{X',Y} = 0$ with high probability for $\alpha > 1/2$, and that at $\alpha=1/2$, the event $Z_{X',Y} = 0$ has probability tending to $1/2$. This explains why the weak limit \eqref{eq:weak_law_null} is stated only for $\alpha \in (0,1/2)$.

We now fix $\sigma \in S_{X',Y}$, which corresponds to the following \emph{skip vector} $v \in \N^M$. For each $i \in [M]$, we let $v_i$ denote the number of instances of $Y_i$ strictly after the earliest instance of $Y_i$ following $X'_{\sigma(i-1)}$ (for $i = 1$, this is the earliest instance of $Y_1$ in $X'$) up to and including $X'_{\sigma(i)}$. We think of $v_i$ as the number of \textit{skips} that the $i$\textsuperscript{th} bit of $Y$ plays over the greedy algorithm when constructing $\sigma$. This mapping from $S_{X',Y}$ to skip vectors $v \in \N^M$ is evidently injective. Furthermore, for each $s \in \N$, we define
\begin{align*}
    \mathcal V_s := \left\{ v \in \N^M: \sum_{i=1}^M v_i = s \right\}.
\end{align*}
It similarly follows that for any $v \in \mathcal V_s$, the event that $v$ is a skip vector corresponding to a configuration in $S_{X',Y}$ is the event that the sum of $M+s$ i.i.d. $\Geom(1/2)$ random variables is at most $N$. Here, each $\Geom(1/2)$ random variable corresponds to the number of bits of $X'$ necessary to advance the candidate embedding with skip vector $v$ by one step. We conclude that
\begin{align}
    Z_{X',Y} & = \sum_{s=0}^{(1/\alpha-1)M} \sum_{v \in \mathcal V_s} \1\left\{ v \text{ corresponds to an elmt. of $S_{X',Y}$} \right\} \nonumber \\
    & \geq \sum_{s=0}^{\left(\frac{1}{2\alpha} - 1 \right)M} \sum_{v \in \mathcal V_s} \1\left\{ \sum_{i=1}^{M+s} \Geom(1/2) \leq N \right\} \geq e^{-o_p(N)} \sum_{s=0}^{\left(\frac{1}{2\alpha} - 1 \right)M} |\mathcal V_s| \approx e^{-o_p(N)} \Big|\mathcal V_{\left(\frac{1}{2\alpha} - 1 \right)M} \Big| \label{eq:lb_indicator_expression} \\
    & \approx e^{-o_p(N)} \exp\left(M \cdot \frac{h(2\alpha) }{2\alpha} \right) = e^{-o_p(N)} \exp\left(N \cdot \frac{h(2\alpha)}{2} \right), \nonumber
\end{align}
where the latter inequality of \eqref{eq:lb_indicator_expression} follows from Cram\'er's theorem applied to the $\Geom(1/2)$ law and Markov's inequality to control the number of ``bad summands." The lower bound of \eqref{eq:weak_law_null_bounds} follows.

\medskip

It remains to prove \eqref{eq:weak_law_null}. Our strategy is guided by the observation that $Z_{X',Y}$ exhibits an approximate superadditivity due to its self-replicating structure. For large integers $M_1$ and $M_2$, the inequality
\begin{align} \label{eq:superadditivity}
    \log Z_{X'_{1:N_1+N_2} Y_{1:M_1+M_2}} \geq \log Z_{X'_{1:N_1}, Y_{1:M_1}} + \log Z_{X'_{N_1+1:N_1+N_2}, Y_{M_1+1:M_1+M_2}},
\end{align}
holds whenever each partition function in \eqref{eq:superadditivity} is positive. We partition $S_{X',Y}$ by partitioning $X'$ and $Y$ into contiguous blocks and grouping embeddings according to how blocks of $Y$ are matched to blocks of $X'$. To handle dependencies across these classes, we restrict to an extremal class of this partition and argue that this choice is enough to fully capture the typical exponential behavior of $Z_{X',Y}$.

We define the collection of \emph{assignments}, corresponding to weak compositions of $N$ into $\sqrt{M}$ parts, via
\begin{align*}
    \mathcal A = \mathcal A_N := \left\{ \mu: [\sqrt{M}] \to [N] \ \Bigg| \  \sum_{i=1}^{\sqrt{M}} \mu(i) = N \right\}.
\end{align*}
Corresponding to $\mu \in \mathcal A$ is a decomposition of $X'$ given by
\begin{align*}
    X' = X_{\mu}'^{(1)}X_{\mu}'^{(2)}\cdots X_{\mu}'^{(\sqrt{M})}.
\end{align*}
For $i \in [\sqrt{M}]$, the \textit{block} $X_{\mu}'^{(i)}$ induced by $\mu$ is a contiguous sequence of $\mu(i)$ bits of $X'$, with the superscript respecting the order in which the bits appear in $Y$. Letting $\Bar{\nu}: [\sqrt{M}] \to [M]$ denote the map with $\Bar{\nu}(i) = \sqrt{M}$ for all $i \in [M]$, we may decompose $Y$ analogously. For $\mu \in \mathcal A$, the number of embeddings of $Y$ as a subsequence of $X'$ for which $Y_{\Bar{\nu}}^{(i)}$ is \textit{assigned} to $X_{\mu}'^{(i)}$ is then
\begin{align} \label{eq:mu_embeddings}
    Z_{X', Y}(\mu) := \prod_{i=1}^{\sqrt{M}} Z_{X_{\mu}'^{(i)}, Y_{\Bar{\nu}}^{(i)}}.
\end{align}
We focus on the following two kinds of assignments.
\begin{itemize}
    \item We let $\mu^*$ denote the (random) extremal assignment maximizing $Z_{X', Y}(\mu)$.
    \item We let $\Bar{\mu}$ denote a ``baseline'' for which $\Bar{\mu}(i) = N/\sqrt{M}$ for all $i \in [\sqrt{M}]$.
\end{itemize}
We work under the convention that $\log 0 := 0$. It will suffice to derive the weak law in this setting, as it can be easily shown that the probability that $Z_{X',Y} = 0$ vanishes. Altogether, we have that
\begin{align} \label{eq:Z_expression}
    Z_{X',Y} & = O\left( |\mathcal A| \right) \cdot Z_{X', Y}(\mu^*) = e^{\Tilde{O}\left(\sqrt{N} \right)} \prod_{i=1}^{\sqrt{M}} Z_{X_{\mu^*}'^{(i)}, Y_{\Bar{\nu}}^{(i)}}.
\end{align}
For $i \in [\sqrt{M}]$, the event $Z_{X_{\Bar{\mu}}'^{(i)}, Y_{\Bar{\nu}}^{(i)}} = 0$ corresponds to the event that a sum of $\sqrt{M}$ i.i.d. $\Geom(1/2)$ random variables is greater than $N/\sqrt{M} = \sqrt{M}/\alpha$. Thus, where the following \eqref{eq:mu_hat_union_bound} relies crucially on the assumption that $\alpha < 1/2$,
\begin{align} \label{eq:mu_hat_union_bound}
    \Prob\left(\bigcup_{i \in [\sqrt{M}]} \left\{ Z_{X_{\Bar{\mu}}'^{(i)}, Y_{\Bar{\nu}}^{(i)}} = 0 \right\} \right) \stackrel{\text{(Hoeffding)}}{=} e^{-\Omega(\sqrt{N})} \ll 1.
\end{align}
Rearranging \eqref{eq:Z_expression} and invoking \eqref{eq:mu_hat_union_bound} yields, where the $o_p(1)$ term is $O(1)$,
\begin{align} \label{eq:subsequence_equation_1}
    \frac{1}{M} \log Z_{X',Y} = \frac{1}{M} \sum_{i=1}^{\sqrt{M}} \log Z_{X_{\Bar{\mu}}'^{(i)}, Y_{\Bar{\nu}}^{(i)}} + \frac{1}{M}\underbrace{\left(\log Z_{X', Y}(\mu^*) - \log Z_{X', Y}(\Bar{\mu}) \right)}_{\geq 0} + \Tilde{O}\big( N^{-1/2} \big) + o_p(1).
\end{align} 

\begin{remark}
    We comment that fractional lengths (e.g., note that $\Bar{\nu}(i) = \sqrt{M}$ may not be an integer) are a benign issue throughout the argument and do not affect the asymptotics. The adjustment to the argument presented below is to take floors for every fractional length and to have the final multiplicand in \eqref{eq:mu_embeddings} become a larger ``remainder block." We omit the routine modifications here.
\end{remark}

We are now ready to prove \cref{lem:increasing_expectation}, an analogous convergence result in expectation.
\begin{lemma} \label{lem:increasing_expectation}
    The expression $\frac{1}{N} \E[\log Z_{X',Y}]$ converges to a constant.
\end{lemma}

\begin{proof}[Proof of \cref{lem:increasing_expectation}]
    It suffices to prove \cref{lem:increasing_expectation} when normalizing by $M$ instead of $N$. We first restrict our attention to those values of $M$ in the collection
    \begin{align*}
        \mathcal S := \{2, 4, 16, 256, \dots\}. 
    \end{align*}
    From \eqref{eq:mu_hat_union_bound} and \eqref{eq:subsequence_equation_1}, it follows for $M \geq 4$ from $\mathcal S$ that
    \begin{align} \label{eq:subsequence_ineq}
        \frac{1}{M}\E\left[\log Z_{X',Y}\right] \geq \frac{\sqrt{M}}{M}\E\left[\log Z_{X_{\Bar{\mu}}'^{(1)},\, Y_{\Bar{\nu}}^{(1)}}\right] - e^{-\Omega(\sqrt{N})} = \frac{1}{\sqrt{M}}\E\left[\log Z_{X'_{1:\sqrt{M}/\alpha}, Y_{1:\sqrt{M}}}\right] - e^{-\Omega(\sqrt{N})}.
    \end{align}
    Let $C, c > 0$ be such that the exponential term in the RHS of \eqref{eq:subsequence_ineq} is at most $Ce^{-c\sqrt{M}}$ for all $M$ considered. It follows that for $M \in \mathcal S$, the sequence
    \begin{align*}
        \frac{1}{M}\E\left[\log Z_{X',Y}\right] + \sum_{M' \in \mathcal S: M' \leq M} Ce^{-c\sqrt{M'}}
    \end{align*}
    is increasing. This sequence is also uniformly bounded since $e^{-\Omega(\sqrt{N})}$ is summable and
    \begin{align*}
        \frac{1}{M} \E\left[ \log Z_{X',Y} \right] \leq \frac{1}{M} \log \binom{N}{M} = O(1),
    \end{align*}
    so it converges. We conclude that the sparse subsequence $\frac{1}{M}\E[\log Z_{X',Y}]$ on $M \in \mathcal S$ converges.
    
    \smallskip
    
    We now extend the result to all values of $M$. We define the iterated square root
    \begin{align*}
        \phi(M) = \begin{cases}
        0 & M \leq 2 \\
        1 + \phi(\sqrt{M}) & M > 2
    \end{cases}
    \end{align*}
    and we let 
    \begin{align*}
        & \ell = \ell(M) := 2^{2^{\phi(M)-2}} \in \mathcal S;
        & L = L(M) := M/\ell(M) \geq \sqrt{M}.
    \end{align*}
    We define $\mu \in \mathcal A_{N}$ and $\nu \in \mathcal A_M$ such that for all $i \in [L-1]$, we have that $\mu(i) = \ell/\alpha$ and $\nu(i) = \ell$. We decompose $X'$ and $Y$ via
    \begin{align*}
        & X' = X_\mu'^{(1)}X_\mu'^{(2)}\cdots X_\mu'^{(L)};
        & Y = Y_{\nu}^{(1)}Y_{\nu}^{(2)}\cdots Y_{\nu}^{(L)}.
    \end{align*}
    By considering those configurations of $S_{X',Y}$ where $Y_{\nu}^{(i)}$ is assigned to $X_{\mu}'^{(i)}$ for $i \in [L]$, it follows that
    \begin{align}
        \frac{1}{M} \E\left[ \log Z_{X',Y} \right] & \geq \frac{1}{M}\sum_{i=1}^L \E\left[ \log Z_{X_\mu'^{(i)},\, Y_{\nu}^{(i)}} \right] = \frac{1}{\ell L} \sum_{i=1}^L \E\left[ \log Z_{X_\mu'^{(i)},\, Y_{\nu}^{(i)}} \right] \nonumber \\
        & = \frac{1}{\ell} \E\left[ \log Z_{X_\mu'^{(1)},\, Y_{\nu}^{(1)}} \right] + \frac{1}{M} \E\left[\log Z_{X_\mu'^{(L)},\, Y_{\nu}^{(L)}} - \log Z_{X_{\mu}'^{(1)},\, Y_{\nu}^{(1)}} \right] \nonumber \\
        & = \frac{1}{\ell}\E\left[ \log Z_{X'_{1:\ell/\alpha},\,Y_{1:\ell}} \right] + o(1). \label{eq:interpolation_lower_bound}
    \end{align}
    On the other hand, we define 
    \begin{align*}
        & u = u(M) := 2^{2^{\phi(M)+1}} \in \mathcal S;
        & U = U(M) := u(M)/M \geq M.
    \end{align*}
    We define $\Tilde{\mu} \in \mathcal A_{N}$ and $\Tilde{\nu} \in \mathcal A_M$ so for all $i \in [U-1]$, we have that $\Tilde{\mu}(i) = N$ and $\Tilde{\nu}(i) = M$. We decompose $X'_{1:u/\alpha}$ and $Y_{1:u}$ via
    \begin{align*}
        & X'_{1:u/\alpha} = X_{\Tilde{\mu}}'^{(1)}X_{\Tilde{\mu}}'^{(2)}\cdots X_{\Tilde{\mu}}'^{(U)};
        & Y_{1:u} = Y_{\Tilde{\nu}}^{(1)}Y_{\Tilde{\nu}}^{(2)}\cdots Y_{\Tilde{\nu}}^{(U)}.
    \end{align*}
    By considering those configurations of $S_{X'_{1:u/\alpha},Y_{1:u}}$ for which $Y_{\Tilde{\nu}}^{(i)}$ is assigned to $X_{\Tilde{\mu}}'^{(i)}$ for $i \in [U]$, it follows that
    \begin{align}
        \frac{1}{u} \E\left[ \log Z_{X'_{1:u/\alpha},\,Y_{1:u}} \right] & \geq \frac{1}{u}\sum_{i=1}^U \E\left[ \log Z_{X_{\Tilde{\mu}}'^{(i)},\, Y_{\Tilde{\nu}}^{(i)}} \right] = \frac{1}{M} \cdot \frac{1}{U}\sum_{i=1}^U \E\left[ \log Z_{X_{\Tilde{\mu}}'^{(i)},\, Y_{\Tilde{\nu}}^{(i)}} \right] \nonumber \\
        & = \frac{1}{M}\E\left[ \log Z_{X_{\Tilde{\mu}}'^{(1)},\, Y_{\Tilde{\nu}}^{(1)}} \right] + \frac{1}{U} \cdot \frac{1}{M} \E\left[ \log Z_{X_{\Tilde{\mu}}'^{(U)},\, Y_{\Tilde{\nu}}^{(U)}} - \log Z_{X_{\Tilde{\mu}}'^{(1)},\, Y_{\Tilde{\nu}}^{(1)}} \right] \nonumber  \\
        & = \frac{1}{M}\E[\log Z_{X', Y}] + o(1). \label{eq:interpolation_upper_bound}
    \end{align}
    Altogether, we have that
    \begin{align*}
        \frac{1}{\ell} \E\left[\log Z_{X'_{1:\ell/\alpha},\,Y_{1:\ell}}\right] + o(1) \stackrel{\eqref{eq:interpolation_lower_bound}}{\leq} \frac{1}{M}\E\left[ \log Z_{X',Y} \right] \stackrel{\eqref{eq:interpolation_upper_bound} }{\leq} \frac{1}{u} \E\left[\log Z_{X'_{1:u/\alpha},\,Y_{1:u}}\right] + o(1),
    \end{align*}
    from which we conclude that $\frac{1}{M}\E[\log Z_{X',Y}]$ converges.
\end{proof}

\begin{proof}[Proof of Equation \eqref{eq:weak_law_null}]
    Taking expectations in \eqref{eq:subsequence_equation_1} with normalization by $N$ yields
    \begin{align*}
        \frac{1}{N} \E\left[ \log Z_{X',Y} \right] = \frac{1}{\sqrt{M}/\alpha} \E\left[ \log Z_{X'_{1:\sqrt{M}/\alpha},\, Y_{1:\sqrt{M}}} \right] + \E\left[\frac{1}{N}\left(\log Z_{X', Y}(\mu^*) - \log Z_{X', Y}(\Bar{\mu}) \right)\right] + o(1).
    \end{align*}
    \cref{lem:increasing_expectation} and Markov's inequality now imply that
    \begin{align} \label{eq:discrepancy_from_max}
        0 \leq \frac{1}{N}\left(\log Z_{X', Y}(\mu^*) - \log Z_{X', Y}(\Bar{\mu}) \right) = o_p(1).
    \end{align}
    Therefore, we may write \eqref{eq:subsequence_equation_1} with normalization by $N$ as
    \begin{align} \label{eq:subsequence_equation_3}
        \frac{1}{N}\log Z_{X',Y} = \frac{1}{\sqrt{N}} \sum_{i=1}^{\sqrt{M}} \frac{1}{\sqrt{N}} \log Z_{X_{\Bar{\mu}}'^{(i)},\, Y_{\Bar{\nu}}^{(i)}} + o_p(1),
    \end{align}
    with the $o_p(1)$ term being $O(1)$. Since $\Bar{\mu}$ and $\Bar{\nu}$ are deterministic, the $\sqrt{M}$ pairs 
    \begin{align*}
        \big( X_{\Bar{\mu}}'^{(i)},\, Y_{\Bar{\nu}}^{(i)} \big)
    \end{align*}
    are independent, so the corresponding summands in \eqref{eq:subsequence_equation_3} are also independent. We conclude that
    \begin{align} 
        \Var\left(\frac{1}{N}\log Z_{X',Y}\right) & = \Var\left( \frac{1}{\sqrt{N}} \sum_{i=1}^{\sqrt{M}} \frac{1}{\sqrt{N}} \log Z_{X_{\Bar{\mu}}'^{(i)},\, Y_{\Bar{\nu}}^{(i)}} \right) + o(1) \nonumber \\
        & = \frac{\sqrt{M} \cdot \Var\left( \frac{1}{\sqrt{N}} \log Z_{X_{\Bar{\mu}}'^{(1)},\, Y_{\Bar{\nu}}^{(1)}} \right)}{N} + o(1) = O\big( N^{-1/2} \big) + o(1) \ll 1. \label{eq:variance_computation}
    \end{align}
    The desired weak law now follows from  \cref{lem:increasing_expectation}, \eqref{eq:variance_computation}, and Chebyshev's inequality.
\end{proof}

\subsection{Planted model}

We now prove the weak limit part of \eqref{eq:weak_law_planted_bounds}, the analogue of \eqref{eq:weak_law_null} for the planted setting. We begin by proving \cref{prop:subset_weak_limit_planted_deletion_channel}, a deletion channel variant of this weak limit which settles \cite[Conjecture 1]{pernice2024mutual}.

\begin{proposition}[{Deletion channel planted limit; \cite[Conjecture 1]{pernice2024mutual}}] \label{prop:subset_weak_limit_planted_deletion_channel}
    Fix $\alpha \in (0,1)$. Let $\BDC_{1-\alpha}(X)$ denote the outcome of passing $X \in \{0,1\}^N$ through a deletion channel with deletion probability $1 - \alpha$. There exists a constant $\fpl(\alpha) > 0$ such that
    \begin{align} \label{eq:random_subset_weak_law}
        \frac{1}{N} \log Z_{X,\, \BDC_{1-\alpha}(X)} \xrightarrow{a.s.} \fpl(\alpha).
    \end{align}
\end{proposition}

\begin{proof}
    It is clear that for any $N_1, N_2 \in \N$, we have that
    \begin{align}
        & \log Z_{X_{1:N_1+N_2},\, \BDC_{1-\alpha}(X_{1:N_1+N_2})} \nonumber \\
        & \qquad \geq \log Z_{X_{1:N_1},\, \BDC_{1-\alpha}(X_{1:N_1})} + \log Z_{X_{N_1+1:N_1+N_2},\, \BDC_{1-\alpha}(X_{N_1+1:N_1+N_2})}. \label{eq:planted_superadditive}
    \end{align}
    Specifically, \eqref{eq:planted_superadditive} holds for the following reason. The partition function in the LHS of \eqref{eq:planted_superadditive} counts unconstrained subsequence embeddings of $\BDC_{1-\alpha}(X_{1:N_1+N_2})$ into $X_{1:N_1+N_2}$. On the other hand, the expression on the RHS counts the logarithm of the number of constrained such subsequence embeddings for which $\BDC_{1-\alpha}(X_{1:N_1})$ is mapped into $X_{1:N_1}$ and $\BDC_{1-\alpha}(X_{N_1+1:N_1+N_2})$ is mapped into $X_{N_1+1:N_1+N_2}$. If one of the partition functions of \eqref{eq:planted_superadditive} vanishes, then $\BDC_{1-\alpha}$ has deleted every bit of the corresponding string. In that case the relevant logarithmic term is $0$ by convention, and \eqref{eq:planted_superadditive} is readily checked to remain valid. Therefore, invoking the superadditive ergodic theorem \cite{kingman1973subadditive} on the integrable random variables
    \begin{align*}
        g_N := \log Z_{X_{1:N},\, \BDC_{1-\alpha}(X_{1:N})},
    \end{align*}
    noting that the corresponding transformation is ergodic as it is effectively a Bernoulli shift, yields the desired.
\end{proof}

\begin{proof}[Proof of weak limit of \eqref{eq:weak_law_planted_bounds}]
    We couple the random string $\BDC_{1-\alpha}(X)$ with $Y$ by defining $Y$ via inserting or deleting $|M - \alpha N|$ bits uniformly at random. If the planted string $Y'$ of $X$ is obtained from the planted string $Y$ of $X$ by including a bit of $X$ in its appropriate position, then it holds that
    \begin{align*}
        |\log Z_{X,Y'} - \log Z_{X,Y}| \leq N,
    \end{align*}
    so the log-partition function increases by an additive margin of at most $\log N$. Indeed, this crude bound follows via choosing one of the $N$ bits of $X$ that the new bit is mapped to when forming $Y'$ from $Y$, as the remaining bits of $Y$ must still correspond to a valid subsequence embedding. This observation, together with \cref{prop:subset_weak_limit_planted_deletion_channel} and the standard fact that 
    \begin{align*}
        |\BDC_{1-\alpha}(X)| \sim \Bin(N, \alpha)
    \end{align*}
    concentrates about $\alpha N = M$ with $O(\sqrt{N}) \ll N$ deviations is now enough to derive the desired result. We leave the straightforward details to the reader.
\end{proof}

\begin{remark}
    A weaker version of \cref{thm:weak_law} in which all inequalities are non-strict now follows. Indeed, since the sequence $\frac{1}{N} \log Z_{X',Y}$ is uniformly bounded, it follows that
    \begin{align} \label{eq:null_jensen}
        \fnull(\alpha) = \lim_{N \to \infty} \E\left[ \frac{1}{N} \log Z_{X',Y} \right] \stackrel{\text{(Jensen)}}{\leq} \lim_{N \to \infty} \frac{1}{N}  \log \E\left[ Z_{X',Y} \right] = \fnullann(\alpha).
    \end{align}
    Towards proving the other non-strict inequality, we let $\sigma^* \in \Sigma_{N,M}$ denote the planted embedding, so that
    \begin{align} \label{eq:planted_subsequence}
        Y = X_{\sigma^*}.
    \end{align}
    We let $x \in \{0,1\}^M$ be an independent uniform random binary string. For $z \in \N$, via explicitly pinning down the collection $A$ of $M$-subsequences of $X$ corresponding to embeddings of $Y$, we can write
    \begin{align}
        & \Prob\left( Z_{X,Y} = z \right) = \E_{\sigma^*,X} \left[ \sum_{\substack{A \subseteq \binom{[N]}{M}: \\ |A| = z}} \1\left\{ X|_S = Y \ \text{for all} \ S \in A; \ X|_S \neq Y \ \text{for all} \ S \notin A \right\} \right] \nonumber \\
        & \quad = 2^M \E_{\sigma^*, X, x} \left[ \sum_{\substack{A \subseteq \binom{[N]}{M} \\ |A| = z}} \1\left\{ X|_S = x \ \text{for all} \ S \in A; \ X|_S \neq x \ \text{for all} \ S \notin A \right\} \1\left\{ Y = x \right\} \right] \nonumber \\
        & \quad \stackrel{\eqref{eq:planted_subsequence}}{=} \frac{2^M}{\binom{N}{M}} \E_{X, x} \left[ \sum_{\sigma^* \in \Sigma_{N,M}} \sum_{\substack{A \subseteq \binom{[N]}{M} \\ |A| = z}} \1\left\{ X|_S = x \ \text{for all} \ S \in A; \ X|_S \neq x \ \text{for all} \ S \notin A \right\} \1\left\{ X_{\sigma^*} = x \right\} \right] \nonumber \\
        & \quad = \frac{2^M}{\binom{N}{M}} \E_{X, x} \left[ Z_{X,x} \cdot \1\left\{Z_{X,x} = z \right\} \right] = \frac{2^M}{\binom{N}{M}} \E\left[ Z_{X',Y} \cdot \1\left\{Z_{X',Y} = z \right\} \right]. \label{eq:planted_point_prob}
    \end{align}
    We remark that the above calculation is an application of Nishimori's identity. Indeed, the planted law is obtained from the null law by reweighting with the partition function. Concretely, the distribution of $Z_{X,Y}$ under the planted model is the $Z_{X',Y}$-size-biased tilt of its distribution under the null model. Altogether, we conclude that
    \begin{align}
        \E\left[ \log Z_{X,Y} \right] & = \sum_{z \geq 0} \log z \cdot \Prob\left( Z_{X,Y} = z \right) \stackrel{\eqref{eq:planted_point_prob}}{=} \frac{2^M}{\binom{N}{M}} \sum_{z \geq 0} \log z \cdot \E\left[ Z_{X',Y} \cdot \1\left\{Z_{X',Y} = z \right\} \right] \nonumber \\
        & = \frac{2^M}{\binom{N}{M}} \E\left[ Z_{X',Y} \log Z_{X',Y} \right] \stackrel{\text{(Jensen)}}{\geq} \frac{2^M}{\binom{N}{M}} \E\left[ Z_{X',Y} \right] \log \E\left[ Z_{X',Y} \right] = \log \E\left[ Z_{X',Y} \right]. \label{eq:planted_quenched_jensen}
    \end{align}
    We establish that the Jensen gaps in \eqref{eq:null_jensen} and \eqref{eq:planted_quenched_jensen} are nontrivial in forthcoming sections.
\end{remark}

\section{Proof of the Quenched-Annealed Gaps} \label{sec:main_result_upper_bound}

\subsection{Definitions and notation} \label{subsec:defns_and_notation}

We begin by recording the conventions that will remain in force unless explicitly stated otherwise. Throughout, lowercase letters, such as $x \in \{0,1\}^N$ and $y \in \{0,1\}^M$, denote deterministic binary strings to distinguish them from random strings, which are denoted using uppercase letters. We also fix the density parameter $\alpha \in (0,1)$ and the block length $b \in \N$. The quantity $\alpha b$ should be interpreted as the typical block length of an \textit{$x$-random string}, i.e., a uniformly chosen length-$M$ subsequence of $x$. In particular, in the planted variant of the Random Subsequence Model, $Y$ is an $X$-random string. We assume that $N$ is a large positive integer which tends to infinity. We occasionally suppress the dependence of certain quantities on $N$ in our notation --- this should never raise any confusion. 

For convenience in the argument there, throughout \cref{subsec:alignment_W_random_Y} (namely the proof of \cref{prop:regular_alignment_property}), we proceed with the understanding that we are given a fixed typical (as defined in \cref{defn:typical_W}) deterministic string $x \in \{0,1\}^N$ and work with the decomposition
\begin{align} \label{eq:X_decomp}
    x = \big( x^{(1)}, \dots, x^{(B)} \big),
\end{align}
studying how an $x$-random string aligns with the structure of $x$. We begin by introducing shorthand for the corresponding ``planted'' probability measure induced by a deterministic string.
\begin{definition}[$x$-planted measure] \label{defn:planted_prob_msr}
    For a fixed string $x \in \{0,1\}^N$, we let the \emph{$x$-planted measure} $\Prob_x$ denote the probability measure on $\bigcup_{k=0}^{N} \{0,1\}^k$ corresponding to including each bit of $x$ independently with probability $\alpha$. Specifically, for $0 \leq k \leq N$ and $y \in \{0,1\}^k$, we define
    \begin{align*}
        \Prob_x\!\left( y \right) := \alpha^k (1-\alpha)^{N-k} Z_{x,y}.
    \end{align*}
\end{definition}

\begin{remark} \label{rmk:V_random_msr}
    Given $x \in \{0,1\}^N$, it is clear that
    \begin{align} \label{eq:x_random_string_law}
        \Prob_x\big( \cdot \mid \{0,1\}^M \big) |_{\{0,1\}^M}
    \end{align}
    denotes the law of an $x$-random string, while a local limit theorem (e.g., see \cite[Theorem 3.5.3]{durrett2019probability}) yields 
    \begin{align*}
        \Prob_x\big( \{0,1\}^M \big) = \Theta\big( N^{-1/2} \big).
    \end{align*}
    Thus, we have that
    \begin{align} \label{eq:V_random_msr_bayes}
        \Prob_x\big( \cdot \mid \{0,1\}^M \big) \leq \Theta\big( N^{-1/2} \big)\Prob_x\!\left( \cdot \right).
    \end{align}
    As we strictly concern ourselves with the exponential orders of rare events under \eqref{eq:x_random_string_law}, it suffices to work with the unconditional measure $\Prob_x$.
\end{remark}

Next, we introduce those parameters needed to derive concentration guarantees at the desired level of granularity. We take $\epsilon > 0$ to be some small fixed constant ($\epsilon = 1/24$ suffices for the argument to hold), and we introduce shorthand for
\begin{align} \label{eq:window_parameters}
    & \delta := b^{-1/2+\epsilon};
    & \gamma := b^{-\epsilon}.
\end{align}
We now elaborate on the specific roles that these quantities play over the course of our proof. We recall that our aim is to demonstrate a distinction between samples $(X,Y)$ drawn from the null and the planted variants of the Random Subsequence Model. Loosely, we will show that near-equipartitions into $B$ blocks of $Y$ drawn from the planted model resemble the structure of the corresponding block decomposition of $X$, while there is no near-equipartition of $Y$ drawn from the null model for which such a resemblance is attained. We proceed with the relevant definitions, which crucially rely upon our choices of the parameters introduced in \eqref{eq:window_parameters}.

\begin{definition}[Induced and standardized near-equipartitions] \label{defn:near_equipartitions}
    Given $y \in \{0,1\}^*$, the collection of \emph{induced near-equipartitions} of $y$ is
    \begin{align*}
        \mathcal{NE}_\ind\!\left( y \right) := \left\{ \big(y^{(1)}, \dots, y^{(B)}\big) \in \left(\{0,1\}^* \right)^B \,\middle|\,
        \begin{array}{c}
            y = y^{(1)}\cdots y^{(B)}; \\ 
            |y^{(i)}| \in \{0,1,\dots,b\} \text{ for all } i \in [B]; \\
            \left| \left\{ i \in [B] : |y^{(i)}| \notin \left[ (1-\delta)\alpha b, (1+\delta)\alpha b \right] \right\} \right| \leq \gamma B
        \end{array} \right\}.
    \end{align*}
    On the other hand, the collection of \emph{standardized near-equipartitions} of $y$ is
    \begin{align*}
        \mathcal{NE}_\std\!\left( y \right) := \left\{ \big(y^{(1)}, \dots, y^{(B)}\big) \in \left(\{0,1\}^* \right)^B \,\middle|\,
        \begin{array}{c}
            y = y^{(1)}\cdots y^{(B)}; \\
            |y^{(i)}| \in \{0,1,\dots,b\} \text{ for all } i \in [B]; \\
            \left| \left\{ i \in [B] : |y^{(i)}| \neq \alpha b \right\} \right| \leq 3\gamma B
        \end{array} \right\}.
    \end{align*}
\end{definition}

\begin{remark}
    Although we generally suppress floors and ceilings throughout the analysis, \cref{defn:near_equipartitions} is one location where a brief clarification is helpful. When $\alpha b \notin \mathbb N$, the collection of standardized near-equipartitions $\mathcal{NE}_\std(y)$ should be defined by assigning either $\lfloor \alpha b\rfloor$ or $\lceil \alpha b\rceil$ to each superscript $i$. This is done in such a way that for every $k$, the sum of the first $k$ block lengths differs from $\alpha b k$ by at most $1$. Since $\alpha b$ is fixed, such a choice can be made once and for all. These adjustments do not affect the asymptotic estimates in the proof, and we suppress the corresponding routine modifications, both here and in all other instances where floors and ceilings are omitted.
\end{remark}

Induced near-equipartitions of $y$ correspond to ways of partitioning $y$ into $B$ blocks so that very few of these blocks are multiplicatively far from $\alpha b$, and they should (in light of \cref{rmk:V_random_msr}) importantly be thought of as capturing the \textit{induced} block structure of a typical outcome of a string drawn from the planted model. On the other hand, owing to our choice of definitions and the strict restriction that ``good blocks'' $y^{(i)}$ must have size $\alpha b$, the collection of standardized near-equipartitions of $y$ is far smaller. This will be crucial in \cref{subsec:alignment_impossibility_ind_Y}, where we establish an approximation-type result for the following notion of total agreement when substituting $\mathcal{NE}_\ind(Y)$ for $\mathcal{NE}_\std(Y)$ via the \textit{standardization} algorithm.

\begin{definition}[Induced and standardized total alignment] \label{defn:alignment_fcns}
    Given $y \in \{0,1\}^*$, we respectively define its
    \emph{induced total alignment} and its \emph{standardized total alignment} with the string $x = \big(x^{(1)}, \dots, x^{(B)}\big) \in \{0,1\}^{N}$ via 
    \begin{align}
        & T_\ind\!\left( x, y \right) := \sup_{\left( y^{(1)}, \dots, y^{(B)} \right) \in \mathcal{NE}_\ind\left( y \right)} \frac{1}{B} \sum_{i=1}^B A_\loc\big( x^{(i)}, y^{(i)}\big); \label{eq:total_alignment_first_kind_defn} \\
        & T_\std\!\left( x, y \right) := \sup_{\left( y^{(1)}, \dots, y^{(B)} \right) \in \mathcal{NE}_\std\left( y \right)} \frac{1}{B} \sum_{i=1}^B A_\loc\big( x^{(i)}, y^{(i)}\big). \label{eq:total_alignment_second_kind_defn}
    \end{align}
    The individual \emph{local alignment} terms (i.e., the summands) on the RHS of \eqref{eq:total_alignment_first_kind_defn} and \eqref{eq:total_alignment_second_kind_defn} are defined via
    \begin{align} \label{eq:local_alignment_defn}
         & A_\loc\big( x^{(i)}, y^{(i)} \big) := \begin{cases}
             0 & \maj(x^{(i)}) \neq \maj(y^{(i)}) \\
             1 \land \delta \Delta(y^{(i)}) & \maj(x^{(i)}) = \maj(y^{(i)})
         \end{cases};
         & \Delta(y^{(i)}) := \left| \sum_{j=1}^{|y^{(i)}|} \left( 2( y^{(i)})_j - 1 \right) \right|,
    \end{align}
    where $\maj(z) \in \{0,1\}$ denotes the majority bit of $z \in \{0,1\}^*$, taken to be $1$ in the case of a tie.
\end{definition}
Corresponding the binary string $y^{(i)}$ to a realization of a simple random walk with Rademacher increments in the natural way, we note that the expression $\Delta(y^{(i)})$ in \eqref{eq:local_alignment_defn} is the absolute value of the random walk at time $|y^{(i)}|$. In particular, we may equivalently write the local alignment via (with $\sign(0) = 1$)
\begin{align} \label{eq:local_alignment_SRW}
    A_\loc\big(x^{(i)}, y^{(i)}\big) = \begin{cases}
         0 & \sign\left( \sum_{j=1}^{|x^{(i)}|} \left( 2( x^{(i)} )_j - 1 \right) \right) \neq \sign\left( \sum_{j=1}^{|y^{(i)}|} \left( 2( y^{(i)})_j - 1 \right) \right); \\
         1 \land \delta \Delta(y^{(i)}) & \sign\left( \sum_{j=1}^{|x^{(i)}|} \left( 2( x^{(i)} )_j - 1 \right) \right) = \sign\left( \sum_{j=1}^{|y^{(i)}|} \left( 2( y^{(i)} )_j - 1 \right) \right).
     \end{cases}
\end{align}

We mention in passing that the notions of local alignment and total alignment that we introduce here are loosely reminiscent of recent ideas in the trace reconstruction literature. For instance, \cite{holden2018subpolynomial} introduced robust alignment tests motivated by this correspondence, together with the appropriate scaling needed to make the consequences of such tests transparent.

Our next definition introduces the collection of $N$-bit strings that we restrict our attention to.
\begin{definition}[Typical ambient strings] \label{defn:typical_W}
    We say that $x \in \{0,1\}^{N}$ is \emph{typical} if at least $B/10$ of the blocks $x^{(i)}$ are such that $\Delta(x^{(i)}) \geq \sqrt{b}$.
\end{definition}
It is easy to justify the use of the term ``typical'' in \cref{defn:typical_W}. Indeed, the central limit theorem yields
\begin{align} \label{eq:typical_block_cond}
    \frac{1}{\sqrt{b}} \Delta(X^{(i)}) = \frac{1}{\sqrt{b}} \sum_{j=1}^{b} \left( 2( X^{(i)} )_j - 1 \right) \xrightarrow{d} \mathcal N(0,1) \implies \Prob\left( \Delta(X^{(i)}) \geq \sqrt{b} \right) \stackrel{\text{($b$ large)}}{>} \frac{3}{10},
\end{align}
from which it follows that
\begin{align*}
    \Prob\left( \sum_{i=1}^B \1\left\{ \Delta(X^{(i)}) \geq \sqrt{b} \right\} \leq \frac{B}{10} \right) & \stackrel{\eqref{eq:typical_block_cond}}{\leq} \Prob\left( \sum_{i=1}^B \1\left\{ \Delta(X^{(i)}) \geq \sqrt{b} \right\} - \Prob\left( \Delta(X^{(i)}) \geq \sqrt{b} \right) \leq -\frac{B}{5} \right) \\
    & \stackrel{\text{(Hoeffding)}}{\leq} e^{-\Omega(N)},
\end{align*}
so that a uniform random string $X \in \{0,1\}^N$ is typical, in the sense of \cref{defn:typical_W}, with probability at least $1 - e^{-\Omega(N)}$. We conclude this section by pinning down our key notion of structural resemblance with a typical string $x \in \{0,1\}^N$. We let 
\begin{align} \label{eq:alignment_constant}
    \beta^\star(\alpha) := \frac{\beta(\alpha)}{40} := \frac{1}{40} \left[ \Prob\left( \mathcal N\left( \alpha, \alpha\left( 1 - \alpha \right) \right) \geq 0 \right) - \frac{1}{2} \right] > 0;
\end{align}
we clarify in the forthcoming analysis how this constant shows up.
\begin{definition}[Aligned and good sets]
    Fix $x \in \{0,1\}^{N}$. The \emph{$x$-aligned set} $\mathcal A(x) \subseteq \{0, 1\}^*$ is the collection of all $y \in \{0,1\}^*$ for which 
    \begin{align*}
        T_\ind(x, y) \geq 1/2 + \beta^\star(\alpha).
    \end{align*}
    The \emph{$x$-good set} $\mathcal G(x) \subseteq \{0,1\}^M$ is the restriction of $\mathcal A(x)$ to $\{0,1\}^M$, i.e.,
    \begin{align} \label{eq:good_set_defn}
        \mathcal G(x) := \mathcal A(x) \cap \left\{ 0,1 \right\}^M.
    \end{align}
\end{definition}
In \cref{subsec:alignment_W_random_Y}, we show that for $(X,Y)$ drawn from the planted model, $Y$ will overwhelmingly land in $\mathcal G(X)$. On the other hand, we show in \cref{subsec:alignment_impossibility_ind_Y} that for $(X',Y)$ drawn from the null model, $Y$ overwhelmingly fails to land in $\mathcal G(X')$. This distinction will then lead to the proof of \cref{thm:weak_law}.

\subsection{Aligned structure under planted measures} \label{subsec:alignment_W_random_Y}

With the machinery of \cref{subsec:defns_and_notation} in place, we are now ready to introduce the regular alignment property, which is the key notion of structural alignment between binary strings which we exploit to prove the upper bound of \cref{thm:weak_law}. In the present \cref{subsec:alignment_W_random_Y}, we handle the asymptotic behavior of $x$-random strings for typical $x \in \{0,1\}^N$. Here, we demonstrate via standard concentration arguments that with probability $1 - e^{-\Omega(N)}$, an $x$-random string has large induced total alignment with $x$. Indeed, this is largely as expected --- it is natural to suspect that $x$-random strings should resemble $x$ with high probability, and \cref{prop:regular_alignment_property} establishes induced total alignment as one such appropriate notion of resemblance. We move on to study the inverse problem in \cref{subsec:alignment_impossibility_ind_Y}, in which we consider the same problem for random strings $(X,Y)$ drawn from the null model and obtain the exact opposite result.

\begin{proposition}[Regular alignment property] \label{prop:regular_alignment_property}
    Uniformly over all typical $x \in \{0,1\}^{N}$,
    \begin{align} \label{eq:regular_alignment_property_statement}
        \binom{N}{M}^{-1} \left| \left\{ S \in \binom{[N]}{M}: x|_S \notin \mathcal G(x) \right\} \right| = \Prob_x\!\left( \{0,1\}^M \setminus \mathcal G(x) \mid \{0,1\}^M \right) \leq e^{-\Omega(N)}.
    \end{align}
\end{proposition}

We observe that the equality of \eqref{eq:regular_alignment_property_statement} is an immediate consequence of \cref{defn:planted_prob_msr}. To establish the inequality of \eqref{eq:regular_alignment_property_statement}, it suffices to show that uniformly over all typical $x \in \{0,1\}^{N}$, it holds that
\begin{align} \label{eq:regular_alignment_aligned_set}
    \Prob_x\!\left( \mathcal A(x)^c \right) = e^{-\Omega(N)}.
\end{align}
Indeed, the desired result would then readily follow via
\begin{align*}
    & \Prob_x\!\left( \{0,1\}^M \setminus \mathcal G(x) \mid \{0,1\}^M \right) \stackrel{\eqref{eq:good_set_defn}}{=} \Prob_x\!\left( \mathcal A(x)^c \cap \{0,1\}^M \mid \{0,1\}^M \right) \\
    & \qquad = \Prob_x\!\left( \mathcal A(x)^c \mid \{0,1\}^M \right) \stackrel{\eqref{eq:V_random_msr_bayes}}{\leq} \Theta( N^{1/2} ) \Prob_x\!\left( \mathcal A(x)^c \right) \stackrel{\eqref{eq:regular_alignment_aligned_set}}{\leq} \Theta( N^{1/2} ) e^{-\Omega(N)} = e^{-\Omega(N)}.
\end{align*}
The remainder of the proof of \cref{prop:regular_alignment_property} involves routine techniques from probabilistic combinatorics. The work lies not in developing novel ideas, but in adapting these arguments to the framework and definitions introduced in \cref{subsec:defns_and_notation} (which are necessary to carry out the argument of \cref{subsec:alignment_impossibility_ind_Y}). For this reason, we defer the proof to \cref{app:RAP_proof}.

\subsection{Failure of biased alignment under the null model} \label{subsec:alignment_impossibility_ind_Y}

\cref{prop:regular_alignment_property} shows that, for any typical ambient string $x$, an $x$-random string lies in $\mathcal G(x)$ with overwhelming probability. In this subsection we prove the complementary statement under the null model. Specifically, for $(X',Y) \in \{0,1\}^{N} \times \{0,1\}^M$ drawn from the null model, we show that with overwhelming probability,
\begin{align*}
    Y \notin \mathcal G(X').
\end{align*}
This result produces the desired separation, which we exploit in \cref{subsec:proof_of_thm} to prove \cref{thm:weak_law}. 

As usual, we decompose
\begin{align*}
    X' = \big( X'^{(1)}, \dots, X'^{(B)} \big)
\end{align*}
via equipartitioning $X'$ into $B$ contiguous blocks of length $b$. Since $Y$ is uniformly distributed on $\{0,1\}^M$, we heuristically expect that for any fixed induced near-equipartition $\big(Y^{(1)},\dots,Y^{(B)}\big)\in\mathcal{NE}_\ind(Y)$, the average
\begin{align*}
    \frac{1}{B}\sum_{i=1}^B A_{\loc}\big(X'^{(i)},Y^{(i)}\big)
\end{align*}
should concentrate near $1/2$. The difficulty is to make such a bound uniform over all induced near-equipartitions of $Y$. Indeed, the class $\mathcal{NE}_{\ind}(Y)$ is too large for the most naive approaches (for instance, a union bound over an $\epsilon$-net defined in terms of block endpoints) to yield guarantees of sufficient strength. To overcome this, we show that induced total alignment can be approximated by standardized total alignment, in which we restrict most blocks to have exactly the same average length $\alpha b$ at the expense of permitting more blocks which deviate from this length. 

We begin by establishing that this approach can yield a bound of the desired form.

\begin{proposition}[Standardized alignment bound] \label{prop:total_alignment_2_bd}
    It holds with probability $\geq 1-e^{-\Omega(N)}$ that
    \begin{align*}
        T_\std\!\left( X', Y \right) < \frac{1+\beta^\star(\alpha)}{2}.
    \end{align*}
\end{proposition}

\begin{proof}
    It readily follows from \cref{defn:near_equipartitions} that
    \begin{align} \label{eq:ne_2_size_bd}
        \left|\mathcal{NE}_\std\!\left( Y \right) \right| \leq \binom{B}{3\gamma B} (b+1)^{3\gamma B} \stackrel{\eqref{eq:window_parameters}}{\approx} \exp\left( B \cdot h(3b^{-\epsilon}) + B \cdot 3b^{-\epsilon} \log (b+1) \right) = e^{B \cdot o_b(1)}.
    \end{align}
    Furthermore, \cref{defn:alignment_fcns} also readily implies that for any $\big(Y^{(1)}, \dots, Y^{(B)}\big) \in \mathcal{NE}_\std(Y)$,
    \begin{align}
        & \Prob\left( \frac{1}{B} \sum_{i=1}^B A_\loc\big( X'^{(i)}, Y^{(i)} \big) \geq \frac{1+\beta^\star(\alpha)}{2} \right) \nonumber \\
        & \leq \Prob\left( \frac{1}{B} \sum_{i=1}^B \1\left\{\maj(X'^{(i)}) = \maj(Y^{(i)}), \Delta(Y^{(i)}) > 0 \right\} \geq \frac{1+\beta^\star(\alpha)}{2} \right) \stackrel{\text{(Hoeffding)}}{\leq} \exp\left( - B \cdot \frac{\beta^\star(\alpha)^2}{4} \right). \label{eq:ne_2_hoeffding_bd}
    \end{align}
    We note that the application of Hoeffding's inequality in \eqref{eq:ne_2_hoeffding_bd} is conservative and involves a denominator of $4$ in the exponential (rather than $2$) due to subtleties arising from ties giving $\Delta(Y^{(i)}) = 0$, and thus
    \begin{align} \label{eq:hoeffding_issue}
        \Prob\left( \maj(X'^{(i)}) = \maj(Y^{(i)}), \Delta(Y^{(i)}) > 0 \right) \neq 1/2.
    \end{align}
    It is readily confirmed that as $b \to \infty$, the LHS of \eqref{eq:hoeffding_issue} tends to $1/2$ and the proportion of bad blocks is vanishing, so our Hoeffding invocation is justified. A union bound and these observations now yield that
    \begin{align}
        & \Prob\left( T_\std\!\left( X', Y \right) \geq \frac{1+\beta^\star(\alpha)}{2} \right) \leq \sum_{(Y^{(1)}, \dots, Y^{(B)}) \in \mathcal{NE}_\std(Y)} \Prob\left( \frac{1}{B} \sum_{i=1}^B A_\loc\big( X'^{(i)}, Y^{(i)} \big) \geq \frac{1+\beta^\star(\alpha)}{2} \right) \nonumber \\
        & \qquad \stackrel{ \eqref{eq:ne_2_hoeffding_bd}}{\leq} \left|\mathcal{NE}_\std\!\left( Y \right) \right| \cdot \exp\left( - B \cdot \frac{\beta^\star(\alpha)^2}{4} \right) \stackrel{\eqref{eq:ne_2_size_bd}}{\leq } e^{B \cdot o_b(1)} \exp\left( - B \cdot \frac{\beta^\star(\alpha)^2}{4} \right) \nonumber \\
        & \qquad \stackrel{\text{($b$ large)}}{\leq} \exp\left( - B \cdot \frac{\beta^\star(\alpha)^2}{8} \right) = e^{-\Omega(N)}. \label{eq:ne_2_exp_bd}
    \end{align}
    This establishes the desired guarantee.
\end{proof}

As the statement of \cref{prop:regular_alignment_property} and the definition of the good set $\mathcal G(X)$ were phrased with respect to induced total alignment, our next task is to relate standardized total alignment to induced total alignment. Towards this end, we introduce the following standardization algorithm, which will serve as our key link between the two different kinds of near-equipartitions of \cref{defn:near_equipartitions}.

\begin{definition}[Standardization algorithm] \label{defn:standardization_algorithm}
    For a string $y \in \{0,1\}^M$, say we are given an induced near-equipartition $\big(y^{(1)}, \dots, y^{(B)}\big) \in \mathcal{NE}_\ind(y)$, and denote the collection of indices corresponding to \textit{exceptional} blocks of this induced near-equipartition via
    \begin{align*}
        \mathcal I_{\left(y^{(1)}, \dots, y^{(B)}\right)}^\exc := \left\{ i \in [B] : |y^{(i)}| \notin \left[ (1-\delta)\alpha b, (1+\delta)\alpha b \right] \right\}.
    \end{align*}
    The \emph{standardization algorithm} defines a corresponding map
    \begin{align} \label{eq:standardization_algo_map}
        \varphi_{y}: \mathcal{NE}_\ind(y) \to \mathcal{NE}_\std(y).
    \end{align}
    The algorithm proceeds sequentially through the strings $y^{(i)}$, constructing $\big(\tilde{y}^{(1)}, \dots, \tilde{y}^{(B)}\big) \in \mathcal{NE}_\std(y)$ sequentially. Say that the algorithm has processed $\big(y^{(1)}, \dots, y^{(i)}\big)$ and has constructed $\big(\tilde{y}^{(1)}, \dots, \tilde{y}^{(i)}\big)$ so that
    \begin{align} \label{eq:standarziation_algo_key_cond}
        y^{(1)} \cdots y^{(i)} = \tilde{y}^{(1)} \cdots \tilde{y}^{(i)}.
    \end{align}
    The next strings $\tilde{y}^{(k)}$ for $k \geq i+1$ are constructed via the following procedure. We continue reading strings $y^{(i+1)}, y^{(i+2)}, \dots$ and stop when we either
    \begin{enumerate}
        \item[(a)] first hit a string $y^{(k)}$ for which $k \in \mathcal I_{\left(y^{(1)}, \dots, y^{(B)}\right)}^\exc$,

        \item[(b)] progress through $b^\epsilon$ strings without hitting such a string $y^{(k)}$ for which $k \in \mathcal I_{\left(y^{(1)}, \dots, y^{(B)}\right)}^\exc$, or

        \item[(c)] reach the final string $y^{(B)}$ of the induced near-equipartition $\big(y^{(1)}, \dots, y^{(B)}\big)$.
    \end{enumerate}
    Take the corresponding sequence of strings $\big(y^{(i+1)}, \dots, y^{(j)}\big)$ that we have just progressed through. We construct $\big(\tilde{y}^{(i+1)}, \dots, \tilde{y}^{(j)}\big)$ as follows.
    \begin{itemize}
        \item If $j \in \mathcal I_{\left(y^{(1)}, \dots, y^{(B)}\right)}^\exc$, then let $\tilde{y}^{(j)} = y^{(j)}$. If $j > i+1$, then for the remaining indices $i+1, \dots, j-1$, let $|\tilde{y}^{(k)}| = \alpha b$ for all $k \in \{ i+1, \dots, j-2 \}$ and set $\tilde{y}^{(j-1)}$ accordingly so that
        \begin{align} \label{eq:standarziation_algo_key_cond_j}
            y^{(1)} \cdots y^{(j)} = \tilde{y}^{(1)} \cdots \tilde{y}^{(j)}
        \end{align}
        potentially forcing $|\tilde{y}^{(j-1)}| \neq \alpha b$.

        \item If $j \notin \mathcal I_{\left(y^{(1)}, \dots, y^{(B)}\right)}^\exc$, then let $|\tilde{y}^{(k)}| = \alpha b$ for all $k \in \{ i+1 \dots, j-1 \}$ (this set is empty if $j = i+1$) and set $\tilde{y}^{(j)}$ accordingly so that
        \begin{align*}
            y^{(1)} \cdots y^{(j)} = \tilde{y}^{(1)} \cdots \tilde{y}^{(j)},
        \end{align*}
        potentially forcing $|\tilde{y}^{(j)}| \neq \alpha b$.
    \end{itemize}
    We continue similarly until we have constructed all of $\varphi_{y}\big(y^{(1)}, \dots, y^{(B)}\big) = \big(\tilde{y}^{(1)}, \dots, \tilde{y}^{(B)}\big)$.
\end{definition}

It is clear from \cref{defn:standardization_algorithm} that the map \eqref{eq:standardization_algo_map} preserves the key inductive condition \eqref{eq:standarziation_algo_key_cond}. We now confirm that it is well-defined (i.e., maps to $\mathcal{NE}_\std(y)$). We let
\begin{align*}
    \big(y^{(i+1)},\ldots,y^{(j)}\big)
\end{align*}
denote a maximal consecutive sequences of blocks processed between successive stopping points of the standardization algorithm. In the case where $j \in \mathcal I_{\left(y^{(1)}, \dots, y^{(B)}\right)}^\exc$, it follows from \cref{defn:standardization_algorithm} that $|(j-1)-(i+1)+1| \leq b^\epsilon$ and that for all $k \in \{ i+1, \dots, j-1 \}$, it holds that
\begin{align*}
    |y^{(k)}| \in \left[ (1-\delta)\alpha b, (1+\delta)\alpha b \right].
\end{align*}
Thus, the total amount that we shift the index of the last endpoint as we construct the mapping of strings
\begin{align*}
    \big(y^{(i+1)}, \dots, y^{(j-2)}\big) \to \big(\tilde{y}^{(i+1)}, \dots, \tilde{y}^{(j-2)}\big)
\end{align*}
is readily observed to be bounded by
\begin{align} \label{eq:total_shift}
    (b^\epsilon-2) \cdot \delta \alpha b \leq b^\epsilon \cdot \delta \alpha b \stackrel{\eqref{eq:window_parameters}}{=} \alpha b^{1/2+2\epsilon} \stackrel{\text{($b$ large)}}{\leq} \left( \alpha \land (1-\alpha) \right)b.
\end{align}
Since $j-1 \notin \mathcal I_{\left(y^{(1)}, \dots, y^{(B)}\right)}^\exc$, we may thus define $\tilde{y}^{(j-1)}$ in such a way (i.e., by contracting or extending $y^{(j-1)}$ accordingly) that $|\tilde{y}^{(j-1)}| \in \{0, \dots, b\}$ and \eqref{eq:standarziation_algo_key_cond_j} holds. The analysis for the case in which $j \notin \mathcal I_{\left(y^{(1)}, \dots, y^{(B)}\right)}^\exc$ is effectively identical. Finally, it readily follows that the number of strings $\tilde{y}^{(k)}$ such that $|\tilde{y}^{(k)}| \neq \alpha b$ is at most
\begin{align*}
    \left( \underbrace{2\gamma}_{2 \cdot \big|\mathcal I_{(y^{(1)}, \dots, y^{(B)})} \big| \text{ for bad strings of } (y^{(1)}, \dots, y^{(B)})} + \underbrace{b^{-\epsilon}}_{\text{possible modification every $b^\epsilon$ strings}} \right) B \stackrel{\eqref{eq:window_parameters}}{=} 3\gamma B.
\end{align*}
Altogether, this discussion demonstrates that
\begin{align*}
    \varphi_{y}\big(y^{(1)}, \dots, y^{(B)}\big) = \big(\tilde{y}^{(1)}, \dots, \tilde{y}^{(B)}\big) \in \mathcal{NE}_\std(y),
\end{align*}
and we conclude that the standardization algorithm of \cref{defn:standardization_algorithm} is indeed well-defined. 

The standardization algorithm of \cref{defn:standardization_algorithm} should be thought of as an approximation algorithm in the following sense. We will use it to produce a uniform approximation-type bound of the form
\begin{align} \label{eq:approximation_bd}
    \frac{1}{B}\sum_{i=1}^B A_\loc\big(X'^{(i)}, Y^{(i)}\big) \leq \frac{1}{B} \sum_{i=1}^B A_\loc\big(X'^{(i)}, \tilde{Y}^{(i)}\big) + \frac{\beta^\star(\alpha)}{2} \stackrel{\eqref{eq:total_alignment_second_kind_defn}}{\leq} T_\std(X', Y) + \frac{\beta^\star(\alpha)}{2},
\end{align}
where $\big( Y^{(1)}, \dots, Y^{(B)} \big) \in \mathcal{NE}_\ind(Y)$, we have that
\begin{align*}
    \varphi_{Y}\big( Y^{(1)}, \dots, Y^{(B)} \big) = \big( \Tilde{Y}^{(1)}, \dots, \Tilde{Y}^{(B)} \big) \in \mathcal{NE}_\std(Y),
\end{align*}
and the bound \eqref{eq:approximation_bd} holds simultaneously over all $\big( Y^{(1)}, \dots, Y^{(B)} \big) \in \mathcal{NE}_\ind(Y)$ with probability at least $1 - e^{-\Omega(N)}$. Then we have that
\begin{align}
    T_\ind(X', Y) & \stackrel{\eqref{eq:total_alignment_first_kind_defn}}{=} \sup_{\left( Y^{(1)}, \dots, Y^{(B)} \right) \in \mathcal{NE}_\ind(Y)} \frac{1}{B}\sum_{i=1}^B A_\loc\big(X'^{(i)}, Y^{(i)}\big) \nonumber \\
    & \stackrel{\eqref{eq:approximation_bd}}{\leq} T_\std(X', Y) + \frac{\beta^\star(\alpha)}{2} \stackrel{\text{(\cref{prop:total_alignment_2_bd})}}{<} \frac{1}{2} + \beta^\star(\alpha) \label{eq:T_1_probable_bound}
\end{align}
with probability $\geq 1 - e^{-\Omega(N)}$, producing the desired distinction with \cref{prop:regular_alignment_property}. The key idea in establishing \eqref{eq:approximation_bd} is reducing the extremal behavior of the mapping $\varphi_{Y}$ to the sizes of fluctuations from the simple random walk of contiguous substrings of $Y$. We proceed to formally establish this observation.

\begin{definition}[Biased stretch] \label{defn:biased_stretch}
    A \emph{biased stretch} of $y \in \{0,1\}^M$ is a contiguous substring of $y$ of length at most $b^{1+\epsilon}$ which has a contiguous (sub-)substring $z$ of length at most $b^{1/2+2\epsilon}$ such that $\Delta(z) \geq b^{1/2-2\epsilon}$. A contiguous substring of $y$ of length at most $b^{1+\epsilon}$ that is not biased is said to be an \emph{unbiased stretch}.
\end{definition}
The importance of \cref{defn:biased_stretch} is illustrated by the following key observation. Let $y \in \{0,1\}^M$. Fix an induced near-equipartition $\big(y^{(1)}, \dots, y^{(B)}\big) \in \mathcal{NE}_\ind(y)$. Let $\big(y^{(i+1)}, \dots, y^{(j)}\big)$ denote a sequence of these strings that is transformed into the corresponding sequence of strings $\big(\tilde{y}^{(i+1)}, \dots, \tilde{y}^{(j)}\big)$ when the standardization algorithm is invoked on $\big(y^{(1)}, \dots, y^{(B)}\big)$. It holds for any $k \in \{i+1, \dots, j \}$ that the symmetric difference of $y^{(k)}$ and $\tilde{y}^{(k)}$, where we regard $y^{(k)}$ and $\tilde{y}^{(k)}$ as substrings of $y$ (i.e., we do not compare them bit-by-bit), is exactly given by two contiguous substrings of $y^{(i+1)} \cdots y^{(j)}$. Specifically, it holds that
\begin{itemize}
    \item one contiguous substring is a (possibly empty) \textit{prefix} $\mathcal P\big(y^{(k)}, \tilde{y}^{(k)}\big)$ of either $y^{(k)}$ or $\tilde{y}^{(k)}$;

    \item the other contiguous substring is a (possibly empty) \textit{suffix} $\mathcal S\big(y^{(k)}, \tilde{y}^{(k)}\big)$ of either $y^{(k)}$ or $\tilde{y}^{(k)}$.
\end{itemize}
We may now bound the size of the prefix and the suffix via
\begin{align} \label{eq:sym_diff_bd}
    \left| \mathcal P\big(y^{(k)}, \tilde{y}^{(k)}\big) \right| \lor \left| \mathcal S\big(y^{(k)}, \tilde{y}^{(k)}\big) \right| \stackrel{\text{(\cref{defn:standardization_algorithm}, Cond. (b))}}{\leq} \delta \alpha b \cdot b^\epsilon \stackrel{\eqref{eq:window_parameters}}{=} \alpha \cdot b^{-1/2+\epsilon} \cdot b^{1+\epsilon} \leq b^{1/2+2\epsilon}.
\end{align}
Now, if we further have that the contiguous substring $y^{(i+1)} \cdots y^{(j)}$ of $y$ is an unbiased stretch, then we may use \cref{defn:biased_stretch} to bound the corresponding difference of local alignment values via
\begin{align}
    & \left| A_\loc\big(X'^{(k)}, y^{(k)}\big) - A_\loc\big(X'^{(k)}, \tilde{y}^{(k)}\big) \right| \leq \delta \left| \sum_{j=1}^{|y^{(k)}|} \left( 2(y^{(k)})_j - 1 \right) - \sum_{j=1}^{|\tilde{y}^{(k)}|} \left( 2(\tilde{y}^{(k)})_j - 1 \right) \right| \label{eq:local_alignment_bd_casework} \\
    & \qquad \leq \delta \left( \Delta\left( \mathcal P\big(y^{(k)}, \tilde{y}^{(k)}\big) \right) + \Delta\left( \mathcal S\big(y^{(k)}, \tilde{y}^{(k)}\big) \right) \right) \nonumber \\
    & \qquad \stackrel{\text{($y^{(i+1)} \cdots y^{(j)}$ unbiased stretch)}, \ \eqref{eq:sym_diff_bd}}{\leq} \delta \left( b^{1/2-2\epsilon} + b^{1/2-2\epsilon} \right) \stackrel{\eqref{eq:window_parameters}}{=} 2b^{-1/2+\epsilon} \cdot b^{1/2-2\epsilon} = 2b^{-\epsilon}, \label{eq:local_alignment_bd}
\end{align}
where the inequality of \eqref{eq:local_alignment_bd_casework} follows by routine casework and the definition of local alignment. This casework is based on, as per \eqref{eq:local_alignment_defn}, 
\begin{itemize}
    \item whether the local alignment term $A_\loc\big(X'^{(k)}, y^{(k)}\big)$ is $0$, $\delta \Delta\big(y^{(k)}\big)$, or $1$;
    \item whether the local alignment term $A_\loc\big(X'^{(k)}, \tilde{y}^{(k)}\big)$ is $0$, $\delta \Delta\big(\tilde{y}^{(k)}\big)$, or $1$.
\end{itemize}
We note that if exactly one of the two local alignment terms is $0$, then the definition of local alignment \eqref{eq:local_alignment_defn} implies that the two sums considered in this initial bound have opposite signs, from which the validity of the bound readily follows. As this final bound can be made arbitrarily small by choosing $b$ large, the identity \eqref{eq:local_alignment_bd} suggests a method via which we may relate induced near-equipartitions $\mathcal{NE}_\ind(Y)$ with their mappings under the standardization algorithm map $\varphi_{Y}$. Our strategy will thus be to show that the number of biased stretches of $Y$ is typically modest, which we combine with \eqref{eq:local_alignment_bd} to prove \eqref{eq:approximation_bd} uniformly over all $\big( Y^{(1)}, \dots, Y^{(B)} \big) \in \mathcal{NE}_\ind(Y)$.

\begin{proposition}[Few biased stretches] \label{prop:few_biased_stretches}
    With probability $\geq 1 - e^{-\Omega(N)}$, $Y$ has $\leq M \cdot b^{-(1+2\epsilon)}$ biased stretches.
\end{proposition}

\begin{proof}
    It holds for any contiguous substring $Z$ of $Y$ of length $\ell \leq b^{1/2+2\epsilon}$ that
    \begin{align} \label{eq:one_biased_stretch_bd}
        \Prob\left( \Delta(Z) \geq b^{1/2-2\epsilon} \right) \stackrel{\text{(Hoeffding)}}{\leq} 2\exp\left( - \frac{2\left( b^{1/2-2\epsilon} \right)^2}{4\ell} \right) \stackrel{(\ell \leq b^{1/2+2\epsilon})}{\leq} 2\exp\left( - \frac{b^{1/2-6\epsilon}}{2} \right).
    \end{align}
    So for any $M/\ell$ disjoint contiguous substrings $Z(i)$ of $Y$, each with length at most $\ell \leq b^{1/2+2\epsilon}$,
    \begin{align}
        & \Prob\left( \frac{1}{M/\ell} \sum_{i=1}^{M/\ell} \1\left\{ \Delta\left( Z(i) \right) \geq b^{1/2-2\epsilon} \right\} \geq b^{-(4+8\epsilon)} \right) \label{eq:Y_equipartition_biased_stretches_bd_ind} \\
        & \quad \stackrel{\text{(Chernoff)}}{\leq} \exp\left( -\frac{M}{\ell} \cdot \KL\left( b^{-(4+8\epsilon)} \doubleline \frac{1}{M/\ell} \sum_{i=1}^{M/\ell} \Prob\left( \Delta\left( Z(i) \right) \geq b^{1/2-2\epsilon} \right) \right) \right) \nonumber \\
        & \quad \stackrel{(\ell \leq b^{1/2+2\epsilon}), \eqref{eq:one_biased_stretch_bd},\text{($b$ large)}}{\leq} \exp\left( -\frac{M}{b^{1/2+2\epsilon}} \cdot  \KL\left( b^{-(4+8\epsilon)} \doubleline 2e^{-b^{1/2-6\epsilon}/2} \right) \right) = e^{-\Omega(N)}. \label{eq:Y_equipartition_biased_stretches_bd}
    \end{align}
    We can handle (i.e., correspond to an indicator summand in \eqref{eq:Y_equipartition_biased_stretches_bd_ind}) all contiguous substrings of $Y$ of length $\ell \leq b^{1/2+2\epsilon}$ using 
    \begin{align*}
        b^{1/2+2\epsilon} \cdot b^{1/2+2\epsilon} = b^{1+4\epsilon}
    \end{align*}
    bounds of the form \eqref{eq:Y_equipartition_biased_stretches_bd}. Specifically, for each candidate length $\ell \leq b^{1/2+2\epsilon}$ of the contiguous substrings, we take some $j \leq \ell$ and sequentially equipartition the substring of $Y$ starting at position $j$ into contiguous substrings of length $\ell$ until there are strictly fewer than $\ell$ bits left. This produces at most $M/\ell$ contiguous substrings of length $\ell$, which we then correspond to the indicator summands studied in \eqref{eq:Y_equipartition_biased_stretches_bd}. Furthermore, any contiguous substring of $Y$ of length $\leq b^{1/2+2\epsilon}$ is contained in at most 
    \begin{align*}
        b^{1+\epsilon} \cdot b^{1+\epsilon} = b^{2\left( 1+\epsilon \right)}
    \end{align*}
    biased stretches. This bound is observed by fixing the length of a candidate biased stretch to be at most $b^{1+\epsilon}$, then counting the number of candidate biased stretches with said length that contain the contiguous substring of $Y$. Altogether, a union bound over the conditions for all of the bounds \eqref{eq:Y_equipartition_biased_stretches_bd} that we consider yields that with probability
    \begin{align*}
        \geq 1 - b^{1+4\epsilon} \cdot e^{-\Omega(N)} = 1-e^{-\Omega(N)},
    \end{align*}
    there are at most
    \begin{align*}
        \underbrace{b^{1+4\epsilon}}_{\text{num. bds. of form } \eqref{eq:Y_equipartition_biased_stretches_bd}} \cdot \underbrace{M \cdot b^{-(4+8\epsilon)}}_{\text{upper bd. from } \eqref{eq:Y_equipartition_biased_stretches_bd}} \cdot \underbrace{b^{2\left( 1+\epsilon \right)}}_{\text{upper bd. on num. biased stretches for a substring}} = M \cdot b^{-(1+2\epsilon)}
    \end{align*}
    biased stretches in $Y$.
\end{proof}

We now put everything together to derive the desired uniform approximation-type bound \eqref{eq:approximation_bd}. Let $y \in \{0,1\}^M$ be such a string with at most $M \cdot b^{-(1+2\epsilon)}$ biased stretches. We fix an arbitrary induced near-equipartition 
\begin{align*}
    \big(y^{(1)}, \dots, y^{(B)}\big) \in \mathcal{NE}_\ind(y),
\end{align*}
and we invoke the standardization algorithm on it to get
\begin{align*}
    \varphi_y\big(y^{(1)}, \dots, y^{(B)}\big) = \big(\tilde{y}^{(1)}, \dots, \tilde{y}^{(B)}\big) \in \mathcal{NE}_\std(y).
\end{align*}
We proceed with the bound, where we note that the second summation below (and all similar sums of this form) ranges over the contiguous block-intervals
\begin{align*}
    \big(y^{(i+1)},\ldots,y^{(j)}\big)
\end{align*}
produced by the standardization algorithm of \cref{defn:standardization_algorithm}, i.e., the maximal consecutive sequences of blocks between successive stopping points of the algorithm. These intervals form a partition of the $B$ blocks of the fixed induced near-equipartition. We have that 
\begin{align}
    & \frac{1}{B} \sum_{i=1}^B A_\loc\big(X'^{(i)}, y^{(i)} \big) = \frac{1}{B} \sum_{(y^{(i+1)}, \dots, y^{(j)})} \sum_{k=i+1}^j A_\loc\big(X'^{(k)}, y^{(k)}\big) \nonumber \\
    & \quad \stackrel{\eqref{eq:local_alignment_defn}, \eqref{eq:local_alignment_bd}}{\leq} \frac{1}{B} \left[ \sum_{\substack{(y^{(i+1)},\dots,y^{(j)}):\\ y^{(i+1)}\cdots y^{(j)}\text{ is a}\\ \text{biased stretch}}} \sum_{k=i+1}^j 1 + \sum_{\substack{(y^{(i+1)},\dots,y^{(j)}):\\ y^{(i+1)}\cdots y^{(j)}\text{ is an}\\ \text{unbiased stretch}}} \sum_{k=i+1}^j \left( A_\loc\big(X'^{(k)}, \tilde{y}^{(k)}\big) + 2b^{-\epsilon} \right) \right] \nonumber \\
    & \quad \stackrel{\text{(\cref{defn:standardization_algorithm}, Cond. (b))}}{\leq} \frac{1}{B} \left[ B \cdot 2b^{-\epsilon} + \sum_{\substack{(y^{(i+1)},\dots,y^{(j)}):\\ y^{(i+1)}\cdots y^{(j)}\text{ is a}\\ \text{biased stretch}}} b^\epsilon + \sum_{\substack{(y^{(i+1)},\dots,y^{(j)}):\\ y^{(i+1)}\cdots y^{(j)}\text{ is an}\\ \text{unbiased stretch}}} \sum_{k=i+1}^j A_\loc\big(X'^{(k)}, \tilde{y}^{(k)}\big) \right] \nonumber \\
    & \quad \leq \frac{1}{B} \left[ B \cdot 2b^{-\epsilon} + b^\epsilon \cdot M \cdot b^{-(1+2\epsilon)} + \sum_{k=1}^B A_\loc\big(X'^{(k)}, \tilde{y}^{(k)}\big) \right] \nonumber \\
    & \quad \leq 2b^{-\epsilon} + \alpha b^{1+\epsilon} \cdot b^{-(1+2\epsilon)} + \frac{1}{B}\sum_{k=1}^B A_\loc\big(X'^{(k)}, \tilde{y}^{(k)}\big)\stackrel{\eqref{eq:total_alignment_second_kind_defn}, \text{($b$ large)}}{\leq} T_\std(X', y) + \frac{\beta^\star(\alpha)}{2}. \label{eq:standardization_algo_bd}
\end{align}
The bound \eqref{eq:standardization_algo_bd} holds uniformly over the choice of induced near-equipartition $(y^{(1)}, \dots, y^{(B)}) \in \mathcal{NE}_\ind(y)$ and holds irrespective of the choice of the string $y \in \{0,1\}^M$ with $\leq M \cdot b^{-(1+2\epsilon)}$ biased stretches. Thus, \cref{prop:few_biased_stretches} and the preceding discussion together imply that with probability $\geq 1 - e^{-\Omega(N)}$, the inequality \eqref{eq:T_1_probable_bound} holds, i.e., that
\begin{align} \label{eq:Y_not_in_good_set}
    \Prob\left( Y \notin \mathcal G\left( X' \right) \right) \geq 1 - e^{-\Omega(N)}.
\end{align}

\subsection{\texorpdfstring{Proof of \cref{thm:weak_law}}{Proof of Theorem \ref{thm:weak_law}}} \label{subsec:proof_of_thm}

We are finally ready to prove the first strict inequality of \cref{thm:weak_law}. Towards this end, we first define the event
\begin{align} \label{eq:event_E_defn}
    E_N := \left\{ X' \text{ is typical} \right\} \cap \left\{ Y \notin \mathcal G(X') \right\}.
\end{align}
Then we may write
\begin{align}
    & \E\left[ Z_{X', Y} \cdot \1\{E_N\} \right] \stackrel{\eqref{eq:event_E_defn}}{=} \E\left[ \left( \sum_{S \in \binom{[N]}{M}} \1\left\{ X'|_S = Y \right\} \right) \1\left\{ X' \text{ is typical} \right\} \1\left\{ Y \notin \mathcal G(X') \right\} \right] \nonumber \\
    & \quad = \E\left[ \1\left\{ X' \text{ is typical} \right\} \left( \mathop{\sum_{S \in \binom{[N]}{M}:}}_{X'|_S \notin \mathcal G(X')} \1\left\{ X'|_S = Y \right\} \right) \1\left\{ Y \notin \mathcal G(X') \right\} \right] \nonumber \\
    & \quad \stackrel{\text{(Fubini)}}{=} \E_{X'} \left[ \1\left\{ X' \text{ is typical} \right\} \cdot \E_Y \left[ \mathop{\sum_{S \in \binom{[N]}{M}:}}_{X'|_S \notin \mathcal G(X')} \1\left\{ X'|_S = Y \right\} \right] \right] \nonumber \\
    & \quad = \E_{X'}\left[ \1\left\{ X' \text{ is typical} \right\} \cdot 2^{-M} \cdot \left| \left\{ S \in \binom{[N]}{M}: X'|_S \notin \mathcal G(X') \right\} \right| \right] \nonumber \\
    & \quad \stackrel{\text{(\cref{prop:regular_alignment_property})}}{\leq} \E_{X'}\left[ \1\left\{ X' \text{ is typical} \right\} \cdot 2^{-M} \cdot \binom{N}{M} \cdot e^{-\Omega(N)} \right] \nonumber \\
    & \quad \leq \binom{N}{M}e^{-\Omega(N)} \cdot 2^{-M} = e^{N\left( \fnullann(\alpha) - \Omega(1) \right)}. \label{eq:Z_indicator_bd}
\end{align}
Markov's inequality thus implies that (with a weaker universal constant implicit in the $\Omega(\cdot)$ term than in the corresponding $\Omega(\cdot)$ term in the final expression of \eqref{eq:Z_indicator_bd})
\begin{align*}
    \Prob\left( Z_{X', Y} \cdot \1\{E_N\} < e^{N\left( \fnullann(\alpha) - \Omega(1) \right)} \right) \geq 1 - e^{-\Omega(N)}.
\end{align*}
The discussion after \cref{defn:typical_W} and \eqref{eq:Y_not_in_good_set} together imply $\Prob(E_N) \geq 1 - e^{-\Omega(N)}$, so we conclude that
\begin{align} \label{eq:null_quantitative}
    \Prob\left( Z_{X', Y} < e^{N\left( \fnullann(\alpha) - \Omega(1) \right)} \right) \geq 1 - e^{-\Omega(N)},
\end{align}
from which the weak law proves the strict inequality in \eqref{eq:weak_law_null_bounds}. Furthermore, after shrinking implicit constants if necessary, we may take both $\Omega(\cdot)$ terms in \eqref{eq:null_quantitative} to be governed by the same positive constant.

\smallskip

We complete the proof of \cref{thm:weak_law} via proving the strict inequality of \eqref{eq:weak_law_planted_bounds}. In the rest of \cref{subsec:proof_of_thm}, we let all $\Omega(\cdot)$ terms specifically denote the corresponding implicit constant as guaranteed in \eqref{eq:null_quantitative}. We say that $x \in \{0,1\}^{N}$ is \emph{hoarded} if
\begin{align} \label{eq:hoarded_string_condition}
     \left| \mathcal H(x) \right| := \left| \left\{ y \in \{0,1\}^M : Z_{x,y} \geq e^{N\left( \fnullann(\alpha) - \Omega(1) \right)} \right\} \right| \leq 2^M e^{-\Omega(N)/2}.
\end{align}
It then follows that
\begin{align*}
    e^{-\Omega(N)} & \stackrel{\eqref{eq:null_quantitative}}{\geq} \Prob\left( Z_{X', Y} \geq e^{N\left( \fnullann(\alpha) - \Omega(1) \right)} \right) \\
    & = \Prob\left( Z_{X', Y} \geq e^{N\left( \fnullann(\alpha) - \Omega(1) \right)} \mid X' \text{ hoarded} \right) \cdot \Prob\left( X' \text{ hoarded} \right) \\
    & \quad + \Prob\left( Z_{X', Y} \geq e^{N\left( \fnullann(\alpha) - \Omega(1) \right)} \mid X' \text{ not hoarded} \right) \cdot \Prob\left( X' \text{ not hoarded} \right) \\
    & \stackrel{\eqref{eq:hoarded_string_condition}}{\geq} 2^{-M} \cdot 2^M e^{-\Omega(N)/2} \cdot \Prob\left( X' \text{ not hoarded} \right) = e^{-\Omega(N)/2} \cdot \Prob\left( X' \text{ not hoarded} \right)
\end{align*}
from which we conclude that
\begin{align} \label{eq:hoarded_prob_bd}
    \Prob\left( X' \text{ not hoarded} \right) \leq e^{-\Omega(N)/2}.
\end{align}
Altogether, we have that, letting $X$ play the role of $X'$ above when invoking \eqref{eq:hoarded_string_condition} and \eqref{eq:hoarded_prob_bd} and recalling that $\Prob(Y = y \mid X) = Z_{X,y}/\binom{N}{M}$ for each $y \in \{0,1\}^M$,
\begin{align}
    & \Prob\left( \frac{1}{N} \log Z_{X, Y} > \fnullann(\alpha) + \frac{\Omega(1)}{4} \right) \nonumber \\
    & \quad \geq \Prob\left( X \text{ hoarded} \right) \cdot \left[ 1 - \Prob\left( \frac{1}{N} \log Z_{X, Y} \leq \fnullann(\alpha) + \frac{\Omega(1)}{4} \ \Bigg| \ X \text{ hoarded} \right) \right] \nonumber \\
    & \quad \stackrel{\eqref{eq:hoarded_prob_bd}}{\geq} \left( 1 - e^{-\Omega(N)/2} \right) \cdot \Biggl[ 1 - \Prob\left( Z_{X, Y} \leq e^{N\left( \fnullann(\alpha) + \Omega(1)/4 \right)}, Y \notin \mathcal H(X) \ \Bigg| \ X \text{ hoarded} \right) \nonumber \\
    & \qquad \qquad \qquad \qquad  \qquad \qquad \qquad - \Prob\left( Z_{X, Y} \leq e^{N\left( \fnullann(\alpha) + \Omega(1)/4 \right)}, Y \in \mathcal H(X) \ \Bigg| \ X \text{ hoarded} \right) \Bigg] \nonumber \\ 
    & \ \ \ \stackrel{\eqref{eq:hoarded_string_condition}}{\geq} \left( 1 - e^{-\Omega(N)/2} \right) \cdot \left( 1 - \underbrace{2^M \cdot e^{N\left( \fnullann(\alpha) - \Omega(1) \right)}\frac{1}{\binom{N}{M}}}_{\text{bound on $Y \notin \mathcal H(X)$ term}} - \underbrace{2^Me^{-\Omega(N)/2} \cdot e^{N\left(\fnullann(\alpha) + \Omega(1)/4 \right)} \frac{1}{\binom{N}{M}}}_{\text{bound on $Y \in \mathcal H(X)$ term}} \right) \label{eq:planted_gap_lower_bd_main_step} \\
    & \quad = \left( 1 - e^{-\Omega(N)} \right) \cdot \left( 1 - e^{-\Omega(N)} e^{N \fnullann(\alpha)} \frac{2^M}{\binom{N}{M}} - e^{-\Omega(N)/4} e^{N \fnullann(\alpha)} \frac{2^M}{\binom{N}{M}}\right) \nonumber \\
    & \quad = \left( 1 - e^{-\Omega(N)} \right) \cdot \left( 1 - e^{-\Omega(N) + o(N)} - e^{-\Omega(N)/4 + o(N)} \right) \to 1.\label{eq:planted_gap_lower_bd}
\end{align}
We note that \eqref{eq:planted_gap_lower_bd_main_step} specifically follows by, in each case, bounding the number of length-$M$ subsequences of $X$ that satisfy the conditions of the corresponding event and then multiplying by the probability $\binom{N}{M}^{-1}$ of $Y$ being any particular such length-$M$ subsequence. We conclude in conjunction with the weak limit of \eqref{eq:weak_law_planted_bounds} that $\fpl(\alpha) > \fnullann(\alpha)$. This completes the proof of \cref{thm:weak_law}.

\subsection{\texorpdfstring{Consequences of \cref{thm:weak_law}}{Consequences of Theorem \ref{thm:weak_law}}}

We now turn to the consequences of \cref{thm:weak_law} for uniformly random codes over the deletion channel.

\subsubsection{Positive rate for uniform codes under the deletion channel} \label{sec:coding_applications}

With \cref{thm:weak_law} in hand, we now prove \cref{cor:pos_mutual_info}, settling \cite[Conjecture 3]{pernice2024mutual}.

\begin{proof}[Proof of \cref{cor:pos_mutual_info}]
    The $p=0$ case follows directly from \eqref{eq:unif-mut-inf-limit}, so we are free to fix $p \in (0,1)$. We set $\alpha = 1-p$. In particular, we recall from \eqref{eq:unif-mut-inf-limit} that the largest rate achievable via uniform random codebooks admits the representation
    \begin{align} \label{eq:c_unif_defn}
        C_{\unif}(p) = \lim_{N \to \infty} \frac{1}{N} I\left( X; \BDC_p(X) \right) \stackrel{\eqref{eq:del-chann-cap-connection}}{=} \alpha \log 2 - h(\alpha) + \fpl(\alpha),
    \end{align}
    where the last equality follows since the normalized log-partition function $\frac{1}{N}\log Z_{X,Y}$ is uniformly bounded. Thus, for fixed $\alpha \in (0,1)$, it holds that
    \begin{align}
        0 & \stackrel{\eqref{eq:weak_law_planted_bounds}}{<} \fpl(\alpha) - \fnullann(\alpha) = \fpl(\alpha) - \left( h(\alpha) - \alpha\log 2 \right) \stackrel{\eqref{eq:c_unif_defn}}{=} C_{\unif}(1-\alpha). \label{eq:C_unif_positivity}
    \end{align}
    We conclude that $d_\unif^* = 1$, as desired.
\end{proof}

\subsubsection{Explicit capacity lower bound} \label{sec:numerics}

We now revisit the proof of \cref{thm:weak_law}, keeping track of constants in order to obtain an explicit strictly positive lower bound on $C_{\unif}(p)$. Towards this end, for $\alpha \in (0,1)$, we define the quantity
\begin{align*}
    \kappa(\alpha) := \left\lceil \frac{1920^{96}}{\alpha^{24} \beta(\alpha)^{96} (1-\alpha)^{12}} \right\rceil.
\end{align*}
We stress that we do not attempt to optimize the resulting bound; our objective is simply to extract a fully explicit positive capacity lower bound for uniform random codebooks in the likely deletion regime.
\begin{theorem}[Explicit capacity lower bound] \label{thm:quantitative_capacity_bound}
    For $p \in (0, 1)$, we have that
    \begin{align*}
        C_{\unif}(p) \geq \frac{\left( \beta(1-p) \right)^3}{51200 \cdot \left( \kappa(1-p) \right)^5} > 0.
    \end{align*}
\end{theorem}
Since the proof of \cref{thm:quantitative_capacity_bound} mainly consists of verifying a collection of routine inequalities and parameter constraints, we defer it to \cref{app:quantitative_cap_bd_proof}.

\section{Annealed Free Energy Under the Planted Model} \label{sec:annealed_free_planted}

In this section, we derive an exact formula for the \emph{annealed free energy} of the planted model,
\begin{align*}
    \fplann(\alpha) = \lim_{N\to\infty}\frac{1}{N} \log \E \left[ Z_{X, Y} \right],
\end{align*}
which by Jensen gives an upper bound on $\fpl(\alpha)$. We note that \cite{pernice2024mutual} gave an efficient algorithm to approximate $\fplann(\alpha)$ to any desired precision. Here we improve on their result by proving \cref{thm:main-annealed-intro} and giving an exact asymptotic formula.

It will be convenient for us to identify sets $S \in \binom{[N]}{M}$ with functions $\sigma\in \Sigma = \Sigma_{N,M}$ such that $\text{Im}(\sigma)=S.$ Moreover, for $\sigma,\tau \in \Sigma$, we use the notations
\begin{align*}
    I(\sigma,\tau) &= \{j \in [M]: \sigma_j = \tau_j\}; \\
    \brac{\sigma,\tau} &= \sum_{j=1}^M \1\{\sigma_j = \tau_j\} = |I(\sigma,\tau)|.
\end{align*}
Given a string $X \in \{0,1\}^N$, we use the notation
\begin{align*}
    X_\sigma := X|_{\text{Im}(\sigma)}.
\end{align*}
Note that we have
\begin{align*}
    \E\left[ Z_{X, Y} \right] = \frac{1}{\binom{N}{M}} \sum_{\sigma,\tau\in \Sigma} \Prob\left(X_\sigma = X_\tau \right) = \frac{1}{\binom{N}{M}} \sum_{\sigma,\tau \in\Sigma} 2^{-(M - \brac{\sigma,\tau})} = \frac{1}{\binom{N}{M}2^M}\sum_{\sigma,\tau \in\Sigma}  2^{\brac{\sigma,\tau}}.
\end{align*}
Hence, to prove \cref{thm:main-annealed-intro}, it suffices to show that, letting
\begin{align*}
    \overline{Z} &= \sum_{\sigma,\tau \in\Sigma}  2^{\brac{\sigma,\tau}},
\end{align*}
we have
\begin{align}\label{eq:Z-bar-lim}
    \lim_{M\to\infty} \frac{1}{N}\log \overline{Z} &= -\log x_\alpha\;-\;\alpha\log y_\alpha.
\end{align}
The remainder of \cref{sec:annealed_free_planted} proves \eqref{eq:Z-bar-lim} and hence \cref{thm:main-annealed-intro}. We begin from the elementary identity
\[
2^{\brac{\sigma,\tau}}
=\sum_{J\subseteq I(\sigma,\tau)} 1
=\sum_{\ell=0}^{M}\ \sum_{J\in\binom{[M]}{\ell}} \1\{J\subseteq I(\sigma,\tau)\}.
\]
Summing over $\sigma,\tau$ yields
\begin{equation}\label{eq:expandJ}
\overline{Z} =\sum_{\ell=0}^M\ \sum_{J\in\binom{[M]}{\ell}} \sum_{\sigma,\tau\in\Sigma}\1\{J\subseteq I(\sigma,\tau)\}.
\end{equation}
Fix $\ell$ and $J=\{j_1<\cdots<j_\ell\}$. If $J\subseteq I(\sigma,\tau)$ then
$\sigma_{j_r}=\tau_{j_r}$ for each $r$, and these common values form an increasing $\ell$-tuple
$K=\{i_1<\cdots<i_\ell\}\subseteq [N]$. Thus
\begin{align*}
    \1\{J\subseteq I(\sigma,\tau)\} =\sum_{K\in\binom{[N]}{\ell}}\1\{\sigma_{j_r}=\tau_{j_r}=i_r \text{ for all } r \in [\ell]\}.
\end{align*}
Insert this in \eqref{eq:expandJ} and exchange sums to obtain
\begin{equation}\label{eq:expandJK}
    \overline{Z} =\sum_{\ell=0}^M\ \sum_{1\le j_1<\cdots<j_\ell\le M}\ \sum_{1\le i_1<\cdots<i_\ell\le N} \Bigl|\{\sigma\in\Sigma:\sigma_{j_r}=i_r \text{ for all } r \in [\ell] \}\Bigr|^2,
\end{equation}
where the square comes from the independence of $\sigma$ and $\tau$ given the constraints.

\subsection{Counting constrained increasing maps}
Fix $j_1<\cdots<j_\ell$ and $i_1<\cdots<i_\ell$. Set the convenient boundary values
\begin{align*}
    i_0:=0,\quad j_0:=0,\qquad i_{\ell+1}:=N+1,\quad j_{\ell+1}:=M+1.
\end{align*}
For each $k\in[\ell+1]$, consider the interval of \emph{positions} $\{j_{k-1}+1,\dots,j_k-1\}$ of size $j_k-j_{k-1}-1$, whose values must lie strictly between $i_{k-1}$ and $i_k$. There are exactly $i_k-i_{k-1}-1$ available integers in $\{i_{k-1}+1,\dots,i_k-1\}$, and choosing which of those occupy the $j_k-j_{k-1}-1$ slots uniquely determines $\sigma$ on that interval by monotonicity. Therefore,
\begin{align*}
    \Bigl|\{\sigma\in\Sigma_{N,M}:\sigma_{j_r}=i_r \text{ for all } r \in [\ell] \}\Bigr| =\prod_{k=1}^{\ell+1}\binom{i_k-i_{k-1}-1}{\,j_k-j_{k-1}-1\,}.
\end{align*}
Plugging into \eqref{eq:expandJK} gives the following explicit form.

\begin{lemma}[Gap product formula]\label{lem:gap}
For all $1\le M\le N$,
\begin{align*}
    \overline{Z} =\sum_{\ell=0}^{M}\ \sum_{1\le j_1<\cdots<j_\ell\le M}\ \sum_{1\le i_1<\cdots<i_\ell\le N} \ \prod_{k=1}^{\ell+1}\binom{i_k-i_{k-1}-1}{\,j_k-j_{k-1}-1\,}^{\!2},
\end{align*}
with the boundary convention $i_0=j_0=0$ and $i_{\ell+1}=N+1$, $j_{\ell+1}=M+1$.
\end{lemma}

\subsection{From gaps to compositions}
For each $k\in[\ell+1]$ define the positive integers
\[
a_k:=i_k-i_{k-1}\in\N_{>0},\qquad b_k:=j_k-j_{k-1}\in\N_{>0}.
\]
Then $\sum_{k=1}^{\ell+1}a_k=N+1$ and $\sum_{k=1}^{\ell+1}b_k=M+1$, and
\[
\binom{i_k-i_{k-1}-1}{j_k-j_{k-1}-1}=\binom{a_k-1}{b_k-1}.
\]
Conversely, any $(a_1,b_1),\dots,(a_{\ell+1},b_{\ell+1})\in(\N_{>0}^2)^{\ell+1}$ with those sum constraints
uniquely determines $(i_1,\dots,i_\ell)$ and $(j_1,\dots,j_\ell)$ by partial summation.
Thus, \cref{lem:gap} can be rewritten as a sum over \emph{block compositions}. We let
\begin{align*}
    w(a,b):=\binom{a-1}{b-1}^2\quad\text{for }a\ge 1,\ 1\le b\le a, \qquad w(a,b):=0 \text{ otherwise},
\end{align*}
and for each $\ell \ge 0$ we define
\begin{align*}
    \mathcal{S}_{\ell}(N,M) := \left\{(a_1,b_1),\dots,(a_{\ell+1},b_{\ell+1})\in(\N_{>0}^2)^{\ell+1}: \sum_{k=1}^{\ell+1}a_k=N+1,\ \sum_{k=1}^{\ell+1}b_k=M+1 \right\}.
\end{align*}
Then writing $L=\ell+1$, we have
\begin{equation}\label{eq:composition-sum}
\overline{Z}
=\sum_{L=1}^{M+1} \ \sum_{(X_1,\dots,X_L)\in \mathcal{S}_{L-1}(N,M)}\ \prod_{k=1}^{L} w(X_k).
\end{equation}
For a block tuple $(X_1,\dots,X_L)$ define its empirical measure
\[
\mu = \frac{1}{L}\sum_{k=1}^L \delta_{X_k}.
\]
Let $\mathcal{M}_L$ denote the set of all empirical measures of the form above (i.e., $L$ atoms of weight $1/L$,
allowing repetitions). For $\mu\in\mathcal{M}_L$, let $\calS_{L-1}(\mu)$ be the set of (ordered) tuples
with empirical measure $\mu$. Then we can rewrite
\begin{equation}\label{eq:Z-by-type}
\overline{Z}
=\sum_{L=1}^{M+1}\ \sum_{\mu\in\mathcal{M}_L:\, \E_\mu[a]=\frac{N+1}{L},\ \E_\mu[b]=\frac{M+1}{L}}
|\calS_{L-1}(\mu)|\ \exp\!\Bigl(L\ \E_{\mu}[\log w(a,b)]\Bigr).
\end{equation}
(The mean constraints are forced because $\sum_k X_k=(N+1,M+1)$.)

\subsection{Exponential scale of the entropy term}

Write 
\begin{align*}
    n_{a,b} = n_{a,b}(\mu) := L\,\mu(a,b)\in\{0,1,\dots,L\}
\end{align*}
for the multiplicity of the value $(a,b)$ under $\mu$.
Then the number of tuples with empirical measure $\mu$ is the multinomial coefficient
\begin{equation}\label{eq:multinomial}
|\calS_{L-1}(\mu)|
=\frac{L!}{\prod_{(a,b)} n_{a,b}!}.
\end{equation}
Recall the Shannon entropy
\begin{align*}
    H(\mu) := -\sum_{a,b}\mu(a,b)\log\mu(a,b), \qquad 0\log0 := 0.
\end{align*}

\subsection{Entropy approximation and reduction to a variational problem}

Define the truncated-to-$L\le M$ sum
\begin{equation}\label{eq:Z-by-type-leM}
\overline{Z}^{\le M}
:=\sum_{L=1}^{M}\
\sum_{\mu\in\mathcal{M}_L:\, \E_\mu[a]=\frac{N+1}{L},\ \E_\mu[b]=\frac{M+1}{L}}
|\calS_{L-1}(\mu)|\ \exp\!\Bigl(L\ \E_{\mu}[\log w(a,b)]\Bigr),
\end{equation}
as well as the entropy-replaced sum
\begin{equation}\label{eq:Z-by-type-tilde}
\widetilde{Z}^{\le M}
:=\sum_{L=1}^{M}\
\sum_{\mu\in\mathcal{M}_L:\, \E_\mu[a]=\frac{N+1}{L},\ \E_\mu[b]=\frac{M+1}{L}}
\exp\!\Bigl(L H(\mu)+L\ \E_{\mu}[\log w(a,b)]\Bigr).
\end{equation}

\subsubsection{Uniform Stirling decomposition}

\begin{lemma}[Uniform Stirling decomposition]\label{lem:stirling-uniform}
Let $L\ge 1$ and $\mu\in\mathcal M_L$ with multiplicities $n_{a,b}:=L\mu(a,b)$.
Let $D:=\{(a,b):n_{a,b}\ge 1\}$ and $s:=|D|$.
Then
\begin{equation}\label{eq:stirling-decomp}
\log|\calS_{L-1}(\mu)|
= L H(\mu)
+\frac12\log(2\pi L)-\frac12\sum_{(a,b)\in D}\log(2\pi n_{a,b})
+\rho_L(\mu),
\end{equation}
where the remainder satisfies the explicit bounds
\begin{equation}\label{eq:rho-bounds}
-\sum_{(a,b)\in D}\frac{1}{12\,n_{a,b}}
\ \le\ \rho_L(\mu)\ \le\ \frac{1}{12L}.
\end{equation}
\end{lemma}

\begin{proof}
Apply the explicit Stirling bounds (valid for all integers $k\ge 1$),
\[
\sqrt{2\pi k}\Big(\frac{k}{e}\Big)^k \le k! \le
\sqrt{2\pi k}\Big(\frac{k}{e}\Big)^k e^{\frac{1}{12k}},
\]
to $\log|\calS_{L-1}(\mu)|=\log L!-\sum_{(a,b)\in D}\log(n_{a,b}!)$, and use
$L\log L-\sum_{a,b}n_{a,b}\log n_{a,b}=L H(\mu)$.
\end{proof}

\subsubsection{Support bound and subexponential number of types}

For each $M$, we define the alphabet $\mathcal A_M := \left\{ (a,b)\in\N_{>0}^2:1\le b\le a\le N+1 \right\}$, with
\begin{align*}
    K_M:=|\mathcal A_M|=\sum_{a=1}^{N+1} a=\frac{(N+1)(N+2)}{2}.
\end{align*}
Let $\mathcal I_M$ denote the set of pairs $(L,\mu)$ that appear in \eqref{eq:Z-by-type-leM}, i.e.,
\begin{align*}
    \mathcal I_M:=\Bigl\{(L,\mu):L\in\{1,\dots,M\},\
    \mu\in\mathcal M_L,\ \supp(\mu)\subseteq\mathcal A_M,\
    \E_\mu[a]=\tfrac{N+1}{L},\ \E_\mu[b]=\tfrac{M+1}{L}\Bigr\}.
\end{align*}
\begin{lemma}[Support bound]\label{lem:support-bound}
Let $(L,\mu)\in \mathcal I_M$, and write $D:=\supp(\mu)$. Then
\begin{equation}\label{eq:support-explicit}
    |D|
    \;\le\;
    \frac12\Big(\sqrt[3]{3(N+1)}+2\Big)^2.
\end{equation}
\end{lemma}
\begin{proof}
 Each distinct symbol $(a,b)$ in $D$ appears at least once, hence
\begin{align} \label{eq:cost_bd}
    \sum_{(a,b)\in\supp(\mu)} a \ \le\ \sum_{k=1}^L a_k \ =\ N+1.
\end{align}
For each $a\ge 1$ there are exactly $a$ choices of $b\in\{1,\dots,a\}$, and each such candidate pair $(a,b)$ has ``cost" $a$. We respectively denote the number of pairs with cost $\le t$ and the total cost of all of them via
\begin{align*}
    & N(t):=\sum_{a=1}^t a=\frac{t(t+1)}2;
    & C(t):=\sum_{a=1}^t a^2=\frac{t(t+1)(2t+1)}{6}.
\end{align*}
If $s=|D|$, let $t$ be the smallest integer with $N(t)\ge s$, so that $N(t-1) < s$. Then any set of $s$ pairs has total cost at least $C(t-1)$. Hence 
\begin{align*}
    C(t-1) \leq \sum_{(a,b)\in\supp(\mu)} a \stackrel{\eqref{eq:cost_bd}}{\leq} N+1.
\end{align*}
Using $C(u)\ge u^3/3$ for $u\ge 1$ gives $(t-1)^3/3\le N+1$, so $t\le \sqrt[3]{3(N+1)}+1$. Finally,
\begin{align*}
    s\le N(t)=\frac{t(t+1)}{2}\le \frac{(t+1)^2}{2}\le \frac12\Big(\sqrt[3]{3(N+1)}+2\Big)^2,
\end{align*}
producing the desired bound.
\end{proof}

\begin{lemma}[Subexponential number of relevant types]\label{lem:subexp-types}
We have $\lim_{M\to\infty}\frac{1}{N}\log|\mathcal I_M|=0$.
\end{lemma}

\begin{proof}
By \cref{lem:support-bound}, any feasible $\mu\in\mathcal M_L$
has $|\supp(\mu)|\le s_M$ where $s_M$ is the RHS of \eqref{eq:support-explicit}.
For each $s\le s_M$, the number of choices of a support set $D\subseteq\mathcal A_M$ of size $s$ is $\binom{K_M}{s}$,
and given $D$, the number of positive multiplicity vectors $(n_x)_{x\in D}$ with $\sum_{x\in D}n_x=L$
is $\binom{L-1}{s-1}$. Hence the number of possible $\mu$ satisfies
\begin{align*}
    |\mathcal I_M| \leq \sum_{L=1}^{M}\sum_{s=1}^{\lfloor s_M\rfloor}\binom{K_M}{s}\binom{L-1}{s-1}.
\end{align*}
Using $\binom{K}{s}\le (eK/s)^s$ and the subadditivity of $\log$ yields
\[
\log|\mathcal I_M|\ \le\ s_M\cdot O(\log M),
\]
and since $s_M=O(M^{2/3})$ while $N=M/\alpha$, we get $(1/N)\log|\mathcal I_M|\to 0$.
\end{proof}

\subsubsection{Entropy replacement and sum-to-sup}

\begin{proposition}[Entropy replacement at scale $N$]\label{prop:entropy-replacement}
Fix $\alpha \in (0,1)$. There exists an explicit deterministic $\delta_M=\delta_M(\alpha)$ with
$\delta_M/N\to 0$ such that
\begin{equation}\label{eq:SandStilde-sandwich}
e^{-\delta_M}\,\widetilde{Z}^{\le M}\ \le\ \overline{Z}^{\le M}\ \le\ e^{\delta_M}\,\widetilde{Z}^{\le M}.
\end{equation}
\end{proposition}

\begin{proof}
Fix an admissible summand $(L,\mu)$ in \eqref{eq:Z-by-type-leM} and write $D=\{n_{a,b}\ge 1\}$, $s=|D|$.
By \cref{lem:stirling-uniform},
\begin{align*}
    \log|\calS_{L-1}(\mu)|=LH(\mu)+\Delta_L(\mu)
\end{align*}
with $\Delta_L(\mu)$ given by \eqref{eq:stirling-decomp}.
Since $1\le n_{a,b}\le L\le M$ on $D$, \eqref{eq:rho-bounds} implies
\begin{align*}
    |\Delta_L(\mu)| \le \frac{s+1}{2}\log(2\pi M)+\frac{s}{12}+\frac{1}{12}.
\end{align*}
By \cref{lem:support-bound}, $s\le s_M:=\frac12(\sqrt[3]{3(N+1)}+2)^2$ for every feasible $\mu$, so $|\Delta_L(\mu)|\le \delta_M$ with
\begin{align*}
    \delta_M:=\frac{s_M+1}{2}\log(2\pi M)+\frac{s_M}{12}+\frac{1}{12}, \qquad \frac{\delta_M}{N}\to 0.
\end{align*}
Thus, it holds uniformly over all admissible $(L,\mu)$ that
\begin{align*}
    e^{-\delta_M}e^{LH(\mu)}\le |\calS_{L-1}(\mu)|\le e^{\delta_M}e^{LH(\mu)}.
\end{align*}
Multiplying by $e^{L\E_\mu[\log w]}$ and summing gives \eqref{eq:SandStilde-sandwich}.
\end{proof}

\begin{proposition}[Sum to supremum]\label{prop:sum-to-sup}
We have
\begin{equation}\label{eq:bar-sup}
\frac{1}{N}\log \overline{Z}^{\le M}
=\sup_{(L,\mu)\in\mathcal I_M}\ \frac{L}{N}\Big(H(\mu)+\E_\mu[\log w(a,b)]\Big)\ +\ o(1).
\end{equation}
\end{proposition}

\begin{proof}
Let $\widetilde T_M(L,\mu):=\exp(LH(\mu)+L\E_\mu[\log w])$. Since all summands are nonnegative,
\begin{align*}
    \sup_{(L,\mu)\in\mathcal I_M}\widetilde T_M(L,\mu)
    \ \le\ \widetilde{Z}^{\le M}\ \le\ |\mathcal I_M|\cdot \sup_{(L,\mu)\in\mathcal I_M}\widetilde T_M(L,\mu).
\end{align*}
Taking logs, dividing by $N$, and using \cref{lem:subexp-types} yields
\[
\frac{1}{N}\log\widetilde{Z}^{\le M}
=\sup_{(L,\mu)\in\mathcal I_M} \frac{L}{N}\Big(H(\mu)+\E_\mu[\log w]\Big)+o(1).
\]
Combine with \cref{prop:entropy-replacement} and $\delta_M/N\to 0$ to replace $\widetilde{Z}^{\le M}$ by $\overline{Z}^{\le M}$.
\end{proof}

\subsection{The limiting variational formula}

For $\rho\in(0,1)$ define
\begin{align}
    \Phi(\rho) :=\sup\Big\{H(\nu)+\E_{\nu}[\log w(a,b)]:\; &\supp(\nu)\subseteq \{(a,b):1\le b\le a\}, \label{eq:Phi-def}\\
    &\qquad \ |\supp(\nu)|<\infty,\ \E_\nu[a]=\tfrac{1}{\alpha\rho},\ \E_\nu[b]=\tfrac{1}{\rho}\Big\}.\nonumber 
\end{align}

\begin{lemma}[Bounded-atom correction]\label{lem:one-atom-correction}
Fix $\alpha \in (0,1)$. The following holds for all sufficiently large $M$. Let $L\in\{1,\dots,M\}$ and let $(X_1,\dots,X_L) \in \mathcal S_{L-1}(N,M)$, with $X_k=(a_k,b_k)$ and empirical measure $\mu=\frac1L\sum_{k=1}^L\delta_{X_k}$. Then there exists an index $k_\star$ with
\begin{equation}\label{eq:bounded-good-atom}
    2\le b_{k_\star} \leq a_{k_\star}.
\end{equation}
Define a modified tuple $\left( X'_1,\dots,X'_L \right)$ by setting
\begin{align*}
    & X'_{k_\star} := \left(a_{k_\star}-1,b_{k_\star}-1\right);
    & X'_k := X_k \text{ for } k\neq k_\star,
\end{align*}
and let $\widehat\mu$ be its empirical measure. Then
\begin{align*}
    & \E_{\widehat\mu}[a]=\frac{N}{L};
    & \E_{\widehat\mu}[b]=\frac{M}{L}.
\end{align*}
Moreover, writing $K_M:=\frac{(N+1)(N+2)}2$ for the alphabet size, we have
\begin{equation}\label{eq:correction-entropy-bound}
    \left| H(\widehat\mu) - H(\mu) \right| \leq \frac{1}{L} \log(K_M-1)+h\left( 1/L \right),
\end{equation}
and
\begin{equation}\label{eq:correction-logw-bound}
    \Big|\E_{\widehat\mu}[\log w]-\E_{\mu}[\log w]\Big| \leq \frac{1}{L}\max_{2\le b\le a\le N+1}\Big|\log w(a,b)-\log w(a-1,b-1)\Big|.
\end{equation}
In particular, uniformly over $L\in\{1,\dots,M\}$, as $M\to\infty$,
\begin{align} \label{eq:one-atom-correction-uniform-bd}
    \frac{L}{N}\Big(H(\widehat\mu)+\E_{\widehat\mu}[\log w]\Big)
    = \frac{L}{N}\Big(H(\mu)+\E_\mu[\log w]\Big)+o(1).
\end{align}
\end{lemma}

\begin{proof}
The existence of an atom satisfying \eqref{eq:bounded-good-atom} follows since $L \le M$ and
\begin{align*}
    \sum_{k=1}^{L} b_k = M+1.
\end{align*}
Replacing $(a,b)$ by $(a-1,b-1)$
decreases the total sums by $(1,1)$, hence changes $(N+1,M+1)$ to $(N,M)$, giving the mean identities. For \eqref{eq:correction-entropy-bound}, note that $\mu$ and $\widehat\mu$ are supported on an alphabet of size $\le K_M$ and differ in total variation by at most $1/L$ (one count is moved from one symbol to another), so the explicit Fannes--Audenaert bound gives the stated inequality. Since only one atom is changed, \eqref{eq:correction-logw-bound} is immediate. Finally, using $K_M=O(M^2)$, 
\begin{align*}
    \max_{2\le b\le a\le N+1}|\log w(a,b)-\log w(a-1,b-1)|=O(\log M),
\end{align*}
and $h(1/L)\ll (\log L)/L$, we obtain
\begin{align*}
    \frac{L}{N}\left| H(\widehat\mu)-H(\mu) \right|
    &\le \frac{1}{N}\log(K_M-1)+\frac{L}{N}h(1/L)=o(1), \\
    \frac{L}{N}\Big|\E_{\widehat\mu}[\log w]-\E_\mu[\log w]\Big|
    &\le \frac{1}{N}\max_{2\le b\le a\le N+1}\Big|\log w(a,b)-\log w(a-1,b-1)\Big|=o(1),
\end{align*}
which yields \eqref{eq:one-atom-correction-uniform-bd}.
\end{proof}
\begin{theorem}[Variational formula for $\overline{Z}$]\label{thm:main-variational}
Fix $\alpha \in (0,1)$ and define
\[
\mathsf R(\alpha):=\sup_{\rho\in(0,1)}\ \alpha\rho\,\Phi(\rho).
\]
Then we have
\begin{equation}\label{eq:rate-full}
\lim_{M\to\infty}\frac{1}{N}\log \overline{Z}=\mathsf R(\alpha).
\end{equation}

\end{theorem}

\begin{proof}
We prove matching lower and upper bounds.

\medskip\noindent\emph{Lower bound.}
Fix $\rho\in(0,1)$ and $\eta>0$, and choose $\nu$ with finite support such that
$\E_\nu[a]=1/(\alpha\rho)$, $\E_\nu[b]=1/\rho$ and $H(\nu)+\E_\nu[\log w]\ge \Phi(\rho)-\eta$.
Let $D=\supp(\nu)$. For each large $M$, set $L:=\lfloor \rho M\rfloor$ and choose integers $n_x\ge 0$ for $x\in D$ with
$\sum_{x\in D}n_x=L$ and
\begin{align} \label{eq:n_x_condition}
    |n_x/L-\nu(x)|\le |D|/L.
\end{align}
Let $\mu_M^{(0)} \in \mathcal M_L$ be the corresponding empirical measure. Then as $M \to \infty$,
\begin{align*}
    & H(\mu_M^{(0)})\to H(\nu);
    & \E_{\mu_M^{(0)}}[\log w]\to \E_\nu[\log w].
\end{align*}
The mean constraints for $\overline{Z}^{\le M}$ require $\E_\mu[a]=\frac{N+1}{L}$ and $\E_\mu[b]=\frac{M+1}{L}$,
whereas $\E_{\mu_M^{(0)}}[a]=\frac{A_0}{L}$ and $\E_{\mu_M^{(0)}}[b]=\frac{B_0}{L}$ with
\begin{align*}
    & A_0=\sum_{x \in D} n_x a(x);
    & B_0=\sum_{x \in D} n_x b(x).
\end{align*}
Since $\mu_M^{(0)}\to\nu$, we have $A_0=N+O(1)$ and $B_0=M+O(1)$. Here, the $O(1)$ terms follow from \eqref{eq:n_x_condition}.

We now adjust $\mu_M^{(0)}$ into a new empirical measure $\mu_M^{(1)}\in\mathcal M_L$ satisfying the \emph{exact} constraints by changing only $O(1)$ atoms. We reserve a bounded buffer of the three symbols
\begin{align*}
    x_0:=(1,1),\qquad x_1:=(2,1),\qquad x_2:=(2,2), \qquad\text{for which}\qquad \log w(x_i)=0.
\end{align*}
More precisely, choose a fixed constant $C$ large enough that $C>|\Delta a|+|\Delta b|$ for all sufficiently large $M$, where
\begin{align*}
    (\Delta a,\Delta b):=(N+1-A_0,M+1-B_0)\in\Z^2.
\end{align*}
When constructing the counts $n_x$ above, we instead require $\sum_{x\in D} n_x=L-3C$, and then adjoin $C$ copies each of $x_0,x_1,x_2$. This changes only $O(1)$ atoms, so still $\mu_M^{(0)}\to\nu$, and hence still $A_0=N+O(1)$ and $B_0=M+O(1)$. Now we can increase $(A_0,B_0)$ by $(1,0)$ (resp.\ $(1,1)$) by replacing one copy of $x_0$ by $x_1$ (resp.\ $x_2$), and decrease them by reversing these moves; performing $O(1)$ such moves yields exact means. Since only $O(1)$ atoms are changed and $\log w(x_i)=0$, we still have
\begin{align*}
    & H(\mu_M^{(1)})=H(\mu_M^{(0)})+o(1);
    & \E_{\mu_M^{(1)}}[\log w]=\E_{\mu_M^{(0)}}[\log w]+o(1).
\end{align*}
Therefore,
\[
H(\mu_M^{(1)})+\E_{\mu_M^{(1)}}[\log w]\ \ge\ \Phi(\rho)-\eta+o(1).
\]
Applying \eqref{eq:bar-sup} to the admissible pair $(L,\mu_M^{(1)})$ gives
\[
\liminf_{M\to\infty}\frac{1}{N}\log \overline{Z}^{\le M}
\ \ge\ \liminf_{M\to\infty}\frac{L}{N}\Big(H(\mu_M^{(1)})+\E_{\mu_M^{(1)}}[\log w]\Big)
\ \ge\ \alpha\rho(\Phi(\rho)-\eta).
\]
Since $\overline Z\ge \overline Z^{\le M}$, it follows that
\[
\liminf_{M\to\infty}\frac{1}{N}\log \overline{Z}
\ \ge\ \alpha\rho(\Phi(\rho)-\eta).
\]
Let $\eta\downarrow 0$ and take the supremum over $\rho\in(0,1)$ to obtain
\begin{equation}\label{eq:liminf-full}
\liminf_{M\to\infty}\frac{1}{N}\log \overline{Z}\ge \mathsf R(\alpha).
\end{equation}

\medskip\noindent\emph{Upper bound.}
Write
\begin{align*}
    \overline{Z}=\overline{Z}^{\le M}+\overline{Z}^{(M+1)},
\end{align*}
where
\begin{align*}
    \overline{Z}^{(M+1)}:=\sum_{(X_1,\dots,X_{M+1})\in \mathcal{S}_{M}(N,M)}\ \prod_{k=1}^{M+1} w(X_k).
\end{align*}
By \eqref{eq:bar-sup}, there exists $(L,\mu)\in\mathcal I_M$ such that
\[
\frac{1}{N}\log \overline{Z}^{\le M}
\le \frac{L}{N}\Big(H(\mu)+\E_\mu[\log w]\Big)+o(1).
\]
Set $\rho:=L/M\in(0,1]$. Since $L\le M$ and $\sum_{k=1}^{L} b_k=M+1$, there exists an atom with $b\ge 2$, so \cref{lem:one-atom-correction} applies. Thus there is $\widehat\mu\in\mathcal M_L$ with
\[
\E_{\widehat\mu}[a]=\frac{N}{L}=\frac{1}{\alpha\rho},\qquad \E_{\widehat\mu}[b]=\frac{M}{L}=\frac{1}{\rho},
\]
and, by \eqref{eq:one-atom-correction-uniform-bd},
\[
\frac{L}{N}\Big(H(\widehat\mu)+\E_{\widehat\mu}[\log w]\Big)
=\frac{L}{N}\Big(H(\mu)+\E_\mu[\log w]\Big)+o(1).
\]
Since $\widehat\mu$ is finitely supported, it is admissible in the definition of $\Phi(\rho)$, hence
\[
\frac{L}{N}\Big(H(\mu)+\E_\mu[\log w]\Big)
\le \alpha\rho\,\Phi(\rho)+o(1)
\le \mathsf R(\alpha)+o(1).
\]
Therefore
\begin{equation}\label{eq:limsup-full-leM}
    \limsup_{M\to\infty}\frac{1}{N}\log \overline{Z}^{\le M}\le \mathsf R(\alpha).
\end{equation}

It remains to control $\overline{Z}^{(M+1)}$. If $(X_1,\dots,X_{M+1})\in\mathcal S_M(N,M)$, then necessarily $b_k=1$ for every $k$, since each $b_k\ge 1$ and $\sum_{k=1}^{M+1} b_k=M+1$. Hence $w(X_k)=w(a_k,1)=1$ for all $k$, and
\begin{align*}
    \overline{Z}^{(M+1)}
    =\bigl|\{(a_1,\dots,a_{M+1})\in\N_{>0}^{M+1}:\ a_1+\cdots+a_{M+1}=N+1\}\bigr|
    =\binom{N}{M}.
\end{align*}
Therefore
\begin{equation}\label{eq:limsup-full-Mplus1-exact}
    \lim_{M\to\infty}\frac{1}{N}\log \overline{Z}^{(M+1)}=h(\alpha).
\end{equation}
To compare this with $\mathsf R(\alpha)$, fix $\rho\in(1/2,1)$ and set $\delta_\rho:=\rho^{-1}-1$. Let $B_\rho\in\{1,2\}$ satisfy
\begin{align*}
    & \Prob(B_\rho=2)=\delta_\rho; \\
    & \Prob(B_\rho=1)=1-\delta_\rho,
\end{align*}
so that $\E[B_\rho]=1/\rho$. Let $C_\rho$ be a geometric random variable on $\{0,1,2,\dots\}$ with mean
\begin{align*}
    m_\rho:=\frac{1-\alpha}{\alpha\rho},
\end{align*}
and independent of $B_\rho$, and define $A_\rho:=B_\rho+C_\rho$. Then $A_\rho\ge B_\rho$ almost surely and
\begin{align*}
    & \E[A_\rho]=\E[B_\rho]+\E[C_\rho]=\frac{1}{\rho}+\frac{1-\alpha}{\alpha\rho}=\frac{1}{\alpha\rho}; \\
    & \E[B_\rho]=\frac{1}{\rho}.
\end{align*}
By truncating $C_\rho$ and adjusting the top atom, we may approximate the law of $(A_\rho,B_\rho)$ by finitely supported admissible laws with the same moments and entropy arbitrarily close to $H(A_\rho,B_\rho)$. Since $w(a,b)\ge 1$, it follows that
\begin{align*}
    \Phi(\rho)\ge H(A_\rho,B_\rho)=H(B_\rho)+H(C_\rho).
\end{align*}
Now $H(B_\rho)\to 0$ as $\rho\uparrow 1$, while
\begin{align*}
    H(C_\rho)=\bigl(m_\rho+1\bigr)h\!\left(\frac{1}{m_\rho+1}\right) \longrightarrow \frac{1}{\alpha}h(\alpha)
\end{align*}
as $\rho\uparrow 1$. Therefore
\begin{align*}
    \mathsf R(\alpha)=\sup_{\rho\in(0,1)}\alpha\rho\,\Phi(\rho)
    \ge \limsup_{\rho\uparrow 1} \alpha\rho\bigl(H(B_\rho)+H(C_\rho)\bigr)
    =h(\alpha).
\end{align*}
Together with \eqref{eq:limsup-full-Mplus1-exact}, this yields
\begin{equation}\label{eq:limsup-full-Mplus1}
    \limsup_{M\to\infty}\frac{1}{N}\log \overline{Z}^{(M+1)}\le \mathsf R(\alpha).
\end{equation}
Finally, using $\overline{Z}=\overline{Z}^{\le M}+\overline{Z}^{(M+1)}$ and combining \eqref{eq:liminf-full}, \eqref{eq:limsup-full-leM}, and \eqref{eq:limsup-full-Mplus1}, we conclude that
\[
\lim_{M\to\infty}\frac{1}{N}\log\overline{Z}= \mathsf R(\alpha),
\]
which is \eqref{eq:rate-full}.
\end{proof}

\subsection{Solving the variational problem}\label{subsec:solve-variational}

By \cref{thm:main-variational}, it remains to evaluate the supremum
\begin{align*}
    \mathsf R(\alpha)=\sup_{\rho\in(0,1)}\alpha\rho\,\Phi(\rho),
\end{align*}
where $\Phi(\rho)$ was defined in \eqref{eq:Phi-def}. We solve this double optimisation (over $\nu$ and over $\rho$) via Lagrange multipliers in four steps: first the inner problem for fixed~$\rho$, then a closed-form evaluation of the partition function, then the outer problem in~$\rho$, and finally the resulting algebraic system.

\begin{lemma}[Exponential family optimizer]\label{lem:inner-optimizer}
Fix $\rho\in(0,1).$ The supremum defining $\Phi(\rho)$ is attained by a unique probability measure $\nu^*$ of the form
\begin{equation}\label{eq:inner-optimizer-form}
    \nu^*(a,b)\;=\;\frac{w(a,b)\,x^a\,y^b}{Z(x,y)},\qquad a\ge 1,\;1\le b\le a,
\end{equation}
where $x,y>0$ are the unique positive solutions of the moment equations $\E_{\nu^*}[a]=1/(\alpha\rho)$ and $\E_{\nu^*}[b]=1/\rho$, and
\begin{align*}
    Z(x,y)\;:=\;\sum_{a\ge 1}\sum_{b=1}^{a}\binom{a-1}{b-1}^{2}\,x^{a}\,y^{b}
\end{align*}
is the partition function. The optimal value is
\begin{equation}\label{eq:Phi-from-Z}
    \Phi(\rho)\;=\;\log Z(x,y)\;-\;\frac{1}{\alpha\rho}\,\log x\;-\;\frac{1}{\rho}\,\log y.
\end{equation}
\end{lemma}

\begin{proof}
    We first solve the inner variational problem for fixed $\rho$ by Lagrange multipliers over the space of probability measures satisfying the moment constraints, which identifies the optimizer $\nu^*$ and shows that it has the exponential family form \eqref{eq:inner-optimizer-form}. The functional 
    \begin{align*}
        \nu\mapsto H(\nu)+\E_\nu[\log w]
    \end{align*}
    is strictly concave on the convex set of probability measures satisfying the moment constraints since $H$ is strictly concave and $\E_\nu[\log w]$ is linear in $\nu$. Furthermore, for any probability measure $\nu$ satisfying these constraints, it follows from $w(a,b) \leq 4^a$ that
    \begin{align*}
        \E_\nu[\log w] \leq (\log 4)\,\E_\nu[a] = (\log 4)/(\alpha\rho) < \infty
    \end{align*}
    and, letting $(A,B)$ be a random vector with law $\nu$,
    \begin{align*}
        H(\nu) = H(A,B) = H(A) + H(B \mid A) \leq H(A) + \E\left[ \log A \right] < \infty,
    \end{align*}
    so the evaluation of the functional at any such $\nu$ is finite. Introducing Lagrange multipliers $\lambda_0$ for the normalization constraint $\sum_{a,b}\nu(a,b)=1$, $\lambda_1$ for $\E_\nu[a]=1/(\alpha\rho)$, and $\lambda_2$ for $\E_\nu[b]=1/\rho$, the first-order condition
    \begin{align*}
        \frac{\partial}{\partial\nu(a,b)} \Bigl[H(\nu)+\E_\nu[\log w]-\lambda_0\bigl(\textstyle\sum\nu-1\bigr) -\lambda_1\bigl(\E_\nu[a]-1/(\alpha\rho)\bigr)-\lambda_2\bigl(\E_\nu[b]-1/\rho\bigr)\Bigr]=0
    \end{align*}
    gives $-\log\nu(a,b)-1+\log w(a,b)-\lambda_0-\lambda_1 a-\lambda_2 b=0$, hence
    \begin{align*}
        \nu^*(a,b)=\frac{w(a,b)\,e^{-\lambda_1 a}\,e^{-\lambda_2 b}}{Z'},
    \end{align*}
    where $Z'=e^{1+\lambda_0}$ ensures normalization. It can be shown that
    \begin{enumerate}
        \item the Lagrange multiplier equations admit a solution, i.e., corresponding values $\lambda_0^*$, $\lambda_1^*$, and $\lambda_2^*$ for which the normalization and coordinate mean constraints are satisfied;
        \item the corresponding measure $\nu^*$ is the unique global maximizer of the objective functional.
    \end{enumerate}
    We do this in \cref{app:lagrange-multipliers-addendum}. Letting $\left( \lambda_1^*, \lambda_2^* \right)$ denote corresponding such values, setting the values $x:=e^{-\lambda_1^*}$ and $y:=e^{-\lambda_2^*}$ yields the form \eqref{eq:inner-optimizer-form} with $Z'=Z(x,y)$. 
    
    \medskip
    
    To evaluate the optimal value, we substitute this optimizer into the objective and then approximate it by admissible measures. Since 
    \begin{align*}
        \log\nu^*(a,b)=\log w(a,b)+a\log x+b\log y-\log Z,
    \end{align*}
    substituting $\nu^*$ into the functional to get $H(\nu^*)+\E_{\nu^*}[\log w]$ yields that
    \begin{align*}
        H(\nu^*)+\E_{\nu^*}[\log w] =-\E_{\nu^*}[\log\nu^*]+\E_{\nu^*}[\log w] =\log Z-(\log x)\E_{\nu^*}[a]-(\log y)\E_{\nu^*}[b],
    \end{align*}
    which is \eqref{eq:Phi-from-Z}. Though $\nu^*$ itself does not have finite support, we may approximate it via a sequence of finitely supported probability measures satisfying the moment constraints. Specifically, we define
    \begin{align*}
        & \mu_n(\cdot) := \nu^*\left( \cdot \mid E_n \right);
        & E_n := \left\{ (a,b) \in \N_{>0}^2: b \leq a \leq n \right\}.
    \end{align*}
    It directly follows from this construction that
    \begin{align*}
        & \E_{\mu_n}[a] \to \E_{\nu^*}[a] = 1/(\alpha \rho);
        & \E_{\mu_n}[b] \to \E_{\nu^*}[b] = 1/\rho.
    \end{align*}
    Therefore, we may construct a sequence of probability measures $\nu_n$ supported on the three tuples
    \begin{align*}
        (1,1), (2/ \lfloor \alpha \rho \rfloor, 1), (2/ \lfloor \alpha \rho \rfloor, 2/ \lfloor \alpha \rho \rfloor)
    \end{align*}
    such that there exist values $\epsilon_n \to 0$ for which the mixture probability measure
    \begin{align*}
        \nu_n^* := (1-\epsilon_n)\mu_n + \epsilon_n \nu_n
    \end{align*}
    satisfies the mean constraints $\E_{\nu_n}[a] = 1/(\alpha \rho)$ and $\E_{\nu_n}[b] = 1/\rho$. Then, we have that 
    \begin{align*}
        & \left( H(\nu^*)+\E_{\nu^*}[\log w] \right) - \left( H(\nu_n)+\E_{\nu_n}[\log w] \right) \stackrel{\eqref{eq:objective-functional-opt-relation}}{=} \KL\left( \nu_n \doubleline \nu^* \right) \\
        & \qquad \leq (1-\epsilon_n) \KL\left( \mu_n \doubleline \nu^* \right) + \epsilon_n \KL\left( \nu_n \doubleline \nu^* \right) = -(1-\epsilon_n)\log \nu^*(E_n) + \epsilon_n \KL\left( \nu_n \doubleline \nu^* \right) \ll 1.
    \end{align*}
    This yields \eqref{eq:Phi-from-Z}.
\end{proof}

\begin{lemma}[Closed-form partition function]\label{lem:partition-closed-form}
    For all $x,y>0$ with $(1-x-xy)^2>4x^2 y$,
    \begin{equation}\label{eq:partition-closed-form}
        Z(x,y)\;=\;\frac{xy}{\sqrt{(1-x-xy)^{2}-4x^{2}y}}.
    \end{equation}
    Moreover, $Z(x,y) = \infty$ for all $x, y > 0$ that fail to satisfy this condition.
\end{lemma}

\begin{proof}
We begin by rewriting the partition function using a shift of indices together with Vandermonde's identity. Shifting indices $m:=a-1$, $k:=b-1$ gives
\begin{equation}\label{eq:vandermonde-rewrite}
Z(x,y)=xy\sum_{m=0}^{\infty}x^{m}\sum_{k=0}^{m}\binom{m}{k}^{2}y^{k}.
\end{equation}
For fixed $m$, the inner sum equals $[t^{m}]\,(1+t)^{m}(1+yt)^{m}$. Indeed,
\begin{align*}
    & (1+t)^m=\sum_{i}\binom{m}{i}t^i;
    & (1+yt)^m=\sum_{j}\binom{m}{j}y^j t^j,
\end{align*}
and extracting the coefficient of $t^m$ from their product gives
\begin{align*}
    \sum_{k=0}^m \binom{m}{k}\binom{m}{m-k}y^k = \sum_{k}\binom{m}{k}^2 y^k.
\end{align*}
Hence, we have that
\begin{equation}\label{eq:diagonal-sum}
Z(x,y)=xy\sum_{m=0}^{\infty}[t^{m}]\,\bigl(x(1+t)(1+yt)\bigr)^{m}.
\end{equation}

\medskip We now evaluate the resulting expression via residue computations. We write $f(t):=x(1+t)(1+yt)$ and consider the ``diagonal sum''
\begin{align*}
    \Sigma:=\sum_{m\ge 0}[t^m]\,f(t)^m.
\end{align*}
Since we have that
\begin{align*}
    [t^m]f(t)^m=\frac{1}{2\pi i}\oint f(t)^m\,t^{-m-1}\,dt,
\end{align*}
summing the geometric series inside the integral yields
\begin{align*}
    \Sigma\;=\;\frac{1}{2\pi i}\oint\frac{dt}{t-f(t)},
\end{align*}
valid for a contour encircling the origin inside the region of convergence. The equation $t=f(t)$, i.e., $t=x(1+t)(1+yt)$, rearranges to
\begin{equation}\label{eq:diagonal-quadratic}
    Q(t):=xy\,t^{2}+(x+xy-1)\,t+x=0,
\end{equation}
a quadratic in $t$ with discriminant
\begin{align} \label{eq:closed-form-partition-fn-discriminant}
    D = D(x,y) := (1-x-xy)^{2}-4x^{2}y = \left( 1 - x(1+\sqrt{y})^2 \right)\left( 1 - x(1-\sqrt{y})^2 \right).
\end{align}
The two roots are
\begin{align*}
    t_{\pm}=\frac{(1-x-xy)\pm\sqrt{D}}{2xy},
\end{align*}
and $t_{-}$ is the small root enclosed by the contour. Since $f'(t)=x(1+y)+2xyt$, we have
\begin{align*}
    1-f'(t_{-})=1-x(1+y)-2xy\,t_{-} =1-x-xy-\bigl[(1-x-xy)-\sqrt{D}\bigr]=\sqrt{D}.
\end{align*}
In particular, $t - f(t)$ is increasing at $t_- \neq 0$, so there exist values of $t$ for which $|f(t)| < t$, and such a contour exists (e.g., via a circle with the appropriate radius). The residue at the simple pole $t_{-}$ of $1/(t-f(t))$ is thus $1/\sqrt{D}$, giving $\Sigma=1/\sqrt{D}$, so $Z=xy/\sqrt{D}$.

\medskip Finally, we show that once the discriminant condition fails, the series defining $Z(x,y)$ diverges. It is clear that $Z(x,y)$ is increasing in $x$. We fix $x,y > 0$ such that $D(x,y) \leq 0$. Taking $0 < x' < x$ small enough so that $D(x',y) > 0$ (which can be done, observed via the latter expression of \eqref{eq:closed-form-partition-fn-discriminant}), it holds that 
\begin{align*}
    \frac{x'y}{\sqrt{D(x',y)}} \leq Z(x,y).
\end{align*}
Choosing $x'$ arbitrarily close to $x_y := (1+\sqrt{y})^{-2} > 0$, for which $D(x_y, y) = 0$, implies the result.
\end{proof}

\begin{proposition}[Unique interior maximizer and normalization]\label{prop:Z-equals-one}
Let
\[
g(\rho):=\alpha\rho\,\Phi(\rho), \qquad \rho\in(0,1).
\]
Then \(g\) attains its maximum at a unique point \(\rho^*\in(0,1)\). If
\((x^*,y^*)\) denotes the pair of parameters from \cref{lem:inner-optimizer}
corresponding to \(\rho=\rho^*\), then
\begin{equation}\label{eq:Z-equals-one}
    Z(x^*,y^*)=1,
\end{equation}
and consequently
\begin{equation}\label{eq:R-explicit}
    \mathsf R(\alpha)\;=\;-\log x^*\;-\;\alpha\log y^*.
\end{equation}
\end{proposition}

\begin{proof}
For each \(\rho\in(0,1)\), let \(x=x(\rho)\) and \(y=y(\rho)\) be the unique
positive parameters given by \cref{lem:inner-optimizer}. By \eqref{eq:Phi-from-Z},
\begin{equation}\label{eq:g-rho-from-Z}
    g(\rho)=\alpha\rho\log Z(x,y)-\log x-\alpha\log y.
\end{equation}

We first compute \(g'(\rho)\). The Lagrange multipliers for the inner problem are
\[
\alpha_1=-\log x
\qquad\text{and}\qquad
\alpha_2=-\log y,
\]
corresponding respectively to the constraints
\[
\E_{\nu^*}[a]=\frac{1}{\alpha\rho},
\qquad
\E_{\nu^*}[b]=\frac{1}{\rho}.
\]
By the envelope theorem,
\[
\frac{d\Phi}{d\rho}
=
\alpha_1\frac{d(1/(\alpha\rho))}{d\rho}
+
\alpha_2\frac{d(1/\rho)}{d\rho}
=
\frac{\frac1\alpha\log x+\log y}{\rho^2}.
\]
Therefore,
\begin{align}
g'(\rho)
&=
\alpha\Phi(\rho)+\alpha\rho\,\Phi'(\rho) \notag\\
&=
\alpha\Bigl[\log Z-\frac{\frac1\alpha\log x+\log y}{\rho}\Bigr]
+
\alpha\frac{\frac1\alpha\log x+\log y}{\rho}
=
\alpha\log Z(x(\rho),y(\rho)).
\label{eq:g-prime-alpha-logZ}
\end{align}

Next we compute \(x(\rho)\), \(y(\rho)\), and \(Z(x(\rho),y(\rho))\) explicitly.
By \cref{lem:partition-closed-form},
\[
Z(x,y)=\frac{xy}{\sqrt{D(x,y)}},
\qquad
D(x,y):=(1-x-xy)^2-4x^2y.
\]
Hence
\[
\log Z(x,y)=\log x+\log y-\frac12\log D(x,y).
\]
On the other hand, by the exponential family form of $\nu^*,$ we have
\[
\E_{\nu^*}[a]=x\partial_x\log Z(x,y),
\qquad
\E_{\nu^*}[b]=y\partial_y\log Z(x,y).
\]
A direct computation gives
\[
x\partial_x\log Z(x,y)=\frac{1-x-xy}{D(x,y)},
\qquad
y\partial_y\log Z(x,y)=\frac{1-2x-xy+x^2(1-y)}{D(x,y)}.
\]
Therefore the moment constraints are equivalent to
\begin{equation}\label{eq:moment-system-a}
    \frac{1-x-xy}{D(x,y)}=\frac{1}{\alpha\rho},
\end{equation}
and
\begin{equation}\label{eq:moment-system-b}
    \frac{1-2x-xy+x^2(1-y)}{D(x,y)}=\frac{1}{\rho}.
\end{equation}
Solving \eqref{eq:moment-system-a}--\eqref{eq:moment-system-b} yields
\begin{equation}\label{eq:x-rho-explicit}
    x(\rho)=\frac{(1-\alpha)(2-2\alpha+\alpha\rho)}{2-\alpha\rho},
\end{equation}
and
\begin{equation}\label{eq:y-rho-explicit}
    y(\rho)=\frac{\alpha^2(2-\rho)(1-\rho)}{(1-\alpha)(2-2\alpha+\alpha\rho)}.
\end{equation}
Substituting \eqref{eq:x-rho-explicit}--\eqref{eq:y-rho-explicit} into the
closed form for \(Z\) gives
\begin{equation}\label{eq:Z-rho-explicit}
    Z(x(\rho),y(\rho))
    =
    \frac{\alpha(1-\rho)\sqrt{2-\rho}}
    {\sqrt{\rho}\,\sqrt{(2-\alpha\rho)(2-2\alpha+\alpha\rho)}}.
\end{equation}

In particular,
\[
Z(x(\rho),y(\rho))\longrightarrow\infty
\qquad\text{as }\rho\downarrow 0,
\]
and
\[
Z(x(\rho),y(\rho))\longrightarrow 0
\qquad\text{as }\rho\uparrow 1.
\]
Moreover, \(Z(x(\rho),y(\rho))=1\) is equivalent, after squaring
\eqref{eq:Z-rho-explicit}, to
\begin{equation}\label{eq:rho-quadratic}
    2\alpha^2\rho^2+(4\alpha-5\alpha^2-4)\rho+2\alpha^2=0.
\end{equation}
The polynomial on the left-hand side of \eqref{eq:rho-quadratic} has value
\(2\alpha^2>0\) at \(\rho=0\) and value
\[
2\alpha^2+(4\alpha-5\alpha^2-4)+2\alpha^2
=
4\alpha-\alpha^2-4
=
-(2-\alpha)^2<0
\]
at \(\rho=1\). Hence it has at least one root in \((0,1)\). Since its constant
term equals its leading coefficient, the product of its roots is \(1\), so it
can have at most one root in \((0,1)\). Therefore there is a unique
\(\rho^*\in(0,1)\) such that
\[
Z(x(\rho^*),y(\rho^*))=1.
\]

Because \(Z(x(\rho),y(\rho))\) is continuous, the uniqueness of \(\rho^*\) and
the above boundary limits imply
\[
Z(x(\rho),y(\rho))>1 \quad \text{for } 0<\rho<\rho^*,
\qquad
Z(x(\rho),y(\rho))<1 \quad \text{for } \rho^*<\rho<1.
\]
By \eqref{eq:g-prime-alpha-logZ}, it follows that
\[
g'(\rho)>0 \quad \text{for } 0<\rho<\rho^*,
\qquad
g'(\rho)<0 \quad \text{for } \rho^*<\rho<1.
\]
Thus \(g\) is strictly increasing on \((0,\rho^*)\) and strictly decreasing on
\((\rho^*,1)\), so \(g\) attains its maximum uniquely at \(\rho^*\in(0,1)\).

Finally, setting \((x^*,y^*):=(x(\rho^*),y(\rho^*))\), we have proved
\eqref{eq:Z-equals-one}. Plugging this into \eqref{eq:g-rho-from-Z} gives
\[
\mathsf R(\alpha)=g(\rho^*)=-\log x^*-\alpha\log y^*,
\]
which is \eqref{eq:R-explicit}.
\end{proof}

\begin{proposition}[Explicit solution of the algebraic system]\label{prop:explicit-solution}
Writing $c = 1/\alpha$ for algebraic convenience, for $c>1$ (equivalently $\alpha < 1$), the system $Z(x,y)=1$ together with the stationarity condition from
\cref{prop:Z-equals-one} has a unique solution $(x,y)\in(0,1)\times(0,\infty)$
given by $x=x_\alpha$ and $y=y_\alpha$ as defined in \cref{thm:main-annealed-intro}.
\end{proposition}

\begin{proof}
    By \cref{lem:partition-closed-form},
    $Z(x,y)=1$ is equivalent to 
    \begin{align*}
        xy=\sqrt{(1-x-xy)^2-4x^2y}.
    \end{align*}
    Squaring both sides yields
    \[
    x^2y^2=(1-x-xy)^2-4x^2y.
    \]
    Expanding the right-hand side and cancelling $x^2y^2$ yields
    \begin{equation}\label{eq:constraint-I}
    x^{2}(1-2y)-2x(1+y)+1=0,\tag{I}
    \end{equation}
    which can be solved for $y$ as
    $y=(1-x)^{2}/\bigl[2x(1+x)\bigr]$.
    
    \medskip On the surface $Z=1$, the optimality of $\mathsf R(\alpha)=-\log x-\alpha\log y$ requires
    stationarity with respect to $(x,y)$ constrained to \eqref{eq:constraint-I}.
    Writing
    \begin{align*}
        Q(x,y):=x^2(1-2y)-2x(1+y)+1,
    \end{align*}
    the Lagrange condition
    $\nabla[-\log x-\alpha\log y]=\lambda\,\nabla Q$ gives (after dividing the two components)
    \begin{equation}\label{eq:constraint-II}
    \frac{(1+y)-x(1-2y)}{y(x+1)}=c.\tag{II}
    \end{equation}
    This can be solved for $y$ as $y=(1-x)/\bigl[(c-1)+(c-2)x\bigr]$.
    
    \medskip Equating the two expressions for $y$ from \eqref{eq:constraint-I} and \eqref{eq:constraint-II}
    (and dividing by $1-x\neq 0$ for $c>1$) gives
    \[
    \frac{1-x}{2x(1+x)}=\frac{1}{(c-1)+(c-2)x}.
    \]
    Cross-multiplying:
    $(1-x)\bigl[(c-1)+(c-2)x\bigr]=2x(1+x)$,
    which expands and simplifies to the quadratic
    \begin{equation}\label{eq:quadratic-x}
    cx^{2}+3x-(c-1)=0.
    \end{equation}
    By the quadratic formula, the unique positive root is
    \[
    x=\frac{-3+\sqrt{9+4c(c-1)}}{2c}=\frac{-3+\sqrt{4c^{2}-4c+9}}{2c}.
    \]
    Since $\sqrt{4c^2-4c+9}=c\sqrt{9c^{-2}-4c^{-1}+4}=c\,\Delta(c)$, we obtain
    $x=\bigl(\Delta(c)-3c^{-1}\bigr)/2=x_c$.
    Substituting back into $y=(1-x)^2/[2x(1+x)]$ and using
    $\Delta-3/c=2x_c$ (hence $3/c+2-\Delta=2(1-x_c)$ and $2+\Delta-3/c=2(1+x_c)$), we get
    \[
    y_c=\frac{4(1-x_c)^2}{2\cdot 2x_c\cdot 2(1+x_c)}
    =\frac{(3c^{-1}+2-\Delta(c))^2}{2(\Delta(c)-3c^{-1})(2+\Delta(c)-3c^{-1})},
    \]
    matching the definition in \cref{thm:main-annealed-intro}.
    Since $\sqrt{4c^2-4c+9}>3$ for $c>1$, we have $x_c>0$, and since $0<x_c<1$ we have $y_c>0$.
\end{proof}

\begin{proof}[Proof of \cref{thm:main-annealed-intro}]
    By \cref{thm:main-variational},
    $\lim_{M\to\infty}\frac{1}{N}\log\overline{Z}=\mathsf R(\alpha)$.
    By \cref{prop:Z-equals-one,prop:explicit-solution},
    \begin{align*}
        \mathsf R(\alpha)=-\log x_\alpha-\alpha\log y_\alpha,
    \end{align*}
    which gives \eqref{eq:Z-bar-lim}. By the discussion at the beginning of \cref{sec:annealed_free_planted}, this proves \cref{thm:main-annealed-intro}.
\end{proof}

\section{Open Problems} \label{sec:open_problems}

We conclude the paper with several suggestions for future research.

\subsection{\texorpdfstring{A conjectural formula for $\fpl(\alpha)$}{A conjectural formula for the planted free energy}}

We recall that our main motivation in this paper is the identity \eqref{eq:del-chann-cap-connection}, which reduces the problem of computing the uniform capacity of the deletion channel to understanding the planted quenched free energy $f_{\mathrm{pl}}(\alpha)$ with $\alpha = 1-p$. While \cref{thm:main-annealed-intro} gives an exact formula for the corresponding annealed free energy, the problem of computing the quenched free energy is likely to be much more challenging. In particular, as suggested by \cref{fig:capacity-simulation}, we expect the following analogue of \cref{thm:weak_law} to hold. In words, \cref{conj:planted-jensen-gap} states that the Jensen gap corresponding to the planted Random Subsequence Model is also always nontrivial, just as it was for the null model.
\begin{conjecture} \label{conj:planted-jensen-gap}
    It holds for any $\alpha \in (0, 1)$ that
    \begin{align*}
        \fpl(\alpha) < \fplann(\alpha).
    \end{align*}
\end{conjecture}
Given the suspected difficulty of exactly determining the value $\fpl(\alpha)$, a first step is to at least aim for a convincing conjectural description of this quantity. We illustrate why even non-rigorously deriving such a prediction for this free energy is challenging, by considering two very natural approaches to this problem and discussing the key obstructions that each approach encounters.

\medskip

\noindent {\bf The replica method.} In light of \cref{thm:main-annealed-intro}, perhaps the most natural candidate method for deriving a prediction is with a replica calculation. Let $Z = Z_{X',Y}$ be under the null law on independent uniform strings $X' \in \{0,1\}^N$ and $Y \in \{0,1\}^M$. By Nishimori's identity, if $(X,Y)$ are distributed according to the planted law, we have
\begin{align*}
    \E[\log Z_{X,Y}] = \frac{\E[Z\log Z]}{\E[Z]} = \left.\frac{\partial}{\partial q} \log \E[Z^q]\right|_{q=1}.
\end{align*}
Thus, a replica computation for $\fpl(\alpha)$ would begin by studying
\begin{align*}
    \Phi_r(\alpha) := \lim_{N\to\infty}\frac1N\log \E[Z^r] \qquad\text{for integers } r \ge 1,
\end{align*}
and then seeking (non-rigorous) continuation in the replica number. The difficulty is that the integer moments already appear to be highly nontrivial. Indeed, for $r \in \N$,
\begin{align} \label{eq:replica-moment-calc}
    \E[Z^r] = \sum_{\sigma^1,\dots,\sigma^r\in\Sigma_{N,M}} \Prob\left( X_{\sigma^1} = \cdots = X_{\sigma^r}=Y \right).
\end{align}
For a fixed $r$-tuple $\left(\sigma^1,\dots,\sigma^r\right)$, this probability depends on the full structure of the bipartite union graph $G = G\left(\sigma^1,\dots,\sigma^r\right)$ on vertex set $[M] \sqcup [N]$ obtained by connecting each $j \in [M]$ to those $\sigma^a(j) \in [N]$ used by the replicas for $a \in [r]$. More precisely, we have that
\begin{align*}
    \Prob\left( X_{\sigma^1} = \cdots = X_{\sigma^r}=Y \right) = 2^{\kappa(G)-(M+N)},
\end{align*}
where $\kappa(G)$ is the number of connected components of $G$. Thus the contribution of a replica configuration is not determined merely by a pairwise overlap matrix on configurations in $\Sigma$ for values $r \geq 3$, as was the case in \cref{sec:annealed_free_planted}, which is equivalent to the $r=2$ case of the above computation. Starting at $r=3$, tuples with the same pairwise overlap statistics can have different union graph topology and hence different weights in \eqref{eq:replica-moment-calc}. In particular, there does not seem to be a candidate finite-dimensional order parameter in terms of which the replica computation can be closed.

\medskip

\noindent {\bf The cavity method.} Another appealing route is via the cavity method. Starting from the recursion of \eqref{eq:Z-recurrence} applied to the entire strings $(X,Y) \in \{0,1\}^N \times \{0,1\}^M$, i.e., 
\begin{align*}
    Z_{N,M} = \1\{X_N=Y_M\} Z_{N-1,M-1}+Z_{N-1,M},
\end{align*}
we can eventually obtain
\begin{align*}
    \E\left[\log Z_{N,M} \right] - \E \left[ \log Z_{N-1,M} \right] = -\E\left[ \log \langle \sigma_1 \neq 1 \rangle_{N,M} \right],
\end{align*}
where $\langle \cdot \rangle_{N,M}$ denotes the Gibbs average, i.e., the expectation with respect to the uniform law over the (random) set $S_{X,Y}$. Assuming that the cavity field $\langle \sigma_1 \neq 1 \rangle_{N,M}$ converges in law to a random variable $P(\alpha)$ as $N,M\to\infty$ with $M/N = \alpha$, then that would imply
\begin{align*}
    f_{\mathrm{pl}}(\alpha) = -\int_0^{1-\alpha} \E \left[ \log P\!\left(\frac{\alpha}{\alpha+x}\right) \right] \,dx.
\end{align*}
Thus, the challenge in this approach is to derive a plausible closed-form description for the limit law of the cavity field $P(\alpha)$. In the context of mean-field spin glasses, this is often obtained as a consequence of a self-consistency relation. We have not been able to find such a relation.



\subsection{Mean-field versus rank-one free energies}

The natural mean-field version of the null Random Subsequence Model is where one replaces the ``rank one'' matrix $B_{ij} = \1\{X_i = Y_j\}$ by a matrix with i.i.d. entries drawn from a common distribution $\calD$ (see \cref{subsubsec:mean-field}). The choice of $\calD$ closest to the true Random Subsequence Model is $\calD = \Unif\{0,1\}$, which has been studied under the name ``Bernoulli Matching Model'' in relation to the longest common subsequence problem \cite{boutetdemonvel1999extensive, majumdar2005exact}. The case where $\calD = \text{Gamma}(a,b)$ corresponds to the Strict-Weak lattice polymer \cite{corwin2015strict}. Let $f^{\mathrm{BMM}}(\alpha)$ and $f^{\mathrm{SW}(a,b)}(\alpha)$ denote their (quenched) free energies, respectively.
To what extent do these quantities relate to or control the Random Subsequence Model's free energy? We pose the following conjecture.
\begin{conjecture} \label{conj:MF-rank-one-comp}
    For all $\alpha \in [0,1]$, it holds that
    \begin{align*}
        f_{\mathrm{null}}(\alpha)\leq f^{\mathrm{BMM}}(\alpha) \leq f^{\mathrm{SW}(1,1/2)}(\alpha).
    \end{align*}
\end{conjecture}

We recall that the experience in the directed polymers literature suggests that only problems with special algebraic structure admit exact analytic solutions, which represents an important barrier to obtaining exact formulas for the free energy of the Random Subsequence Model. Moreover, while the solvable Strict-Weak Polymer Model is an analogue of the null Random Subsequence Model, as far as the authors are aware, no solvable analogue (in this sense) of the \emph{planted} Random Subsequence Model is known. Indeed, the planted structure already seems to break the requisite algebraic structure. The next open problem suggests such an analogue. We see this model as breaking the algebraic structure in the minimal way, and hence view it as a promising stepping stone towards the Random Subsequence Model.
\begin{problem}
    Let $B \in \R_{+}^{N\times M}$ have i.i.d. $\mathrm{Gamma}(a,b)$ entries. Independently, sample $\sigma^* \sim \Sigma_{NM},$ and for every $j=1,\dots,M,$ over-write $B_{\sigma(j),j} = 1.$ Let $Z_{NM}$ be the corresponding free energy obtained from \eqref{eq:Z-B-recurrence}, and define
    \begin{align*}
        f^{\mathrm{SW}(a,b)}_{\mathrm{pl}}(\alpha) = \lim_{\substack{N,M\to\infty \\ M/N = \alpha}} \frac{1}{N}\E\left[ \log Z_{NM} \right].
    \end{align*}
    Find an exact analytic formula for $f^{\mathrm{SW}(a,b)}_{\mathrm{pl}}(\alpha)$.
\end{problem}

\subsection{Asymptotics in the likely deletion regime}

\cref{cor:pos_mutual_info} establishes that the uniform capacity of the deletion channel is strictly positive for every $p < 1$. However, the explicit lower bound extracted from our argument in \cref{thm:quantitative_capacity_bound} is too small to capture the asymptotic behavior of $\fpl(\alpha)$ as $\alpha = 1-p \to 0$. It remains open to determine the asymptotic order of magnitude of $\fpl(\alpha)$ as $\alpha \to 0$, which would be very interesting. 

\section*{Acknowledgments}

R.J. would like to thank Robin Pemantle for early encouragement to work on this problem and for responding to several ideas and questions that he had over its duration. We both especially thank Brice Huang, and F.P. especially thanks Hang Du, for several valuable discussions about this project. We are grateful to Amir Dembo, Nike Sun, Shuangping Li, and Tselil Schramm for helpful conversations while this work was in development. Finally, we thank Timo Sepp\"al\"ainen for answering our questions on the work \cite{corwin2015strict}.

\begin{appendices}

\appendix

\section{\texorpdfstring{Proof of Equation \eqref{eq:del-chann-cap-connection}}{Proof of Capacity Identity}} \label{app:del-chann-cap-proof}

\begin{proof}
Let $D \in \{0,1\}^N$ be the deletion pattern applied by the channel to obtain $Y$ from $X.$ For each $i=1,\dots,N,$ we set $D_i=1$ if bit $i$ of $X$ was deleted, and $D_i=0$ otherwise; note that $D\sim \Ber(p)^{\otimes N}.$ We have
\begin{align*}
    I(X;Y) &= H(Y) - H(Y \mid X) \\
    &= H(Y) - (H(D,Y \mid X) - H(D \mid X,Y)) \\
    &= H(Y) - H(D \mid X) + H(D \mid X,Y) \\
    &= H(Y) - H(D) + H(D \mid X,Y),
\end{align*}
where in the third step we used that $Y$ is deterministic given $X$ and $D,$ and in the last step we used that $D$ and $X$ are independent. When $X$ is uniformly random, $Y$ is uniform in $\{0,1\}^{|Y|},$ with $\E |Y| =(1-p)N.$ Setting $\alpha := 1-p,$ this leads to
\begin{align*}
    \lim_{N\to \infty} \frac{1}{N}I(X;Y) &= \alpha \log 2 - h(\alpha) + \lim_{N\to\infty} \frac{1}{N}H(D \mid X,Y).
\end{align*}
It remains to show that $H(D \mid X,Y) = \E[\log Z_{X,Y}^{\mathrm{pl}}]$. Given $X$ and $Y,$ define the set of deletion patterns consistent with the observation:
\[
\mathcal{D}(X,Y) = \{D \in \{0,1\}^N : D(X) = Y\},
\]
where $D(X)$ denotes the string obtained by deleting bit $X_i$ whenever $D_i = 1.$ Since $D \sim \Ber(p)^{\otimes N}$ independently of $X,$ and all $D \in \mathcal{D}(X,Y)$ have the same number of ones (namely $N - |Y|$), the conditional distribution of $D$ given $(X,Y)$ is uniform on $\mathcal{D}(X,Y).$ Therefore 
\begin{align*}
    H(D \mid X,Y) = \E[\log |\mathcal{D}(X,Y)|].
\end{align*}
Finally, there is a natural bijection between $\mathcal{D}(X,Y)$ and the embedding set $S_{X,Y}$: each $D \in \mathcal{D}(X,Y)$ determines a unique $\sigma \in S_{X,Y}$ by letting $\sigma(j)$ be the index of the $j$-th retained bit, and vice versa. Hence $|\mathcal{D}(X,Y)| = Z_{X,Y}^{\mathrm{pl}},$ and $H(D \mid X,Y) = \E[\log Z_{X,Y}^{\mathrm{pl}}]$.

It remains to verify that the limit on the right-hand side of \eqref{eq:del-chann-cap-connection} is well-defined with $M$ fixed at $\alpha N$ rather than random. In the deletion channel, $M = |Y| \sim \Bin(N,\alpha)$ concentrates around $\alpha N$ with deviations of order $\sqrt{N}.$ We couple a random $\lfloor\alpha N\rfloor$-subsequence of $X$ with the deletion channel output $Y$ by inserting or deleting $|M - \alpha N|$ bits. Each single-bit change affects $\log Z_{X,Y}$ by at most $\log N$. Since $\E[|M - \alpha N|] = O(\sqrt{N}),$ this gives
\[
\frac{1}{N}\left|\E_M\left[\E[\log Z_{X,Y} \mid M]\right] - \E[\log Z_{X, Y'}]\right| = O\!\left(\frac{\sqrt{N}\log N}{N}\right) = o(1),
\]
where $Y'$ denotes a uniform random $\lfloor\alpha N\rfloor$-subsequence of $X.$
\end{proof}

\section{\texorpdfstring{A Combinatorial Consequence of \cref{thm:weak_law}}{Addendum to Section \ref{subsec:defns_and_notation}}} \label{app:no-double-phase-transition}

The existence of the weak limit of \eqref{eq:weak_law_null} alone is enough to easily deduce the relative asymptotic behavior between $\fnull(\alpha)$ and $\fnullann(\alpha)$ in the $\alpha < 1/2$ regime. Indeed, we let $\alpha < \tilde\alpha < 1/2$, and we assume that $\fnull(\tilde\alpha) = \fnullann(\tilde\alpha)$. For $(X',Y)$ drawn from the null model with density parameter $\alpha$, we can write
\begin{align} \label{eq:alpha_tilde_decomp}
    Z_{X',Y} = \binom{N - \alpha N}{(\alpha/\tilde\alpha) N - \alpha N}^{-1} \sum_{S \in \binom{[N]}{(\alpha/\tilde\alpha) N}} Z_{X'|_S, Y}.
\end{align}
Each summand on the RHS of \eqref{eq:alpha_tilde_decomp} has the same law as the null partition function with density parameter $\tilde\alpha$, where the embedded string is of length $\alpha N$. It now follows from a standard argument (e.g., via invoking Markov's inequality to control the number of ``bad summands'' that do not behave like $(\alpha / \tilde\alpha) N \fnullann(\tilde\alpha)$ in the exponential) that typically, most of these summands behave like $(\alpha / \tilde\alpha) N \fnullann(\tilde\alpha)$ in the exponential. By comparing to the expectation of \eqref{eq:alpha_tilde_decomp}, this observation readily yields that $\fnull(\alpha) = \fnullann(\alpha)$. Thus, it is necessarily the case that either
\begin{itemize}
    \item there exists some critical value $\alpha_* \leq 1/2$ for which 
    \begin{align*}
        \begin{cases}
            \fnull(\alpha) < \fnullann(\alpha) & \alpha > \alpha_*, \\
            \fnull(\alpha) = \fnullann(\alpha) & \alpha < \alpha_*;
        \end{cases}
    \end{align*}

    \item it holds that $\fnull(\alpha) < \fnullann(\alpha)$ for all $\alpha < 1/2$.
\end{itemize}
In this sense, \cref{thm:weak_law} exactly characterizes the relative asymptotic behavior of $\fnull(\alpha)$ and $\fnullann(\alpha)$. In particular, it shows that there is no double phase transition past $\alpha < 1/2$, i.e., that the expectation always exhibits an exponential gap with the typical partition function which can be readily shown (via comparing the upper and lower bounds of \cref{thm:weak_law}) to diminish as $\alpha \to 0$. We note that this is in contrast to other combinatorial properties of the null variant of the Random Subsequence Model which emerge for smaller values of $\alpha$. For example, via a simple adaptation of the greedy algorithm, it is easy to show that $\alpha = 1/3$ is the threshold for $X'$ to typically contain every string in $\{0,1\}^M$ as a subsequence.

\section{\texorpdfstring{Proof of \cref{prop:regular_alignment_property}}{Proof of Proposition \ref{prop:regular_alignment_property}}} \label{app:RAP_proof}

\begin{proof}
    We fix a typical string $x \in \{0,1\}^N$. Towards establishing \eqref{eq:regular_alignment_aligned_set}, we define 
    \begin{align*}
        Z = \big(Z^{(1)}, \dots, Z^{(B)}\big)
    \end{align*}
    to denote a random string resulting from independently including each bit of $x$ with probability $\alpha$ (i.e., $Z = \BDC_{1-\alpha}(x)$), with $Z^{(i)}$ the part of $Z$ which corresponds to the block $x^{(i)}$ for each $i \in [B]$. It is clear that this random string $Z$ has law $\Prob_x$. It holds from the multiplicative Chernoff bound that
    \begin{align} \label{eq:RAP_mult_chernoff}
        q_{\text{len}} := \Prob\left( \big| |Z^{(1)}| - \alpha b \big| > \delta \alpha b \right) \stackrel{\text{(Chernoff)}}{\leq} 2\exp\left( - \frac{\delta^2 \alpha b}{3} \right) = 2\exp\left( -\frac{\alpha b^{2\epsilon}}{3} \right) \stackrel{\text{($b$ large)}}{\leq} b^{-\epsilon} = \gamma.
    \end{align}
    It thus follows from the additive Chernoff bound that
    \begin{align}
        & \Prob\left( (Z^{(1)}, \dots, Z^{(B)}) \notin \mathcal{NE}_\ind(Z) \right) \leq \Prob\left( \sum_{i=1}^B \1\left\{\big| |Z^{(i)}| - \alpha b \big| > \delta \alpha b \right\} > \gamma B \right) \nonumber \\
        & \quad \stackrel{\text{(Chernoff)}}{\leq} \exp\left( -B \cdot \KL\left( \gamma \doubleline q_{\text{len}} \right) \right) = \exp\left( -\frac{N}{b} \left[ \gamma \log\left( \frac{\gamma}{q_{\text{len}}} \right) + (1-\gamma) \log\left( \frac{1-\gamma}{1-q_{\text{len}}} \right) \right] \right) \nonumber \\
        & \quad \stackrel{\eqref{eq:RAP_mult_chernoff}}{\leq} \exp\left( -\frac{N}{b} \left[ \gamma \left( \log \gamma - \log 2 + \frac{\alpha b^{2\epsilon}}{3} \right) + (1-\gamma) \log\left( \frac{1-\gamma}{1-q_{\text{len}}} \right) \right] \right) \nonumber \\
        & \quad = \exp\left( -\frac{N}{b} \left[ b^{-\epsilon} \left( -\epsilon \log b - \log 2 + \frac{\alpha b^{2\epsilon}}{3} \right) + (1-b^{-\epsilon}) \log\left( \frac{1-b^{-\epsilon}}{1-o_b(1)} \right) \right] \right) \nonumber \\
        & \quad \stackrel{\text{($b$ large)}}{\leq} \exp\left( - \frac{N}{b} \cdot \frac{\alpha b^\epsilon}{6} \right) = e^{-\Omega(N)}. \label{eq:RAP_additive_chernoff}
    \end{align}
    Next, we define the (deterministic) collection of indices
    \begin{align} \label{eq:RAP_I_lb}
        \mathcal I_x := \left\{ i \in [B] : \Delta\big( x^{(i)} \big) \geq \sqrt{b} \right\} \stackrel{\text{($x$ typical)}}{\implies} \left| \mathcal I_x \right| \geq B/10.
    \end{align}
    If it held that $\big(Z^{(1)}, \dots, Z^{(B)}\big) \in \mathcal{NE}_\ind(Z)$, then
    \begin{align} \label{eq:regular_alignment_decomp_bd}
        T_\ind(x, Z) \geq \frac{1}{B} \sum_{i=1}^B A_\loc\big(x^{(i)}, Z^{(i)}\big) = \frac{1}{B} \sum_{i \in \mathcal I_x} A_\loc\big(x^{(i)}, Z^{(i)}\big) + \frac{1}{B} \sum_{i \notin \mathcal I_x} A_\loc\big(x^{(i)}, Z^{(i)}\big).
    \end{align}
    We now sequentially derive lower bounds on the two summands of \eqref{eq:regular_alignment_decomp_bd} which hold with overwhelming probability. Towards this end, we define independent random variables $\zeta_1$ and $\zeta_0$ such that
    \begin{align} \label{eq:RAP_binom_RVs}
        & \zeta_1 \sim \Bin\left( \frac{b+\sqrt{b}}{2} , \alpha \right);
        & \zeta_0 \sim \Bin\left( \frac{b-\sqrt{b}}{2} , \alpha \right),
    \end{align}
    for which the central limit theorem and Slutsky's theorem readily imply that
    \begin{align*}
        & \frac{\zeta_1 - \frac{\alpha b}{2}}{\sqrt{b}} \xrightarrow[b \to \infty]{d} \mathcal N\left( \frac{\alpha}{2}, \frac{\alpha}{2}(1 - \alpha) \right);
        & \frac{\zeta_0 - \frac{\alpha b}{2}}{\sqrt{b}}  \xrightarrow[b \to \infty]{d} \mathcal N\left( -\frac{\alpha}{2}, \frac{\alpha}{2}(1 - \alpha) \right).
    \end{align*}
    We begin with the first summand of \eqref{eq:regular_alignment_decomp_bd}. We fix $i \in \mathcal I$, and we assume without loss of generality that $\maj(x_i) = 1$. It now readily follows that, with the inequality below relying on the fact that $\Delta\big( x^{(i)} \big) \geq \sqrt{b}$, 
    \begin{align}
        & \Prob\left( A_\loc\big(x^{(i)}, Z^{(i)}\big) = 1 \right) \nonumber \\
        & \qquad = \Prob\left( \{ \maj\big(Z^{(i)}\big) = 1 \} \cap \{ \delta \Delta\big(Z^{(i)}\big) \geq 1 \} \right) = \Prob\left( \{ \maj\big(Z^{(i)}\big) = 1 \} \cap \{ \Delta\big(Z^{(i)}\big) \geq b^{1/2-\epsilon} \} \right) \nonumber \\
        & \qquad \geq \Prob\left( \zeta_1 - \zeta_0 \geq b^{1/2-\epsilon} \right) = \Prob\left( \frac{\zeta_1 - \frac{\alpha b}{2}}{\sqrt{b}} - \frac{\zeta_0 - \frac{\alpha b}{2}}{\sqrt{b}} \geq b^{-\epsilon} \right) \nonumber \\
        & \qquad \xrightarrow{b \to \infty} \Prob\left( \mathcal N\left( \alpha, \alpha(1 - \alpha) \right) \geq 0 \right) \stackrel{\eqref{eq:alignment_constant}}{=} \frac{1}{2} + \beta(\alpha) \stackrel{\eqref{eq:alignment_constant}}{=} \frac{1}{2} + 40\beta^\star(\alpha), \label{eq:RAP_first_summand_beta_conv}
    \end{align}
    where it is clear that the constant $\beta(\alpha) > 0$. This implies that, assuming (from \eqref{eq:RAP_first_summand_beta_conv}) $b$ is large enough such that
    \begin{align*}
        \Prob\left( A_\loc\big(x^{(i)}, Z^{(i)}\big) = 1 \right) \geq \frac{1}{2} + \frac{3\beta(\alpha)}{4}
    \end{align*}
    holds (with this guarantee holding uniformly over all $i \in [B]$),
    \begin{align}
        & \Prob\left( \frac{1}{|\mathcal I|} \sum_{i \in \mathcal I_x} A_\loc\big(x^{(i)}, Z^{(i)}\big) \leq \frac{1+\beta(\alpha)}{2} \right) \leq \Prob\left( \sum_{i \in \mathcal I_x} \1\big\{ A_\loc(x^{(i)}, Z^{(i)}) = 1 \big\} \leq |\mathcal I_x| \cdot \frac{1+\beta(\alpha)}{2} \right) \nonumber \\
        & \quad \leq \Prob\left( \sum_{i \in \mathcal I_x} \1\big\{ A_\loc(x^{(i)}, Z^{(i)}) = 1 \big\} - \E\left[ \sum_{i \in \mathcal I_x} \1\big\{ A_\loc(x^{(i)}, Z^{(i)}) = 1 \big\} \right] \leq -|\mathcal I_x| \cdot \frac{\beta(\alpha)}{4} \right) \nonumber \\
        & \stackrel{\text{(Hoeffding)}}{\leq} \exp\left( - \frac{2|\mathcal I_x|^2 \beta(\alpha)^2}{16|\mathcal I_x|} \right) \stackrel{\eqref{eq:RAP_I_lb}}{\leq} \exp\left( -\frac{B \cdot \beta(\alpha)^2}{80} \right) = \exp\left( -\frac{N \cdot \beta(\alpha)^2}{80b} \right) = e^{-\Omega(N)}. \label{eq:RAP_summand_one_lb}
    \end{align}
    Combining \eqref{eq:RAP_additive_chernoff} and \eqref{eq:RAP_summand_one_lb}, it holds with probability $1-e^{-\Omega(N)}$ that
    \begin{align}
        T_\ind(x, Z) & \stackrel{\eqref{eq:regular_alignment_decomp_bd}}{\geq} \frac{1}{B} \sum_{i \in \mathcal I_x} A_\loc\big(x^{(i)}, Z^{(i)}\big) + \frac{1}{B} \sum_{i \notin \mathcal I_x} A_\loc\big(x^{(i)}, Z^{(i)}\big) \nonumber \\
        & \geq \frac{|\mathcal I_x|}{B}\left( \frac{1+\beta(\alpha)}{2} \right) + \frac{1}{B} \sum_{i \notin \mathcal I_x} A_\loc\big(x^{(i)}, Z^{(i)}\big). \label{eq:RAP_decomp_bd_1}
    \end{align}
    We now consider the second summand of \eqref{eq:RAP_decomp_bd_1}. We proceed under the further assumption on $x$ that
    \begin{align} \label{eq:RAP_summand_2_further_assumption}
        \frac{|\mathcal I_x|}{B}\left( \frac{1+\beta(\alpha)}{2} \right) < \frac{1}{2} + \frac{\beta(\alpha)}{4} \iff |\mathcal I_x^c| > B\left( 1 - \frac{\frac{1}{2} + \frac{\beta(\alpha)}{4}}{\frac{1}{2} + \frac{\beta(\alpha)}{2}} \right) = B\left( \frac{\beta(\alpha)}{2\left(1+\beta(\alpha)\right)} \right) = \Omega(N),
    \end{align}
    as if this assumption fails and \eqref{eq:RAP_decomp_bd_1} holds, then it is clear that
    \begin{align*}
        T_\ind(x, Z) \stackrel{\eqref{eq:RAP_decomp_bd_1}}{\geq} \frac{|\mathcal I_x|}{B}\left( \frac{1+\beta(\alpha)}{2} \right) \geq  \frac{1}{2} + \frac{\beta(\alpha)}{4} \geq \frac{1}{2} + \beta^\star(\alpha).
    \end{align*}
    A similar argument as that for the first summand in \eqref{eq:RAP_first_summand_beta_conv} (with the modification that we take the analogues of \eqref{eq:RAP_binom_RVs} to have mean $b/2$ instead) yields that for large $b$, it holds uniformly over $i \notin \mathcal I_x$ that
    \begin{align*}
        \Prob\left( A_\loc\big(x^{(i)}, Z^{(i)}\big) = 1 \right) \geq \frac{1}{2} - \frac{\beta(\alpha)}{72}.
    \end{align*}
    Then by a similar application of Hoeffding's inequality as in \eqref{eq:RAP_summand_one_lb}, it holds that
    \begin{align} \label{eq:RAP_summand_two_lb}
        \Prob\left( \frac{1}{|\mathcal I_x^c|} \sum_{i \notin \mathcal I_x} A_\loc\big(x^{(i)}, Z^{(i)}\big) \leq \frac{1}{2} - \frac{\beta(\alpha)}{36} \right) \stackrel{\eqref{eq:RAP_summand_2_further_assumption}}{=} e^{-\Omega(N)}.
    \end{align}
    Therefore, we conclude that with probability $1 - e^{-\Omega(N)}$,
    \begin{align*}
        T_\ind(x, Z) & \stackrel{\eqref{eq:RAP_decomp_bd_1}}{\geq} \frac{|\mathcal I_x|}{B}\left( \frac{1+\beta(\alpha)}{2} \right) + \frac{1}{B}\sum_{i \notin \mathcal I_x} A_\loc\big(x^{(i)}, Z^{(i)}\big) \\
        & = \frac{|\mathcal I_x|}{B}\left( \frac{1+\beta(\alpha)}{2} \right) + \frac{|\mathcal I_x^c|}{B}\left( \frac{1}{|\mathcal I_x^c|} \sum_{i \notin \mathcal I_x} A_\loc\big(x^{(i)}, Z^{(i)}\big) \right) \\
        & \stackrel{\eqref{eq:RAP_summand_two_lb}}{\geq} \frac{|\mathcal I_x|}{B}\left( \frac{1+\beta(\alpha)}{2} \right) + \frac{|\mathcal I_x^c|}{B}\left( \frac{1}{2} - \frac{\beta(\alpha)}{36} \right) = \frac{1}{2} + \left( \frac{|\mathcal I_x|}{B} \cdot \frac{\beta(\alpha)}{2} - \frac{|\mathcal I_x^c|}{B} \cdot \frac{\beta(\alpha)}{36} \right) \\
        & \stackrel{\eqref{eq:RAP_I_lb}}{\geq} \frac{1}{2} + \left( \frac{\beta(\alpha)}{20} - \frac{\beta(\alpha)}{40} \right) = \frac{1}{2} + \frac{\beta(\alpha)}{40} = \frac{1}{2} + \beta^\star(\alpha).
    \end{align*}
    Altogether, we conclude that
    \begin{align*}
        & \Prob_x\!\left( \mathcal A(x)^c \right) = \Prob\left( Z \notin \mathcal A(x) \right) = \Prob\left( T_\ind(x, Z) < 1/2 + \beta^\star(\alpha) \right) = e^{-\Omega(N)}.
    \end{align*}
    The uniformity of the guarantee over typical $x \in \{0,1\}^{N}$ is now a consequence of the fact that the preceding argument did not depend on the particular choice of typical $x$.
\end{proof}

\section{\texorpdfstring{Proof of \cref{thm:quantitative_capacity_bound}}{Proof of Proposition \ref{thm:quantitative_capacity_bound}}} \label{app:quantitative_cap_bd_proof}

\begin{proof}
    We begin by fixing auxiliary parameters $\epsilon$ and $b$ for which the argument goes through. Throughout this section we take
    \begin{align*}
        \epsilon = 1/24.
    \end{align*}
    We also record the following quantitative Berry-Esseen bound for later use.
    \begin{theorem}[{\cite[Theorem 1]{zolotarev1967sharpening}}] \label{thm:berry_esseen_quantitative}
        Let $\xi_1, \dots, \xi_n$ be independent random variables with mean zero, variances $\sigma_1^2, \dots, \sigma_n^2$, and finite absolute third moments $\beta_1, \dots, \beta_n$. Let $F$ denote the distribution function of $\sum_{k=1}^n \xi_k/\sqrt{\sum_{k=1}^n \sigma_k^2}$. Then
        \begin{align*}
            \sup_x |F(x) - \Phi(x)| \leq \frac{\sum_{k=1}^n \beta_k}{\left( \sum_{k=1}^n \sigma_k^2 \right)^{3/2}}.
        \end{align*}
    \end{theorem}
    We now enumerate each point in the proof where we require $b$ to be sufficiently large, and record an explicit condition on $b$ ensuring that step.
    \begin{itemize}
        \item In \eqref{eq:RAP_mult_chernoff}, we require that
        \begin{align} \label{eq:RAP_mult_chernoff_numerical_cond}
            2\exp\left( -\frac{\alpha b^{2\epsilon}}{3} \right) \leq b^{-\epsilon} \iff \log \left( 2b^{1/24} \right) \leq \frac{\alpha b^{1/12}}{3}.
        \end{align}
        This condition holds for all $b \geq (3/\alpha)^{24}$, on which
        \begin{align*}
            \log \left( 2b^{1/24} \right) \leq b^{1/24} \leq \frac{\alpha}{3}b^{1/12}.
        \end{align*}
    
        \item In \eqref{eq:RAP_additive_chernoff}, we require that
        \begin{align*}
            \exp\left( -\frac{N}{b} \left[ b^{-\epsilon} \left( -\epsilon \log b - \log 2 + \frac{\alpha b^{2\epsilon}}{3} \right) + (1-b^{-\epsilon}) \log\left( \frac{1-b^{-\epsilon}}{1-q_{\text{len}}} \right) \right] \right) \leq \exp\left( - \frac{N}{b} \cdot \frac{\alpha b^\epsilon}{6} \right),
        \end{align*}
        which follows if we have that
        \begin{align*}
            b^{-\epsilon} \left( -\epsilon \log b - \log 2  \right) + (1-b^{-\epsilon})\log(1-b^{-\epsilon}) - (1-b^{-\epsilon}) \log\left(1-q_{\text{len}}\right) \geq -\frac{\alpha b^\epsilon}{6}.
        \end{align*}
        In particular, this holds whenever
        \begin{align} \label{eq:RAP_additive_chernoff_numerical_conds}
            & \log b^{1/24} + \log 2 \leq \frac{\alpha b^{1/12}}{12};
            & (1-b^{-1/24})\log(1-b^{-1/24}) \geq -\frac{\alpha b^{1/24}}{12}.
        \end{align}
        Controlling the two LHS expressions in \eqref{eq:RAP_additive_chernoff_numerical_conds} via
        \begin{align*}
            & \log \left( 2b^{1/24} \right) \leq b^{1/24};
            & (1-x)\log(1-x) \geq -x
        \end{align*}
        gives that both conditions are satisfied whenever $b \geq (12/\alpha)^{24}$.
        
        \item In \eqref{eq:ne_2_size_bd} and \eqref{eq:ne_2_exp_bd}, we require that
        \begin{align*}
            B \cdot h(3b^{-\epsilon}) + B \cdot 3b^{-\epsilon} \log (b+1) \leq B \cdot \frac{\beta^\star(\alpha)^2}{8} \iff h(3b^{-\epsilon}) + 3b^{-\epsilon} \log (b+1) \leq \frac{\beta^\star(\alpha)^2}{8}.
        \end{align*}
        In particular, this holds whenever
        \begin{align} \label{eq:ne_2_bd_numerical_conds}
            -3b^{-\epsilon}\log(3b^{-\epsilon}) \leq \frac{\beta^\star(\alpha)^2}{32}; \qquad -(1-3b^{-\epsilon})\log(1-3b^{-\epsilon}) \leq \frac{\beta^\star(\alpha)^2}{32}; \qquad b^{-\epsilon}\log(b+1) \leq \frac{\beta^\star(\alpha)^2}{48}.
        \end{align}
        Controlling the three LHS expressions in \eqref{eq:ne_2_bd_numerical_conds} via
        \begin{align*}
            -x\log x \leq \sqrt{x}; \qquad -(1-x)\log(1-x)\leq x; \qquad \log(b+1) \leq 48 \cdot b^{1/48}
        \end{align*}
        gives that all conditions in \eqref{eq:ne_2_bd_numerical_conds} are satisfied whenever $b \geq \left( 1920/\beta(\alpha) \right)^{96}$.
        
        \item In our invocation of Hoeffding's inequality in \eqref{eq:ne_2_hoeffding_bd}, we assumed that
        \begin{align*}
            & \exp\left( -\frac{2\left( B\left( \frac{1+\beta^*(\alpha)}{2} - \Prob\left( \maj(X'^{(i)}) = \maj(Y^{(i)}), \Delta(Y^{(i)}) > 0  \right) \right) \right)^2}{B} \right) \leq \exp\left( -\frac{B \cdot \beta^*(\alpha)^2}{4} \right) \\
            & \iff 2\Prob\left( \maj(X'^{(i)}) = \maj(Y^{(i)}), \Delta(Y^{(i)}) > 0 \right) - 1 \leq \left( 1 - \frac{1}{\sqrt{2}} \right)\beta^*(\alpha).
        \end{align*}
        Bounding the probability that a simple random walk hits $0$ after a fixed number of steps implies that this holds whenever
        \begin{align*}
            2 \cdot \Prob\left( \Delta(X'^{(i)}) = 0 \text{ or } \Delta(Y^{(i)}) = 0 \right) \leq \frac{4}{\sqrt{\alpha b}} \leq \left( 1 - \frac{1}{\sqrt{2}} \right)\beta^*(\alpha) \implies \alpha b \geq \left( \frac{4}{\left( 1 - \frac{1}{\sqrt{2}}\right)\beta^*(\alpha)} \right)^2
        \end{align*}
        and this is satisfied whenever
        \begin{align*}
            b \geq \left( \frac{160}{\alpha\left( 1 - \frac{1}{\sqrt{2}}\right)\beta(\alpha)} \right)^2.
        \end{align*}

        \item The inequality \eqref{eq:total_shift} is satisfied whenever $\alpha b^{1/2+2\epsilon} \leq \alpha b$ and $\alpha b^{1/2+2\epsilon} \leq (1-\alpha) b$. Both of these conditions are satisfied whenever
        \begin{align*}
            b \geq \left(1-\alpha\right)^{-12/5}.
        \end{align*}

        \item The inequality of \eqref{eq:Y_equipartition_biased_stretches_bd} used the fact that
        \begin{align*}
            b^{-(4+8\epsilon)} \geq 2e^{-b^{1/2-6\epsilon}/2} \iff \frac{1}{2}b^{1/4} \geq \frac{13}{3}\log b + \log 2.
        \end{align*}
        This condition holds for all $b \geq 30^8$, on which in particular $\log b \leq b^{1/8}$.
    
        \item In \eqref{eq:standardization_algo_bd}, we require that
        \begin{align} \label{eq:biased_stretches_bd}
            2b^{-\epsilon} + \alpha b^{1+\epsilon} \cdot b^{-(1+2\epsilon)} = (2+\alpha)b^{-\epsilon} \leq \frac{\beta^\star(\alpha)}{2} \iff b \geq \left( \frac{80(2+\alpha)} {\beta(\alpha)} \right)^{24}.
        \end{align}
    
        \item In \eqref{eq:RAP_summand_one_lb}, it suffices for $b$ to be large enough so that for any $i \in \mathcal I_x$, it holds that
        \begin{align*}
            \Prob\left( \zeta_1 - \zeta_0 \geq b^{1/2-\epsilon} \right) \geq \frac{1}{2} + \frac{3\beta(\alpha)}{4}.
        \end{align*}
        We may express the probability of interest as
        \begin{align*}
            & \Prob\left( \zeta_0 - \zeta_1 \leq -b^{1/2-\epsilon} \right) = \Prob\left( \frac{\left( \zeta_0-\frac{\alpha(b-\sqrt{b})}{2} \right) - \left( \zeta_1-\frac{\alpha(b+\sqrt{b})}{2} \right)}{\sqrt{\alpha b (1-\alpha)}} \leq \frac{\alpha \sqrt{b} - b^{1/2-\epsilon}}{\sqrt{\alpha b (1-\alpha)}} \right) \\
            & \qquad =: \Prob\left( S_b \leq \frac{\alpha \sqrt{b} - b^{1/2-\epsilon}}{\sqrt{\alpha b (1-\alpha)}} \right) = \Prob\left( S_b \leq \sqrt{\frac{\alpha}{1-\alpha}} - \frac{1}{b^\epsilon \sqrt{\alpha(1-\alpha)}} \right).
        \end{align*}
        It therefore follows that
        \begin{align}
            & \left| \Prob\left( \zeta_1 - \zeta_0 \geq b^{1/2-\epsilon} \right) - \left(\frac{1}{2} + \beta(\alpha) \right) \right| = \left| \Prob\left( \zeta_1 - \zeta_0 \geq b^{1/2-\epsilon} \right) - \Phi\left( \sqrt{\frac{\alpha}{1-\alpha}} \right) \right| \nonumber \\
            & \qquad \leq \sup_x |F_{S_b}(x) - \Phi(x)| + \left| \Phi\left( \sqrt{\frac{\alpha}{1-\alpha}} - \frac{1}{b^\epsilon \sqrt{\alpha(1-\alpha)}} \right) - \Phi\left( \sqrt{\frac{\alpha}{1-\alpha}} \right) \right|. \label{eq:good_set_b_bound}
        \end{align}
        It suffices to show that \eqref{eq:good_set_b_bound} is at most $\beta(\alpha)/4$. In particular, this holds whenever
        \begin{align} \label{eq:good_indices_two_bounds}
            & \sup_x |F_{S_b}(x) - \Phi(x)| \leq \frac{\beta(\alpha)}{8};
            & \left| \Phi\left( \sqrt{\frac{\alpha}{1-\alpha}} - \frac{1}{b^\epsilon \sqrt{\alpha (1-\alpha)}} \right) - \Phi\left( \sqrt{\frac{\alpha}{1-\alpha}} \right) \right| \leq \frac{\beta(\alpha)}{8}.
        \end{align}
        Invoking \cref{thm:berry_esseen_quantitative}, the former condition of \eqref{eq:good_indices_two_bounds} is satisfied whenever
        \begin{align*}
            \frac{b}{\left( \alpha b ( 1 -  \alpha) \right)^{3/2}} \leq \frac{\beta(\alpha)}{8} \iff \frac{64}{\alpha^3\beta(\alpha)^2(1-\alpha)^3} \leq b.
        \end{align*}
        On the other hand, since it readily follows from the Mean Value Theorem that $\Phi$ is $1$-Lipschitz, the latter condition of \eqref{eq:good_indices_two_bounds} is satisfied whenever
        \begin{align*}
            \frac{1}{b^\epsilon \sqrt{\alpha(1-\alpha)}} \leq \frac{\beta(\alpha)}{8} \iff \frac{8^{24}}{\alpha^{12}\beta(\alpha)^{24} (1-\alpha)^{12}} \leq b.
        \end{align*}
    
        \item In carrying out \eqref{eq:RAP_summand_two_lb}, we assume that $b$ is large enough so that for any $i \notin \mathcal I_x$, it holds that
        \begin{align*}
            \Prob\left( \Tilde{\zeta}_1 - \Tilde{\zeta}_0 \geq b^{1/2-\epsilon} \right) \geq \frac{1}{2} - \frac{\beta(\alpha)}{72}
        \end{align*}
        where we have here that independently
        \begin{align*}
            & \Tilde{\zeta}_1, \Tilde{\zeta}_0 \sim \Bin\left( b/2 , \alpha \right).
        \end{align*}
        We may express the probability of interest as
        \begin{align*}
            \Prob\left( \Tilde{\zeta}_0 - \Tilde{\zeta}_1 \leq -b^{1/2-\epsilon} \right) & = \Prob\left( \frac{\left( \Tilde{\zeta}_0-\frac{\alpha b}{2} \right) - \left( \Tilde{\zeta}_1-\frac{\alpha b}{2} \right)}{\sqrt{\alpha b \cdot (1-\alpha)}} \leq \frac{-b^{1/2-\epsilon}}{\sqrt{\alpha b \cdot (1-\alpha)}} \right) \\
            & =: \Prob\left( \Tilde{S}_b \leq -\frac{1}{b^\epsilon \sqrt{\alpha(1-\alpha)}} \right).
        \end{align*}
        It therefore follows that
        \begin{align} \label{eq:bad_set_b_bound}
            & \left| \Prob\left( \Tilde{\zeta}_1 - \Tilde{\zeta}_0 \geq b^{1/2-\epsilon} \right) - \frac{1}{2} \right| \leq \sup_x |F_{\Tilde{S}_b}(x) - \Phi(x)| + \left| \Phi\left( -\frac{1}{b^\epsilon \sqrt{\alpha(1-\alpha)}} \right) - \frac{1}{2} \right|.
        \end{align}
        It suffices to show that \eqref{eq:bad_set_b_bound} is at most $\beta(\alpha)/72$. In particular, this holds whenever
        \begin{align} \label{eq:bad_indices_two_bounds}
            & \sup_x |F_{\Tilde{S}_b}(x) - \Phi(x)| \leq \frac{\beta(\alpha)}{144};
            & \left| \Phi\left( -\frac{1}{b^\epsilon \sqrt{\alpha(1-\alpha)}} \right) - \frac{1}{2} \right| \leq \frac{\beta(\alpha)}{144}.
        \end{align}
        Invoking \cref{thm:berry_esseen_quantitative}, the former condition of \eqref{eq:bad_indices_two_bounds} is satisfied whenever
        \begin{align*}
            \frac{b}{\left( \alpha b \left( 1 - \alpha \right) \right)^{3/2}} \leq \frac{\beta(\alpha)}{144} \iff \frac{20736}{\alpha^3\beta(\alpha)^2(1-\alpha)^3} \leq b.
        \end{align*}
        On the other hand, as $\Phi$ is $1$-Lipschitz, the latter condition of \eqref{eq:bad_indices_two_bounds} is satisfied whenever
        \begin{align*}
            \frac{1}{b^\epsilon \sqrt{\alpha(1-\alpha)}} \leq \frac{\beta(\alpha)}{144} \iff \frac{144^{24}}{\alpha^{12}\beta(\alpha)^{24} (1-\alpha)^{12}} \leq b.
        \end{align*}
    \end{itemize}
    Finally, we record a crude quantitative estimate for later use. We would like $b$ to be large enough so that
    \begin{align} \label{eq:planted_gap_b_asymp}
        \frac{\alpha \cdot \KL\left( b^{-(4+8\epsilon)} \doubleline 2e^{-b^{1/2-6\epsilon}/2} \right)}{b^{1/2+2\epsilon}} & \geq \frac{\alpha \cdot b^{-(4+8\epsilon)}\left( \frac{1}{2}b^{1/2-6\epsilon} -(4+8\epsilon)\log b - \log 2 \right)}{b^{1/2+2\epsilon}} \geq \frac{1}{b^5},
    \end{align}
    where the first inequality follows from the definition of the KL divergence of two Bernoulli distributions. The latter inequality simplifies to
    \begin{align*}
        \frac{1}{2}b^{1/3} - b^{1/12}\left( \frac{13}{3}\log b + \log 2 \right) \geq \frac{1}{\alpha}.
    \end{align*}
    This condition holds whenever $b \geq (120/\alpha)^{24}$, on which in particular $\log b \leq b^{1/24}$.

    \smallskip
    
    Combining all of these conditions on $b$ yields that the proof follows whenever
    \begin{align*}
        b \geq \left\lceil \frac{1920^{96}}{\alpha^{24}\beta(\alpha)^{96} (1-\alpha)^{12}} \right\rceil = \kappa(\alpha).
    \end{align*}
    We take $b = \kappa(\alpha)$. We now make the implicit constants in the proof of \eqref{eq:null_quantitative} explicit and then replace them by their minimum. Retrieving explicit expressions for \eqref{eq:RAP_summand_one_lb} and \eqref{eq:RAP_summand_two_lb}, the guarantee of \cref{prop:regular_alignment_property}, whose corresponding constant is that of the initial $\Omega(1)$ term of \eqref{eq:null_quantitative}, can be explicitly written via
    \begin{align*}
        \exp\left( -\frac{N \cdot  \beta(\alpha)^3(1+o(1))}{5184(1+\beta(\alpha))\kappa(\alpha)} \right) \leq \exp\left( -N \left[ \frac{\beta(\alpha)^3}{8000 \cdot \kappa(\alpha)}\right] \right)
    \end{align*}
    for all large $N$. We next consider the $\Omega(N)$ term of \eqref{eq:null_quantitative}. The guarantee of \cref{prop:total_alignment_2_bd} here is
    \begin{align*}
        \exp\left( - B \cdot \frac{\beta^\star(\alpha)^2}{8} \right) = \exp\left( -N \left[ \frac{\beta(\alpha)^2}{12800 \cdot \kappa(\alpha)} \right] \right),
    \end{align*}
    while the guarantee of \cref{prop:few_biased_stretches} is 
    \begin{align*}
        \exp\left( -\frac{M}{b^{1/2+2\epsilon}} \cdot \KL\left( b^{-(4+8\epsilon)} \doubleline 2e^{-b^{1/2-6\epsilon}/2} \right) \right) \stackrel{\eqref{eq:planted_gap_b_asymp}}{\leq} \exp\left( -N \left[ \frac{1}{\kappa(\alpha)^5} \right] \right).
    \end{align*}
    Altogether, we have that
    \begin{align} \label{eq:C_unif_lower_bd}
        \fpl(\alpha) - \fnullann(\alpha) \geq \frac{1}{4} \cdot \frac{\beta(\alpha)^3}{12800 \cdot \kappa(\alpha)^5} \stackrel{\eqref{eq:C_unif_positivity}}{\implies} C_{\unif}(1-\alpha) \geq \frac{\beta(\alpha)^3}{51200 \cdot \kappa(\alpha)^5}.
    \end{align}
    Finally, we fix $p \in (0, 1)$. Setting $\alpha = 1-p$, we conclude that
    \begin{align*}
        C_\unif(p) = C_\unif\left( 1 - \alpha \right) \stackrel{\eqref{eq:C_unif_lower_bd}}{\geq} \frac{\left(\beta(1-p)\right)^3}{51200 \cdot \left(\kappa(1-p)\right)^5} > 0,
    \end{align*}
    yielding the desired result.
\end{proof}

\section{\texorpdfstring{Addendum to the Proof of \cref{lem:inner-optimizer}}{Addendum to the Proof of Lemma \ref{lem:inner-optimizer}}} \label{app:lagrange-multipliers-addendum}

We list some results on exponential families here that we invoke later.

\begin{theorem}[{\cite[Equation (3.28)]{wainwright2008graphical}}] \label{thm:exp-fam-mean-representation}
    In an exponential family with discrete state space $\mathcal X$ and sufficient statistic $\phi(X)$, the mean parameter space $\mathcal M$ admits the representation
    \begin{align*}
        \mathcal M = \textup{conv}\left(\{ \phi(x) : x \in \mathcal X \} \right),
    \end{align*}
    where $\textup{conv}$ denotes the convex hull operation.
\end{theorem}

\begin{theorem}[{\cite[Theorem 3.3]{wainwright2008graphical}}] \label{thm:exp-fam-mean-realization}
    In a minimal exponential family with parameter space $\Omega$, sufficient statistic $\phi(X)$, and log-partition function $A$, the gradient map $\nabla_\theta A$ is onto $\mathcal M^\circ$, the interior of the mean parameter space. Consequently, for each $\mu \in \mathcal M^\circ$, there exists some $\theta \in \theta(\mu) \in \Omega$ such that $\E_\theta[\phi(X)] = \mu$.
\end{theorem}

For $\theta = (\theta_1, \theta_2)$, defining (with the correspondence $\theta_1 = -\lambda_1$ and $\theta_2 = -\lambda_2$)
\begin{align*}
    A(\theta) := \log\left( \sum_{a,b} w(a,b) e^{\theta_1 a + \theta_2 b} \right)
\end{align*}
gives the log-partition function of the exponential family supported on $\{ (a,b)\in\N_{>0}^2:1\le b\le a\}$ with mass function $p_\theta(a,b)$ defined via
\begin{align*}
    p_\theta(a,b) \propto w(a,b) e^{\theta_1 a + \theta_2 b} \implies p_\theta(a,b) = w(a,b) e^{\left\langle \theta, (a,b) \right\rangle - A(\theta)}.
\end{align*}
As the sufficient statistic of this exponential family is $(a,b)$, which is not contained in an affine line, the family is thus minimal. Additionally, since $p_\theta$ has discrete support, it follows from \cref{thm:exp-fam-mean-representation} that the mean parameter space is
\begin{align*}
    \mathcal M := \left\{ (\alpha, \beta) \in \R^2: 1 \leq \beta \leq \alpha \right\} \implies \mathcal M^\circ = \left\{ (\alpha, \beta) \in \R^2: 1 < \beta < \alpha \right\}.
\end{align*}
It follows from \cref{lem:partition-closed-form} (whose proof does not rely on the preceding \cref{lem:inner-optimizer}) and \eqref{eq:closed-form-partition-fn-discriminant} that
\begin{align*}
    \Omega := \left\{ (\theta_1, \theta_2) \in \R^2 : A(\theta) < \infty \right\} = \left\{ (\theta_1, \theta_2) \in (-\infty, 0)^2 : e^{\theta_1}(1+e^{\theta_2/2})^2 < 1 \right\}.
\end{align*}
It readily follows from this description of $\Omega$ that for any $\gamma = (\gamma_1, \gamma_2) \in \Omega$, there exists an open box $I \subset \Omega$ such that $\gamma \in I$ and uniformly over $\theta \in I$, 
\begin{align*}
    & \sum_{a,b} \frac{\partial p_\theta(a,b)}{\partial \theta_1} = \sum_{a,b} a \cdot p_\gamma(a,b) < \infty;
    & \sum_{a,b} \frac{\partial p_\theta(a,b)}{\partial \theta_2} = \sum_{a,b} b \cdot p_\gamma(a,b) \leq \sum_{a,b} a \cdot p_\gamma(a,b) < \infty.
\end{align*}
It thus follows that this family satisfies, for all $\theta \in \Omega$, that
\begin{align} \label{eq:exp_fam_opt_existence}
    \nabla_\theta A(\theta) = \left( \E_\theta[a], \E_\theta[b] \right),
\end{align}
which we want to be equal to $(1/(\alpha\rho), 1/\rho) \in \mathcal M^\circ$. Altogether, \cref{thm:exp-fam-mean-realization} yields the existence of $(\theta_1^*, \theta_2^*)$ for which \eqref{eq:exp_fam_opt_existence} is satisfied. Letting $\left( \lambda_1^*, \lambda_2^* \right) = (-\theta_1^*, -\theta_2^*)$, we may then set $\lambda_0^*$ via 
\begin{align*}
    e^{1+\lambda_0^*} = \sum_{a,b} w(a,b) e^{-\lambda_1^* a - \lambda_2^* b}.
\end{align*}
We conclude that this system of equations has a solution. We now let $\nu^*$, $Z^*$, $\lambda_1^*$, and $\lambda_2^*$ denote the corresponding quantities for such a solution. It holds that
\begin{align} \label{eq:optimum_prob_expression}
    \log \nu^*(a,b) = \log w(a,b) - \lambda_1^* a -\lambda_2^* b - \log Z^*.
\end{align}
For any probability measure $\nu$, the objective functional may be expressed via
\begin{align}
    & H(\nu) + \E_\nu\left[ \log w \right] = -\sum_{a,b}\nu(a,b) \log\nu(a,b) + \sum_{a,b} \nu(a,b) \log w(a,b) \nonumber \\
    & \qquad = -\sum_{a,b} \nu(a,b) \log \nu(a,b) + \sum_{a,b} \nu(a,b) \log \nu^* + \sum_{a,b} \nu \lambda_1^*a + \sum \nu(a,b) \lambda_2^* b + \sum_{a,b} \nu(a,b) \log Z^* \nonumber \\
    & \qquad = -\sum_{a,b} \nu(a,b) \log\frac{\nu(a,b)}{\nu^*(a,b)} + \frac{c\lambda_1}{\rho} + \frac{\lambda_2}{\rho} + \log Z^* = \frac{c\lambda_1}{\rho} + \frac{\lambda_2}{\rho} + \log Z^* - \KL\left( \nu \doubleline \nu^* \right) \nonumber \\
    & \qquad = H(\nu^*) + \E_{\nu^*}\left[ \log w \right] - \KL\left( \nu \doubleline \nu^* \right), \label{eq:objective-functional-opt-relation}
\end{align}
where the first equality in the third line is due to the normalization and mean constraints applied to the solution. Since $\KL\left( \nu \doubleline \nu^* \right) \geq 0$, we conclude that $\nu^*$ is a global maximizer of the objective functional.
\end{appendices}

\section*{References}

\printbibliography[heading=none]

\end{document}